\documentclass[submission,copyright,creativecommons]{eptcs}
\providecommand{\event}{} 

\usepackage[T1]{fontenc} 

\usepackage{amssymb} 

\usepackage{calculi}
\usepackage{amsthm}
\usepackage{thmtools}
\usepackage{listings}

\declaretheorem[style=plain,name=Theorem]{theorem}

\declaretheorem[style=plain,name=Corollary]{corollary}
\declaretheorem[style=plain,name=Lemma]{lemma}
\declaretheorem[style=definition,qed=$\blacksquare$,name=Definition]{definition}

\definecolor{strings}{rgb}{0,0.5,0}
\definecolor{emphs}{rgb}{0.64,0.08,0.08}
\definecolor{comments}{rgb}{0.17,0.57,0.68}
\colorlet{keywords}{blue!50!cyan}

\lstdefinestyle{tinysol}{
  language=C,
  captionpos=b,
  numbers=left,
  numberstyle=\tiny,
  frame=lines,
  showspaces=false,
  showtabs=false,
  breaklines=true,
  showstringspaces=false,
  breakatwhitespace=true,
  emph={malloc, calloc, realloc, free, memset},
  emphstyle={\rmfamily\bfseries\color{emphs}},
  commentstyle=\color{comments},
  morekeywords={contract, account, this, var, prop, value, balance, sender},
  keywordstyle={\bfseries\color{keywords}},
  stringstyle=\color{strings},
  basicstyle=\ttfamily\small,
  escapechar=@
}

\newcommand{\ghostsubsection}[1]{%
\phantomsection%
\addcontentsline{toc}{subsection}{#1}%
}

\def\DEFSYM{\ensuremath{\overset{\Delta}{=}}}
\newcommand{\code}[1]{\texttt{#1}}

\def\WCLANG{\texttt{WC}}
\newcommand{\CONF}[1]{\ensuremath{\left<#1\right>}}
\newcommand{\UPDATE}[2]{\ensuremath{\kern-2pt\left[#1 \mapsto #2\right]}}
\newcommand{\SUBSTITUTE}[2]{\ensuremath{\left\{#1/#2\right\}}}
\newcommand{\WCCALL}[3]{\ensuremath{\code{call $#1$.$#2$($#3$)}}}

\def\VALUES{\ensuremath{\normalfont\text{\sffamily Val}}}
\def\EXPR{\ensuremath{\normalfont\text{\sffamily Exp}}}
\def\PROC{\ensuremath{\normalfont\text{\sffamily Proc}}}
\def\STM{\ensuremath{\normalfont\text{\sffamily Stm}}}

\def\MNAMES{\ensuremath{\normalfont\text{\sffamily MNames}}} 
\def\VNAMES{\ensuremath{\normalfont\text{\sffamily VNames}}} 
\def\FNAMES{\ensuremath{\normalfont\text{\sffamily FNames}}} 
\def\CNAMES{\ensuremath{\normalfont\text{\sffamily CNames}}} 

\newcommand{\ENV}[1]{\ensuremath{\normalfont\text{\sffamily env}_{#1}}}

\newcommand{\SETENV}[1]{\ensuremath{\normalfont\text{\sffamily Env}_{#1}}}

\def\TRUE{\ensuremath{\text{\sffamily T}}}
\def\FALSE{\ensuremath{\text{\sffamily F}}}

\newcommand{\IFTHENELSE}[3]{\ensuremath{\textbf{if~} #1 \textbf{~then~} #2 \textbf{~else~} #3}}
\newcommand{\LOCNEW}[2]{\NEW{#1 \mathrel{\code{:=}} #2}}
\newcommand{\LOCREAD}[2]{#1 \triangleright \code{($#2$)}}
\newcommand{\LOCWRITE}[2]{#1 \triangleleft \code{<$#2$>}}

\newcommand{\TLOC}[2][i \in I]{\ensuremath{\text{\sffamily loc}\SET{#2}}_{#1}}
\def\TRET{\ensuremath{I^\text{\sffamily ret}}}
\newcommand{\TIFC}[1]{\ensuremath{I^{#1}}}

\def\titlerunning{Typing Composite Subjects}
\def\authorrunning{L. Aceto et al.}
\title{\titlerunning}

\newlength\QEDSymbolSpace
\newenvironment{fixqed}
{\begin{minipage}[b]{\dimexpr\linewidth-\QEDSymbolSpace}\centering}
{\end{minipage}\vspace{-1em}}

\author{%
Luca Aceto\thanks{The work of Luca Aceto and Stian Lybech was supported by the Icelandic Research Fund Grant No.\@ 218202-05(1-3).}
\institute{%
  Dept.~of Computer Science \\ 
  Reykjavik University      \\ 
  Reykjavik (Iceland)
}
\institute{%
  Gran Sasso Science Institute\\
  L'Aquila (Italy)
}
\email{luca@ru.is}
\and
Daniele Gorla
\institute{%
  Dept.~of Computer Science           \\ 
  Sapienza University of Rome (Italy) \\ 
}
\email{gorla@di.uniroma1.it}
\and
Stian Lybech\textsuperscript{\thefootnote}
\institute{%
  Dept.~of Computer Science \\
  Reykjavik University      \\ 
  Reykjavik (Iceland)
}
\email{stian21@ru.is}
}

\begin{document}
\maketitle

\section{Introduction}\label{wc:sec:introduction}
Process calculi are formalisms for modelling and reasoning about concurrent and distributed programs, and one way of reasoning about the safety of such programs is by means of \emph{type systems}.
In that approach, types are usually assigned to atomic identifiers, such as variables, function names and, in the case of process calculi with communication primitives, to \emph{channel names}.
For example, the original \PI-calculus \cite{PICALC,milner_walker_parrow1992picalc} features communication primitives for transmitting single, atomic names on named channels: $\OUTPUT{x}{z}.P$ outputs the name $z$ on the channel $x$ and continues as $P$, whilst $\INPUT{x}{y}.P$ receives a name on the channel $x$ and binds it to $y$ within the continuation $P$.
In standard \PI-calculus terminology, $x$ is referred to as the \emph{subject}, and $y$ (resp.\@ $z$) as the \emph{object} of the input (resp.\@  output) construct.
This can be compared to a procedure $x$ in an imperative-style language, taking a single parameter, with $y$ and $z$ as the formal and actual parameters, respectively.

Further extensions of the \PI-calculus also allow polyadicity for objects, such that multiple values can be transmitted in a single communication, written $\INPUT{x}{\VEC{y}}.P$ and $\OUTPUT{x}{\VEC{z}}.P$ for input and output, respectively, where $\VEC{y}$ and $\VEC{z}$ represent vectors of names.
This is known as the polyadic \PI-calculus \cite{sangiorgi2003pi}, and although it is well-known that polyadicity can be encoded in the monadic \PI-calculus, this extension makes the correspondence with imperative procedure calls even clearer, since procedures usually are able to take multiple arguments.
This correspondence between synchronisation in the \PI-calculus, and classic imperative features such as procedures, references and variables is well-known and has already been explored in a number of papers, e.g.\@ \cite{hirschkoff2020references_picalc,KS02,NR99,RV96,Sangiorgi98,Vasco95,Walker95}.

However, not all process calculi use only \emph{atomic} identifiers as channels.
For example, the language \EPI{} \cite{CARBONEMAFFEIS} extends the \PI-calculus with polyadic \emph{synchronisation}.
Thus, in \EPI, subjects can be \emph{vectors} of names of arbitrary length, e.g.\@ $\OUTPUT{x_1 \cdot x_2}{z}.P$ or $\INPUT{x_1 \cdot x_2 \cdot x_3}{y}.P$; and since names can be passed around, such name vectors can also be composed at runtime, although they cannot grow in length.
Every \PI-calculus process is then also an \EPI-process, with the length of subjects restricted to a single name; hence, the \PI-calculus is also referred to as $\PI^1$, and every fixed length $n$ of subjects then yields a language $\PI^n$, with \EPI{} being the union of all these languages.

Polyadic synchronisation and runtime composition of channels provide a more realistic way of expressing communication in a distributed setting, compared to the standard \PI-calculus.
Consider for example an IP address in CIDR notation: it is not an \emph{atomic} identifier, but rather a \emph{composition} of elements, which also may be built at runtime.
Unlike the case for polyadic objects, such name compositions cannot be expressed in the standard \PI-calculus without introducing divergence, as was shown by Carbone and Maffeis in \cite{CARBONEMAFFEIS}.
Admitting polyadic subjects into the language thus yields increased expressivity.

The aforementioned correspondence with imperative features then poses a new question: namely, how to view polyadic synchronisation in an imperative setting.
Our proposal in the present paper is that subject vectors are comparable to \emph{name-spaced} procedures and variables, such as in class-based and object-oriented programming, where methods and fields are grouped under a shared class name.
A natural way to represent a method call
\begin{center}
  \WCCALL{X}{f}{v_1, \ldots, v_n}
\end{center}

\noindent to a method $f$, in a class or object $X$, with actual parameters $v_1, \ldots, v_n$, would therefore simply be as an output
\begin{equation*}
  \OUTPUT{X \cdot f}{v_1, \ldots, v_n}
\end{equation*}

\noindent where we assume $X, f$ are names.\footnote{This representation technically only makes use of the $\PI^2$ sublanguage of $\EPI$, but one could easily imagine a language with e.g.\@ nested classes, which would require arbitrarily long subject vectors to represent an arbitrary degree of nesting.}
Likewise, an input 
\begin{equation*}
  \INPUT{X \cdot f}{\VEC{y}} . P
\end{equation*}

\noindent would correspond to a \emph{declaration} of a method $f$ with body $P$ in a class or object named $X$.
We shall explore this correspondence in detail in the following.

This encoding, in itself, is quite trivial as indicated above.
However, we can then use the intuitions afforded by this encoding to guide the development of a type system for \EPI{}, which again should be such that it corresponds to an ``expectable'' type system for a simple, class-based or object-oriented, imperative language.
This is non-trivial, since types are normally given to individual names, whereas subjects in \EPI{} are composites.
Hence, the type of a subject $x_1 \cdot \ldots \cdot x_n$ must somehow be derived from the types of the individual names $x_1, \ldots, x_n$.
This can be done in several possible ways, each with some benefits and drawbacks, depending on what subject vectors are intended to represent.

\section{On typing polyadic subjects}\label{wc:sec:related_works}
Structured channel names pose some interesting problems for the development of type systems, because types are normally given to \emph{atomic} identifiers.
This is explored in some detail in \cite{lybech/2024/iandc/rhocalc}, which describes the \RHO-calculus: a language with structured channels, but without any atomic identifiers.
The type system therein is consequently very limited, and in practice only suitable for encodings of languages that \emph{do} feature atomic channel names, such as the \PI-calculus.

In \EPI, the situation is somewhat better, since even though channels are composable, they are at least composed of atomic names.
However, composing channels at runtime still means that we cannot assign a \emph{single} type to each channel, as is otherwise standard in type systems for the \PI-calculus (e.g.\@ \cite{PICALC, PICAPABLE,sangiorgi2003pi,turner1996phd}).

A third language, that allows polyadic subjects is the \PSI-calculus \cite{bengtson2009psi, bengtson2011psi}.
Briefly, the \PSI-calculus is a generalisation of many of the variants and extensions of the \PI-calculus.
It is a generic framework based on nominal sets \cite{gabbay2002nominal}, and the idea is that one defines three nominal sets (called \emph{terms}, \emph{conditions} and \emph{assertions}) and three operations on them (called \emph{channel equivalence}, \emph{assertion composition} and \emph{entailment}), which are supplied to the \PSI-calculus framework to yield an \emph{instance} of the \PSI-calculus.
In \cite{bengtson2009psi, bengtson2011psi}, Bengtson et al.\@ then show how the \PI-calculus and several of its variants, including \EPI, can be obtained as instances, by choosing different settings for these parameters.
For example, both subjects and objects of communication are drawn from the set of terms, so an instance with both polyadic subjects and objects can be obtained by defining terms as lists of names.
In \cite{parrow2014higher}, Parrow et al.\@ further extend the \PSI-calculus with a construct for higher-order communication, and in \cite{huttel/2024/iandc/hopsitypes} it is shown that this variant of the \PSI-calculus can even encode the \RHO-calculus.

A key point of the \PSI-calculus approach is that certain results (e.g.\@ on bisimulation equivalences) can be formulated and proved for the abstract \PSI-calculus, and they are then automatically inherited by every instance.
This idea was then taken up by Hüttel in a series of papers \cite{huttel2011typed, huttel2013resourcespsi, huttel2016sessionpsi}, wherein he created generic type systems (of different typing disciplines) for the \PSI-calculus and proved a general subject-reduction result.
The generic type systems take a number of parameters, including the parameters for a \PSI-calculus instance, and yield type systems for that instance which then inherit the subject-reduction property, thereby reducing the amount of work needed in order to show soundness of the type system.

However, the aforementioned approach also has some disadvantages.
Being a \emph{static} analysis technique, type systems are closely tied to the \emph{syntax} of the language, but in the \PSI-calculus key elements of the syntax are left unspecified (namely, terms, conditions and assertions).
This, in turn, means that the generic type systems must require type rules for these syntactic categories to be provided as parameters for the instantiation. 
The aforementioned papers do give some examples of instances, including type rules for these syntactic categories; however, all except one of these instances are for variants of the \PI-calculus which only have \emph{single} names as subjects.
Thus, these papers do not actually provide much insight on the question of how to type composite subjects: they merely push the question of how to assign a single type to a composite subject onto the instance parameters.

The sole exception is an instance described in \cite{huttel2011typed}, which is a type system for the (\PSI-calculus instance of) D\PI{} \cite{DPICALC}.
In this language, a process $P$ executes at a specific \emph{location} $l$, written $l[P]$, and can then migrate between locations with the construct $\code{go}~l'.P'$.
The located processes are organised in a single-level network, i.e.\@ without nested networks occurring inside locations.
As was shown in \cite{CARBONEMAFFEIS}, this language can actually be encoded in \EPI, specifically in the sub-language $\PI^2$ (where subject vectors all have length 2).
The relevant clauses of the encoding are as follows:
\begin{center}
\begin{math}
\begin{array}{r @{~} l}
  \PTRANS[{l[P]}]              & = \PTRANS[P](l) \\
  \PTRANS[\OUTPUT{x}{z}.P](l)  & = \OUTPUT{l \cdot x}{z}.\PTRANS[P](l) \\
\end{array}\qquad
\begin{array}{r @{~} l}
  \PTRANS[\code{go}~k.P](l)    & = \PTRANS[P](k) \\
  \PTRANS[\INPUT{x}{y}.P](l)   & = \INPUT{l \cdot x}{y}.\PTRANS[P](l)
\end{array}
\end{math}
\end{center}

The type language (another parameter) given for the type system in \cite{huttel2011typed} is then the following:
\begin{equation*}
  T \DCLSYM \TCHAN{T} \ORSYM \TLOC{x_i : \TCHAN{T_i}} \ORSYM B
\end{equation*}

\noindent where $I$ is a finite index set, $x$ is a channel name and $B$ denotes base types (integers, booleans etc.).
The type rule for assigning a single type to a subject vector is then given as
\begin{center}
\begin{semantics}
  \RULE[t-loc]
  { E \vdash l : \TLOC{x_i : \TCHAN{T_i} } \AND \exists j \in I  \SUCHTHAT E \vdash x_j : \TCHAN{T_j}}
  { E \vdash l \cdot x_j : \TCHAN{T_j} }
\end{semantics}
\end{center}

\noindent where $E$ is a type environment.
As can be seen, the choice here is to repeat the type for $x_j$, so it must occur both on its own and inside the type of the location.
In effect, this means that there is no difference w.r.t.\@ channel capabilities between using the composite subject $l \cdot x_j$ and just using the name $x_j$.
Furthermore, as channel names $x_i$'s appear directly inside the location types $\TLOC{x_i : \TCHAN{T_i}}$, this setting also encounters problems with \ALPHA-conversion, in case some of the $x_i$'s are restricted.
In general, these choices may suffice for this \emph{specific} instance, since it is limited to \emph{encoded} D\PI-processes, with the encoding ensuring that no processes would lead to problems with the chosen type language and type rules.
However, as a \emph{general} type system for $\PI^2$, this choice does not seem sufficient, nor does it offer much insight into how it might be generalised to the full \EPI{} language, where subject vectors may have an arbitrary length.

Another approach is described in \cite[Chapt. 6.5]{carbone2005phd}, where Carbone creates two type systems for \EPI{}: a `structural' one and a `nominal' one.
Of these, to our mind, only the latter provides a satisfactory solution to the problem of deriving a composite type $T$ for a subject vector such as $x_1 \cdot x_2$ from the types of the constituents $x_1$ and $x_2$.
This is done by using \emph{named types} with type names $I$, and two separate type environments, $\Delta$ and $\Gamma$.
The first maps single names to type names, e.g.\@ $\Delta(x_1) = I_1$ and $\Delta(x_2) = I_2$, and the second maps composite type names to types, e.g.\@ $\Gamma(I_1 \cdot I_2) = T$. 
This avoids the problem with \ALPHA-conversion, in case a name is restricted, since names do not appear directly in the types.

The downside of the aforementioned approach is that this way of structuring the types bears little semblance to familiar concepts from object-oriented languages, which is what we are aiming for in the present work.
As argued in the preceding section, the correspondence between methods grouped under a shared class name, and inputs with a common prefix in their subjects, is straightforward.
In OO-terminology, the type of a class or object is its interface, consisting of the signatures of the fields and methods it declares, but the correspondence between the type-name vectors and such interfaces is less obvious.

In the present paper we shall focus on developing a simple type system for \EPI{} for ensuring correct channel usage, whilst ensuring that the intuition from the correspondence with the imperative, object-oriented paradigm remains clear.
Our proposal stems from the observation that the capability of a name to be used as a channel and its capability to be used in compositions are orthogonal, and, furthermore, that the order in a composition matters; i.e.\@ $x_1\cdot x_2$ is different from $x_2\cdot x_1$, and different compositions may have distinct behaviours. 
For example, $x_1$ alone can deliver integers, $x_1\cdot x_2$ can deliver pairs of integers and $x_1\cdot x_3$ can deliver booleans. 
We use a tree-structure to represent both capabilities in our types as follows:
Each node in such a `tree type' is labelled with a type name $I$ and a communication capability $C$, describing whether the vector can be used as a channel or not (and, in that case, what it can communicate).
The root node describes the type of a single name (e.g.\@ $x_1$); each sub-tree below the root describes the type of a vector composed with $x_1$, and so also which compositions are allowed.
Thus, as \EPI{} allows for name vectors of arbitrary lengths, types for names are trees of arbitrary height.

The text is structured as follows:
We start in Section~\ref{wc:sec:epi_language} by presenting the syntax and labelled operational semantics of \EPI{}.
We then move to the focus of this paper and discuss in Section~\ref{wc:sec:epi_type_system} the question of how to type subject vectors.  
The type system we propose satisfies the usual properties of subject reduction and safety, which alone assure us of the quality of our proposal. 
In fact, it is equivalent to a type system proposed by Carbone in \cite{carbone2005phd}, although his types have a different structure, as discussed above.
However, our proposal of `tree-shaped' types also resembles the concept of interfaces, known from type systems for object-oriented languages.

We spell out this correspondence in detail, following the approach of \cite{KS02,NR99,RV96,Sangiorgi98,Vasco95,Walker95}. 
In Section~\ref{wc:sec:wc_language}, we define the simple class-based language \WCLANG, (\textbf{W}hile with \textbf{C}lasses), an imperative language obtained by extending \code{While} \cite{nielson_nielson2007semantics_with_applications} with classes, fields and methods, together with an ``expectable'' type system for \WCLANG{} programs.
Then, in Section~\ref{wc:sec:encoding}, we propose an encoding of \WCLANG{} into \EPI{} and show that the type system presented here for \EPI{} exactly corresponds to the ``expectable'' one for \WCLANG. 
This comparison contributes to understanding the relationship between our types and conventional types of OO languages and provides \EPI{} with a type system that makes it a good metalanguage for the semantics of typed OO languages.

\section{The \EPI-calculus}\label{wc:sec:epi_language}
The version of \EPI{} presented in \cite{CARBONEMAFFEIS} extends the \emph{monadic} \PI-calculus only with polyadic synchronisation;
a natural further extension is to also allow polyadicity for \emph{objects}, as in the polyadic \PI-calculus \cite{PICALC}.
Furthermore, we extend \EPI{} with: data values $v$, consisting of names, integers and booleans; integer and boolean expressions $\VEC{e}$ in outputs $\OUTPUT{\VEC{x}}{\VEC{e}}$; and a guarded choice operator 
$\sum_{i=1}^n[e_i]P_i$, where the $e_i$ are boolean expressions. 
Hence, the syntax of \EPI\ processes that we consider in this paper is:
\begin{syntax}[h]
  P    \in \PROC   \IS \NIL 
                   \OR \INPUT{\VEC{x}}{\VEC{y}}.P 
                   \OR \OUTPUT{\VEC{x}}{\VEC{e}}.P
                   \OR P_1 \PAR P_2                
                   \OR \NEW{\VEC{x}}P              
                   \OR \REPL{P}                    
                   \OR \sum_{i=1}^n[e_i]P_i\quad (n \geq 1)    \tabularnewline 
  e    \in \EXPR   \IS \op(\VEC{e})
                   \OR v                           \tabularnewline
  v    \in \VALUES \IS \NAMES \UNION \mathbb{Z} \UNION \mathbb{B}
\end{syntax} 

\noindent where: $\NAMES$ is the set of names, ranged over by $x, y$; $\mathbb{Z}$ is the set of integers; $\mathbb{B} = \SET{\TRUE, \FALSE}$ is the set of boolean values, and $\op$ denotes a collection of standard operators on these values.
We use a tilde $\VEC{\cdot}$ to denote an ordered sequence of elements, including for example name vectors (that can occur both as subjects and as objects of prefixes).

Most constructs are as in the \PI-calculus:
$\NIL$ is the stopped process (we usually omit trailing $\NIL$ after prefixes); $P_1 \PAR P_2$ is the parallel composition of processes $P_1$ and $P_2$; $\NEW{\VEC{x}}P$ creates new names $\VEC{x}$ with visibility limited to $P$; and the replication operator $\REPL{P}$ creates arbitrarily many copies of $P$.
The input and output operators will synchronise on the \emph{vector} of names $\VEC{x}$, and we likewise allow the data transmitted to be a list of values $\VEC{v}$.
In the output operator $\OUTPUT{\VEC{x}}{\VEC{e}}$, we allow expressions to appear in object position of the output, which must evaluate to a list of values $\VEC{v}$ to be transmitted.
The guarded choice operator $\sum_{i=1}^n[e_i]P_i$, which will be often abbreviated to $\sum [\VEC{e}]\VEC{P}$, allows us to choose between the processes $P_i$ for which the associated guards $e_i$ evaluate to the boolean value $\TRUE$ (true).
Henceforth, we shall write $\IFTHENELSE{e}{P_\TRUE}{P_\FALSE}$ in place of 
$[e]P_\TRUE + [\neg e]P_\FALSE$.

For the operators $\op$ appearing in expressions $e$, we shall assume that they, and their arguments, can be evaluated to a single value by some semantics $\trans_e$, which we shall not detail further.
We shall write $\VEC{e} \trans_e \VEC{v}$ as an abbreviation of $e_1 \trans_e v_1, \ldots , e_n \trans_e v_n$.
Note that we assume that no operation may yield a name as a value.
Thus, names may appear as arguments to an operator $\op$, e.g.\@ for equality testing, but an operator may not be used to \emph{create} a name.

We can give the semantics for processes as a (mostly) standard, early labelled semantics.
First, we define the labels:
\begin{syntax}[h]
  \ALPHA \IS \TAU 
         \OR \LSEND{\VEC{x}}{\NEW{\VEC{y}}\VEC{v}} 
         \OR \LRECV{\VEC{x}}{\VEC{v}}
\end{syntax}
The $\TAU$ label is the unobservable action, and the label $\LRECV{\VEC{x}}{\VEC{v}}$ is for inputting values $\VEC{v}$ from the name vector $\VEC{x}$.
The label $\LSEND{\VEC{x}}{\NEW{\VEC{y}}\VEC{v}}$ is for (bound) output: here, $\VEC{y}$ is a list of bound names, such that $\VEC{y} \subseteq \VEC{v}$, i.e.\@ they bind into the object vector, and $\FRESH{\VEC{y}}\VEC{x}$ (meaning $\VEC{y} \cap \VEC{x} = \emptyset$), i.e.\@ none of the bound names can appear in the subject vector.
Hence, we define $\BN{\LSEND{\VEC{x}}{\NEW{\VEC{y}}\VEC{v}}} = \VEC{y}$, and $\BN{\alpha} = \emptyset$ for all other labels $\alpha$.
If there are no bound names, we shall just write $\LSEND{\VEC{x}}{\VEC{v}}$ instead of $\LSEND{\VEC{x}}{(\NEWSYM\epsilon)\VEC{v}}$, which is the label for \emph{free} output. 

\begin{figure}
\begin{center}
\begin{semantics}
   \RULE[E-Out][epi_early_out]({ \VEC{e} \trans_e \VEC{v} })
    { }
    { \OUTPUT{\VEC{x}}{\VEC{e}}.P \trans[\LSEND{\VEC{x}}{\VEC{v}}] P } 

  \RULE[E-Open][epi_early_open]({
    \begin{array}{l}
      \VEC{z} \subseteq \VEC{v} \\
      \FRESH{\VEC{z}\ }\ \VEC{x},\VEC{y}
    \end{array}
  })
    { P \trans[\LSEND{\VEC{x}}{\NEW{\VEC{y}}\VEC{v}}] P' }
    { \NEW{\VEC{z}}P \trans[\LSEND{\VEC{x}}{\NEW{\VEC{y},\VEC{z}}\VEC{v}}] P' }

  \RULE[E-Com$_1$][epi_early_com1]({
      \FRESH{\VEC{y}}P_2})
    { P_1 \trans[\LSEND{\VEC x}{\NEW{\VEC{y}}\VEC v}] P_1' \AND P_2 \trans[\LRECV{\VEC{x}}{\VEC{v}}] P_2' }
    { P_1 \PAR P_2 \trans[\TAU] \NEW{\VEC{y}}\PAREN{P_1' \PAR P_2'} }

  \RULE[E-Sum][epi_early_sum]({ e_i \trans_e \TRUE  })
    { P_i \trans[\ALPHA] P_i' }
    { \sum [\VEC{e}]\VEC{P} \trans[\ALPHA] P_i' }
\end{semantics}
\begin{semantics}
  \RULE[E-In][epi_early_in]
    { }
    { \INPUT{\VEC{x}}{\VEC{y}}.P \trans[\LRECV{\VEC{x}}{\VEC{v}}] P\SUBSTITUTE{\VEC{v}}{\VEC{y}} }

  \RULE[E-Res][epi_early_res](\FRESH{\VEC{x}}\ALPHA)
    { P \trans[\ALPHA] P' }
    { \NEW{\VEC{x}}P \trans[\ALPHA] \NEW{\VEC{x}}P' }

  \RULE[E-Par$_1$][epi_early_par1](\FRESH{\BN{\ALPHA}}P_2)
    { P_1 \trans[\ALPHA] P_1' }
    { P_1 \PAR P_2 \trans[\ALPHA] P_1' \PAR P_2 }

  \RULE[E-Rep][epi_early_rep]
    { \REPL{P} \PAR P \trans[\ALPHA] P' }
    { \REPL{P} \trans[\ALPHA] P' }

\end{semantics}
\end{center}
\caption{Early labelled semantics for the $\EPI$-calculus.}
\label{wc:fig:epi_early_semantics}
\end{figure}

The early, labelled semantics is given by the transition relation $\trans[\ALPHA]$ defined by the rules in Figure~\ref{wc:fig:epi_early_semantics}.
In the semantic rules, a substitution, written $P\SUBSTITUTE{\VEC{v}}{\VEC{x}}$, is the pointwise (capture avoiding) replacement of each name $x_i$ by the corresponding value $v_i$; thus we require that $\VEC{x}$ and $\VEC{v}$ have the same arity, and we therefore assume that terms are well-sorted.
Note that we omit the symmetric versions of the rules \nameref{epi_early_par1} and \nameref{epi_early_com1}.
Furthermore, we assume all rules are defined up to \ALPHA-equivalence, and that bound names within $\NEW{\VEC{x}}P$ can be reordered.
Lastly, we use the notation $\FRESH{\VEC{x}}{P}$ to say that the names $\VEC{x}$ are \emph{fresh} for $P$, analogously to the notation used for labels.

The \nameref{epi_early_par1}, \nameref{epi_early_res}, \nameref{epi_early_open} and \nameref{epi_early_in} are standard, similar to e.g.\@ the labelled semantics for the \PI-calculus given in \cite{parrow2001introduction}, but extended in the obvious way to allow both polyadic objects, as in the polyadic \PI-calculus \cite{PICALC}, and polyadic subjects, as in \cite{CARBONEMAFFEIS}.
As usual, polyadicity (i.e.\@ polyadic objects) leads to have a bound output label both in the premise and in the conclusion of the \nameref{epi_early_open} rule.
In the rule \nameref{epi_early_com1}, the list of bound names $\VEC{y}$ in $\LSEND{\VEC x}{\NEW{\VEC y}\VEC v}$ can be empty; thus, it can act both as an ordinary communication rule and as a close rule.

Finally, as usual, we write $\trans[\tau]*$ to denote the reflexive and transitive closure of $\trans[\tau]$.

\section{A simple type system for \EPI}\label{wc:sec:epi_type_system}
The type system we wish to create is a standard type system for ensuring correct channel usage, similar to the simple type system for the polyadic \PI-calculus \cite{sangiorgi2003pi, turner1996phd}.
The types are intended to describe \emph{data} (i.e.\@ channel usage and values), so they will be assigned once and do not change.
In the type system for the polyadic \PI-calculus, the types $T$ are of the form $\TCHAN{\VEC{T}}$, for names that can be used as channels (subjects) to send/receive values of types $\VEC{T}$, and $\TNIL$, for names that cannot be used as channels.
Other data values (e.g.\@ integers and booleans) have a base type $B$ as usual (e.g.\@ $\TINT$, $\TBOOL$).
Thus, for example, $\TCHAN{\TINT, \TCHAN{\TNIL}}$ is the type of a channel that can be used to communicate two values: an integer and a name of type $\TCHAN{\TNIL}$, which in turn can be used as a channel to send/receive names that cannot be used as channels.

A key problem in devising a type system for \EPI{} is to decide how type judgments for subject vectors $\VEC{x}$ should be concluded, if types are given to individual names.
Of course, this does not have necessarily to be the case: types could be given directly to vectors of names.
However, that is an inflexible solution. Indeed, since subject vectors can be composed at runtime in \EPI, using the above-mentioned approach would necessitate that the type environment $\Gamma$ must contain entries for all name vectors that  may potentially be created during the execution of a process.
This might be sufficient in certain scenarios, e.g.\@ with translations from other languages, where the form of all subject vectors may be known (and known not to change at runtime), but as a general solution it does not seem satisfactory.
Hence, we maintain that types should be given to individual names.
A number of possible solutions then come to mind: 
\begin{itemize}
  \item We can require that all names in a name vector must have the same type to be well-typed.
    Thus, if $x_1$ has type $T$ and $x_2$ has type $T$, then $x_1 \cdot x_2$ would have type $T$ as well, and so would $x_2 \cdot x_1$.
    This solution does not seem satisfactory, since it does not distinguish between using the names individually, and using the composition.

  \item We can let types be \emph{sets} of capabilities.
  The type of a composition of names could then be the intersection of the sets.
    For example, if 
    \begin{equation*}
      \Gamma \vdash x_1 : \SET{ \TCHAN{T_1}, \TCHAN{T_2} } \qquad \text{and} \qquad  \Gamma \vdash x_2 : \SET{ \TCHAN{T_2}, \TCHAN{T_3} }
    \end{equation*}

    \noindent then $\Gamma \vdash x_1 \cdot x_2 : \SET{ \TCHAN{T_2} }$.
    However, that would lead to issues of non-well-foundedness, if we want to allow names to be sent on the vector, which are also \emph{in} the vector.
    For example, to type the process $\OUTPUT{x_1 \cdot x_2}{x_1}.\NIL$, we would need the type of $x_1 \cdot x_2$ to contain $\TCHAN{T}$, which implies that both the types (sets) of $x_1$ and $x_2$ must contain $\TCHAN{T}$. 
    But since $T$ is the type of $x_1$, then we have some non well-foundedness issues.

  \item We can allow names to appear in types to express which compositions should be allowed.
    This is the method used in the type system in \cite{huttel2011typed} for the distributed \PI-calculus D\PI{} \cite{DPICALC}, as described above in Section~\ref{wc:sec:introduction}.
    Suppose, for example, that we have $\Gamma \vdash x_1 : \SET{ x_2, x_3 }$, to signify that $x_1$ can be followed by either $x_2$ or $x_3$ in a composition.
    This avoids the problem of non-well-foundedness, since we could then have e.g.\@ that $\Gamma \vdash x_2 : \TCHAN{\SET{x_2, x_3}}$, which would make it possible to type the process $\OUTPUT{x_1 \cdot x_2}{x_1}.\NIL$.
    However, having names appear directly in types creates other problems, if the names are bound.
    Consider the process
    \begin{equation*}
      P \DEFSYM \NEW{x_2 : \TCHAN{\SET{x_2, x_3}}}( \OUTPUT{x_1 \cdot x_2}{x_1}.\NIL \PAR \INPUT{x_1 \cdot x_2}{y}.\NIL )
    \end{equation*}

    \noindent with the type annotation added in the restriction.
    If $\Gamma \vdash x_1 : \SET{x_2, x_3}$ as before, then $P$ is well-typed.
    However, since $x_2$ is bound, then $P$ is \ALPHA-equivalent to another process $P'$ where the name $x_2$ has been converted to some other, fresh name, e.g.\@ $z$.
    But $z$ does not appear in the type of $x_1$ in $\Gamma$, so \ALPHA-conversion must also be applied to $\Gamma$, lest well-typedness would not be preserved by \ALPHA-equivalence.
    This is doable, but at least inelegant, since it would mean that all restrictions, at least implicitly, must be assumed to have scope over $\Gamma$.
    Furthermore, it also does not allow us to distinguish between the use of $x_2$ alone and in the composition $x_1 \cdot x_2$.
\end{itemize}

\subsection{The language of types}\label{wc:sec:epi_simple_type_language}
In light of the considerations above, we shall choose our language of types as follows.
Firstly, to avoid using channel names directly in types, we shall instead use \emph{named types} as was done in \cite{carbone2005phd}.
Thus, we introduce a separate, countably infinite set of atomic \emph{type names} $\TNAMES$, ranged over by $I$, which is distinct from the set of channel names $\NAMES$.
To avoid the issues of non-well-foundedness, we shall give every type a name from this set, i.e.\@ we shall be using \emph{named types}.

Secondly, the capability of a name to be used as a channel and its capability to be used in compositions appear to be more or less orthogonal.
Hence, we shall distinguish between the `channel capabilities' of a name and its `composition capabilities.'
We shall therefore use a pair for the type of a name: the first component describes whether the name can be used as a channel or not, and, in the former case, what it can be used to communicate; the second component describes whether the name can be composed with other names and, if so, which ones. 

Lastly, as $x_1 \cdot x_2$ is distinct from $x_2 \cdot x_1$, we may want to allow a name to have different meanings, depending on \emph{where} it occurs in a composition.
The type of a name is therefore not a single tuple, but a \emph{tree} of capabilities.
One way to view this is to think of each type name $I$ as a declaration of a name space, which in turn may contain other name space declarations.

\begin{definition}[Types]\label{wc:def:epi_simple_type_language}
We use the following language of types and type environments:
\begin{center}
\begin{syntax}[h]
  B                             \IS I \OR \TINT \OR \TBOOL                \NAME{base types}       \tabularnewline
  C                             \IS \TCHAN{\VEC{B}} \OR \TNIL             \NAME{capability types} \tabularnewline
  T              \in \TYPES     \IS B \OR (C, \Delta)                     \NAME{types}            \tabularnewline
  \Gamma, \Delta \in \TYPEENVS  \IS \NAMES \UNION \TNAMES \PARTIAL \TYPES \NAME{type environments}
\end{syntax}
\end{center}

\noindent We assume that $\Gamma$ can be written as an unrodered sequence of pairs $x: T , \Gamma'$, with $\emptyset$ denoting the empty environment; we use the same notation to denote that the environment $\Gamma'$ is extended with the pair assigning type $T$ to $x$.
\end{definition}

$\Gamma$ stores assumptions about the types of names, as well as about type names.
The idea in this type language is that a name $x$ should only have one \emph{base type} $B$, as recorded in $\Gamma$.
The type of a name should then almost always be a type name $I$, except when we type the continuation after an input, since the bound names here may also be substituted for the other data types.
Furthermore, for a type name $I$ to be meaningful, it must also have its own entry in $\Gamma$, which must be a tuple $(C, \Delta)$, where $C$ is a capability type and $\Delta$ is a type environment that records the meaning of type names under the present name.
We shall use the notation $I \mapsto (C, \Delta)$ to refer to a type name plus its entry, and we say a type environment $\Gamma$ is \emph{well-formed} if all type names have entries.
We shall only consider well-formed type environments in the following.

We shall henceforth use a \emph{typed} syntax for processes, with explicit (base) types added to the declaration of new names, i.e.\@ $\NEW{\VEC{x} : \VEC{B}}P$, where $|\VEC{x}| = |\VEC{B}|$.
Likewise, types must also now appear in the list of bound names in labels, so we write $\LSEND{\VEC{x}}{\NEW{\VEC{z} : \VEC{B}}\VEC{v}}$.

The type language above is richer than  the usual one for the \PI-calculus, although most of the complexities will not appear directly in the syntax of processes, since the idea is that names should only be given base types $B$.
However, with this more complex type language, we can also achieve the same effect as some of the simpler solutions mentioned above, such as for example requiring that all names in a vector must have the same type.
This can easily be achieved, since types are referred to by their type name $I$, so multiple names can have the \emph{same} type (and not just structurally similar types).
As an added benefit, named types also allow us to type processes such as $\OUTPUT{x}{x}.\NIL$ without requiring explicit recursive types to be added to the language, since $x$ can now simply be given the type $I \mapsto (\TCHAN{I}, \EMPTYSET)$.

\begin{figure}[t]
\begin{center}
\begin{semantics}
  \RULE[t-val][epi_simple_type_val]({B = 
    \begin{cases}
      \TINT     & \text{if $v \in \mathbb{Z}$} \tabularnewline
      \TBOOL    & \text{if $v \in \mathbb{B}$} \tabularnewline
      \Gamma(v) & \text{if $v \in \NAMES$}
    \end{cases}
  })
    { }
    { \Gamma \vdash v : B}

  \RULE[t-vec$_1$][epi_simple_type_vec1]
    { (\snd \circ \Delta \circ \Gamma)(x) = \Delta_x  \AND \Gamma; \Delta_x \vdash \VEC{x} : C }
    { \Gamma; \Delta \vdash x \cdot \VEC{x} : C }

  \RULE[t-out][epi_simple_type_out]
    { \Gamma; \Gamma \vdash \VEC{x} : \TCHAN{\VEC{B}} \AND \Gamma \vdash \VEC{e} : \VEC{B} \AND \Gamma \vdash P }
    { \Gamma \vdash \OUTPUT{\VEC{x}}{\VEC{e}} . P }

  \RULE[t-sum][epi_simple_type_sum]
    { \forall i. (\Gamma \vdash e_i : \TBOOL \AND \Gamma \vdash P_i) }
    { \Gamma \vdash \sum [\VEC{e}]\VEC{P} }
\end{semantics}
\begin{semantics}
  \tabularnewline
  \RULE[t-op][epi_simple_type_op]
    { \Gamma \vdash \VEC{e} : \VEC{B} \AND \op : \VEC{B} \to B }
    { \Gamma \vdash \op(\VEC{e}) : B }
  \tabularnewline

  \RULE[t-vec$_2$][epi_simple_type_vec2]
    { (\fst \circ \Delta \circ \Gamma)(x) = C }
    { \Gamma; \Delta \vdash x : C }

  \RULE[t-in][epi_simple_type_in]
    { \Gamma; \Gamma \vdash \VEC{x} : \TCHAN{\VEC{B}} \AND \Gamma, \VEC{y} : \VEC{B} \vdash P }
    { \Gamma \vdash \INPUT{\VEC{x}}{\VEC{y}}.P }

  \RULE[t-rep][epi_simple_type_rep]
    { \Gamma \vdash P}
    { \Gamma \vdash \REPL{P} }
\end{semantics}

\bigskip
\begin{semantics}
  \RULE[t-nil][epi_simple_type_nil]
    { }
    { \Gamma \vdash \NIL }
\end{semantics}\qquad
\begin{semantics}
  \RULE[t-res][epi_simple_type_res]
    { \Gamma, \VEC{x} : \VEC{B} \vdash P }
    { \Gamma \vdash \NEW{\VEC{x} : \VEC{B}}P }
\end{semantics}\qquad
\begin{semantics}
  \RULE[t-par][epi_simple_type_par]
    { \Gamma \vdash P_1 \AND \Gamma \vdash P_2}
    { \Gamma \vdash P_1 \PAR P_2 }
\end{semantics}
\end{center}
\caption{Type rules for values, subject vectors, and processes.}
\label{wc:fig:epi_simple_type_judgments}
\end{figure}

\subsection{The type system}
The type judgement $\Gamma \vdash P$ ensures that $P$ uses channels as prescribed by $\Gamma$, ensures that it only contains well-typed expressions 
(e.g., by rejecting $\TRUE + 2$), and also rules out obviously malformed processes such as \mbox{$\INPUT{x \cdot 2}{y}.P'$}, since such processes would lead to the semantics being stuck.
As we assume well-sortedness of terms, we do not check for arity-mismatch in output objects, although this can also be included if desired.

The rules for values, subject vectors and processes are given in Figure~\ref{wc:fig:epi_simple_type_judgments}.
For expressions, if $e$ is a value, we use rule \nameref{epi_simple_type_val}; if $e$ contains an operator $\op$, we use rule \nameref{epi_simple_type_op}, that relies on the signature $\VEC{B} \to B$, obtained by considering $\op$ as a function $\op : \VALUES^* \to \VALUES$ (e.g.\@ $+ : \code{int}^2 \to \code{int}$ and $\land : \code{bool}^2 \to \code{bool}$).
As expectable, the main peculiarity of these rules is the typing for subject vectors $\VEC{x}$ (required in the premise of rules \nameref{epi_simple_type_in} and \nameref{epi_simple_type_out}).
This judgment is of the form $\Gamma; \Gamma \vdash \VEC{x} : C$, and is concluded by the rules \nameref{epi_simple_type_vec1} and \nameref{epi_simple_type_vec2}.
To conclude that $x_1 \cdot x_2 \cdot \ldots\cdot x_n$ (for $n > 1$) has type $C$, we proceed as follows: 
\begin{enumerate}
  \item we start with rule \nameref{epi_simple_type_vec2}, that takes $x_1$ and looks it up in $\Gamma$; since $\Gamma$ is well-formed, this lookup must yield $\Gamma(x_1) = I_{x_1}$, with $\Gamma(I_{x_1}) = (C_{x_1}, \Delta_{x_1})$;
  \item we then use function $\snd$ to extract the second component of this tuple, i.e. the environment $\Delta_{x_1}$, that will be used to give the meaning of a selection of type names below $I_{x_1}$ (the names that may be composed with $x_1$ are the names that occur in the domain of $\Delta_{x_1}$);
  \item we recur on $x_2 \cdot \ldots\cdot x_n$, but we now use $\Delta_{x_1}$ for giving a type to $I_{x_2}$ (= $\Gamma(x_2)$) in point 1 above.
\end{enumerate}
\noindent For vectors of length 1, we use \nameref{epi_simple_type_vec1}, which consist of the same first two steps above, except that it uses the function $\fst$ to extract the returned channel capability $C$.

As an example, imagine that $x_1$ is a name that: (1) when used alone, it can carry integers; (2) when followed by $x_2$, it can carry pairs of integers; (3) when followed by $x_2$ and $x_3$, it can carry booleans; and (4) when preceeded by $x_2$, it can carry pairs of booleans. To accommodate this situation, we need to use a typing environment $\Gamma$ such that:
$\Gamma(x_1) = I_1$, 
$\Gamma(x_2) = I_2$, 
$\Gamma(x_3) = I_3$, 
$\Gamma(I_1) = \langle \TCHAN{\TINT} , \Delta_1 \rangle$, and
$\Gamma(I_2) = \langle \TNIL , \Delta_2 \rangle$, where
$\Delta_1(I_2) = \langle \TCHAN{\TINT,\TINT} , \Delta_3 \rangle$,
$\Delta_2(I_1) = \langle \TCHAN{\TBOOL,\TBOOL} , \EMPTYSET \rangle$, and
$\Delta_3(I_3) = \langle \TCHAN{\TBOOL} , \EMPTYSET \rangle$. 
Then, the processes $\OUTPUT{x_1}{3}$, $\OUTPUT{x_1\cdot x_2}{3,5}$, $\OUTPUT{x_1\cdot x_2 \cdot x_3}{\TRUE}$, and $\OUTPUT{x_2\cdot x_1}{\TRUE,\FALSE}$ are all typeable under $\Gamma$. 
We only depict the third inference, since the other ones are similar or simpler: 
$$
\dfrac{
\dfrac{
\begin{array}{c}
\tabularnewline
(\snd \circ \Gamma \circ \Gamma)(x_1) = \Delta_1  
\end{array}
\ 
\dfrac{
\begin{array}{c}
\tabularnewline
(\snd \circ \Delta_1 \circ \Gamma)(x_2) = \Delta_3  
\end{array}
\ 
\dfrac{(\fst \circ \Delta_3 \circ \Gamma)(x_3) = \TCHAN{\TBOOL}}
{\Gamma; \Delta_3 \vdash x_3 : \TCHAN{\TBOOL }}
}
{\Gamma; \Delta_1 \vdash x_2 \cdot x_3 : \TCHAN{\TBOOL}}
}
{\Gamma;\Gamma \vdash x_1\cdot x_2 \cdot x_3 : \TCHAN{\TBOOL}}
\
\begin{array}{c}
\tabularnewline
\Gamma \vdash \TRUE : \TBOOL
\quad
\Gamma \vdash \NIL
\end{array}
}{
\Gamma \vdash \OUTPUT{x_1\cdot x_2 \cdot x_3}{\TRUE}.\NIL
}
$$

\subsection{Run-time errors, Safety, and Soundness results}
\begin{figure}[h]
\begin{center}
\begin{semantics}
  \RULE[w-in][epi_simple_wrong_in]
    { \Gamma \not\vdash \VEC{x} : \TCHAN{\VEC{B}}\ \lor\ |\VEC{B}| \neq |\VEC{y}| }
    { \WRONG{\INPUT{\VEC{x}}{\VEC{y}}.P} }

  \RULE[w-par][epi_simple_wrong_par]
    { \WRONG{P_1}\ \lor\ \WRONG{P_2} }
    { \WRONG{P_1 \PAR P_2 } }

  \RULE[w-res][epi_simple_wrong_res]
    { \WRONG[\Gamma, \VEC{x} : \VEC{B}]{P}}
    { \WRONG{\NEW{\VEC{x} : \VEC{B}}P} }
\end{semantics}
\qquad\qquad
\begin{semantics}
  \RULE[w-out][epi_simple_wrong_out]
    {\Gamma \not\vdash \VEC{x} : \TCHAN{\VEC{B}}\ \lor\ \Gamma \not\vdash \VEC{e} : \VEC{B} }
    { \WRONG{\OUTPUT{\VEC{x}}{\VEC{e}}.P} }

  \RULE[w-sum][epi_simple_wrong_sum]
    {\exists i. (\Gamma \not\vdash e_i : \TBOOL\ \lor\ \WRONG{P_i}) }
    { \WRONG{\sum [\VEC{e}]\VEC{P}} }

  \RULE[w-rep][epi_simple_wrong_rep]
    { \WRONG{P} }
    { \WRONG{\REPL{P}} }
\end{semantics}
\end{center}
\caption{Error-predicate for processes.}
\label{wc:fig:epi_simple_type_error}
\end{figure}

We formalise the notion of safety as an error-predicate, written $\WRONG{P}$, which checks for mismatched types in input/output, and also whether each expression guard in a sum has a boolean type.
It is given by the rules in Figure~\ref{wc:fig:epi_simple_type_error}.
We say that a process $P$ is \emph{now-safe}, written $\NSAFE{P}$, if $P$ does not satisfy the error-predicate; and likewise, we say that a process is \emph{safe}, written $\SAFE{P}$, if it is now-safe for all its \TAU-labelled transitions:
\begin{equation*}
  \NSAFE{P} \DEFSYM \neg \WRONG{P} \qquad
  \SAFE{P} \DEFSYM \forall P'\SUCHTHAT (P {\trans[\tau]*} P' \implies \NSAFE{P'})
\end{equation*}

\noindent Note that the concepts of now-safety and safety only make sense if all the free names of $P$ appear in $\Gamma$; hence, in what follows, we shall implicitly assume that $\FN{P} \subseteq \DOM{\Gamma}$.
Then, our first result is that well-typed processes are now-safe: 

\begin{theorem}[Safety]\label{wc:thm:epi_simple_type_safety}
If $\Gamma \vdash P$ then $\NSAFE{P}$.
\end{theorem}

\begin{proof}
By induction on (the height of) the judgement $\Gamma \vdash P$. The base case is when \nameref{epi_simple_type_nil} has been used and it is trivial, since the error predicate has no rule for $\NIL$.
For the inductive step, we reason by case analysis on the last type rule used for concluding $\Gamma \vdash P$.
\begin{itemize}

  \item \nameref{epi_simple_type_par}: 
  From the premise, we know that $\Gamma \vdash P_1$ and $\Gamma \vdash P_2$.
  By the induction hypothesis, this implies that $\NSAFE{P_1}$ and $\NSAFE{P_2}$, hence $\neg \WRONG{P_1}$ and $\neg \WRONG{P_2}$.
  The only rule, that can conclude an error in a parallel composition is $\nameref{epi_simple_wrong_par}$, but by the above, none of the premises can hold.
  Thus we conclude that $\NSAFE{P_1 \PAR P_2}$.

  \item \nameref{epi_simple_type_rep}: 
  From the premise, we have that $\Gamma \vdash P$.
  By the induction hypothesis, this implies $\NSAFE{P}$, hence $\neg \WRONG{P}$.
  To conclude that $\REPL{P}$ is not safe, we can only use \nameref{epi_simple_wrong_rep}, but by the above, the premise cannot hold.
  Thus we conclude that $\NSAFE{\REPL{P}}$.

  \item \nameref{epi_simple_type_res}: 
  From the premise, we have that $\Gamma, \VEC{x} : \VEC{B} \vdash P$, which by the induction hypothesis implies that $\NSAFE[\Gamma, \VEC{x}:\VEC{B}]{P}$, hence $\neg \WRONG[\Gamma, \VEC{x}:\VEC{B}]{P}$.
  The only rule that can be used to conclude that $\WRONG{\NEW{\VEC{x}:\VEC{B}}P}$ is \nameref{epi_simple_wrong_res}, but by the above, the premise cannot hold.
  Thus we conclude that $\NSAFE{\NEW{\VEC{x}:\VEC{B}}P}$. 

  \item \nameref{epi_simple_type_sum}: 
  From the premise, we have that every guard $e_i$ must have a boolean type, and every $P_i$ is such that $\Gamma \vdash P_i$.
  By the induction hypothesis, this implies $\NSAFE{P_i}$, hence $\neg \WRONG{P_i}$.
  The only rule that can be used to conclude $\WRONG{\sum [\VEC{e}]\VEC{P}}$ is \nameref{epi_simple_wrong_sum}, but by the above the premise cannot hold for every $i$.
  Thus we conclude that $\NSAFE{\sum [\VEC{e}]\VEC{P}}$.

  \item \nameref{epi_simple_type_in}: 
  From the premise, we know that $\Gamma;\Gamma \vdash \VEC{x} : \TCHAN{\VEC{B}}$, and $\Gamma, \VEC{y} : \VEC{B} \vdash P$, where the latter implies that $|\VEC{y}| = |\VEC{B}|$, since every name must be given a type.
  The only rule that can be used to conclude an error for input is \nameref{epi_simple_wrong_in}, but by the above, the premises cannot hold.
  Hence we conclude that $\NSAFE{\INPUT{\VEC{x}}{\VEC{y}}.P}$.

  \item \nameref{epi_simple_type_out}: 
  From the premise we have that $\Gamma \vdash \VEC{x} : \TCHAN{\VEC{B}}$ and $\Gamma \vdash \VEC{e} : \VEC{B}$.
  The only rule that can be used to conclude an error for an output is \nameref{epi_simple_wrong_out}, but by the above, none of the premises can hold.
  Hence we conclude that $\NSAFE{\OUTPUT{\VEC{x}}{\VEC{e}}.P}$.
\end{itemize}

This concludes the proof.
\end{proof}

The second result states that well-typedness is preserved by the (labelled) semantics.
To express this, we first need the concept of well-typed labels:

\begin{definition}[Well-typed label]\label{wc:def:epi_simple_type_welltyped_labels}
We say that a label $\alpha$ is well-typed for a given $\Gamma$ if it can be concluded using one of the following rules:
\begin{center}
\begin{fixqed}
\begin{semantics}
  \RULE[t-tau][epi_simple_type_label_tau]
    { }
    { \Gamma \vdash \tau }
\end{semantics}\kern-6pt
\ \
\begin{semantics}
  \RULE[t-snd][epi_simple_type_label_send]
    { \Gamma; \Gamma \vdash \VEC{x} : \TCHAN{\VEC{B}_2} \AND \Gamma, \VEC{z} : \VEC{B}_1 \vdash \VEC{v} : \VEC{B}_2 }
    { \Gamma \vdash \LSEND{\VEC{x}}{\NEW{\VEC{z}:\VEC{B}_1}\VEC{v}} }
\end{semantics}\kern-6pt
\ \ 
\begin{semantics}
  \RULE[t-rcv][epi_simple_type_label_recv]
    { \Gamma; \Gamma \vdash \VEC{x} : \TCHAN{\VEC{B}} \AND \Gamma \vdash \VEC{v} : \VEC{B} }
    { \Gamma \vdash \LRECV{\VEC{x}}{\VEC{v}} }
\end{semantics}
\end{fixqed}\qedhere
\end{center}
\end{definition}

The second result we desire to have is then that well-typedness is preserved by the transition relation:
\begin{theorem}[Subject reduction]\label{wc:thm:epi_simple_subject_reduction} 
Let $\Gamma \vdash P$.
If $P \trans[\alpha] P'$ and $\Gamma \vdash \alpha$, then $\Gamma' \vdash P'$, where $\Gamma' = \Gamma, \VEC{z} : \VEC{B}$ if $\alpha = \LSEND{\VEC{x}}{\NEW{\VEC{z} : \VEC{B}}\VEC{v}}$, and $\Gamma' = \Gamma$ otherwise.
\end{theorem}

The proof this theorem relies on some standard lemmata, that are all easily shown by induction on the typing judgement.
Weakening and strengthening say that we can add and remove unused typing assumptions.
Substitution says that we can replace a list of names by a list of values (which may include names), as long as the two lists have pointwise matching types; hence, no well-typed substitution can alter the type of (a vector of) names.

\begin{lemma}[Weakening]\label{wc:lemma:epi_simple_type_weakening}
If $\Gamma \vdash P$ and $x \notin \FN{P}$, then $\Gamma, x : B \vdash P$.
\end{lemma}

\begin{lemma}[Strengthening]\label{wc:lemma:epi_simple_type_strengthening}
If $\Gamma, x : B \vdash P$ and $x \notin \FN{P}$, then $\Gamma \vdash P$.
\end{lemma}

\begin{definition}[Well-typed substitution]
We say a substitution $\sigma$ is well-typed, written $\Gamma \vdash \sigma$, if $\DOM{\sigma} \subseteq \DOM{\Gamma}$ and $\forall x \in \DOM{\sigma} \SUCHTHAT \Gamma \vdash \sigma(x) : \Gamma(x)$.
\end{definition}

\begin{lemma}[Substitution]\label{wc:lemma:epi_simple_type_substitution}
If $\Gamma \vdash P$ and $\Gamma \vdash \sigma$, then $\Gamma \vdash P\sigma$.
\end{lemma}

The next lemma says that expression evaluation respects typing; this is shown by induction on the rules of $\trans_e$, which we have omitted since they depend on the operators $\op$ that we choose to admit in the language.

\begin{lemma}[Safety for expressions]\label{wc:lemma:epi_simple_type_safety_expressions}
If $\Gamma \vdash \VEC{e} : \VEC{B}$ and $\VEC{e} \trans_e \VEC{v}$ then $\Gamma \vdash \VEC{v} : \VEC{B}$
\end{lemma}

Lastly, we show that a well-typed process can only originate well-typed (output/$\tau$) labels:
\begin{lemma}[Well-typed labels]\label{wc:lemma:epi_simple_type_welltyped_labels}
If $\Gamma \vdash P$ and $P \trans[\alpha] P'$, for $\alpha = \tau$ or $\alpha = \LSEND{\VEC{x}}{\NEW{\VEC{z} : \VEC{B}}\VEC{v}}$, then $\Gamma \vdash \alpha$.
\end{lemma}

\begin{proof}[Proof sketch]
If the label is $\tau$, then the result is immediate, since $\Gamma \vdash \tau$ for any $\Gamma$ by \nameref{epi_simple_type_label_tau}.
Otherwise, we have a number of cases to examine, depending on which rule was used to conclude the transition.
However, the two most interesting cases are for the output and open rules:
\begin{itemize}
  \item If \nameref{epi_early_out} was used, then the conclusion is of the form $\OUTPUT{\VEC{x}}{\VEC{e}}.P \trans[\LSEND{\VEC{x}}{\VEC{v}}] P$.
    The label is a free output.
    From the premise we have that $\VEC{e} \trans_e \VEC{v}$.
    Now $\Gamma \vdash \OUTPUT{\VEC{x}}{\VEC{e}}.P$ must have been concluded by \nameref{epi_simple_type_out}, and from its premises we know that $\Gamma; \Gamma \vdash \VEC{x} : \TCHAN{\VEC{B}}$ and $\Gamma \vdash \VEC{e} : \VEC{B}$.
    By Lemma~\ref{wc:lemma:epi_simple_type_safety_expressions}, we then have that $\Gamma \vdash \VEC{v} : \VEC{B}$.
    Hence, by \nameref{epi_simple_type_label_send}, we have that $\Gamma \vdash \LSEND{\VEC{x}}{\VEC{v}}$.

  \item If \nameref{epi_early_open} was used, then the conclusion is of the form $\NEW{\VEC{z} : \VEC{B}_z}P \trans[\LSEND{\VEC{x}}{\NEW{\VEC{y},\VEC{z} : \VEC{B}_y,\VEC{B}_z}\VEC{v}}] P'$.
    From the premise and side condition, we have that $P \trans[\LSEND{\VEC{x}}{\NEW{\VEC{y} : \VEC{B}_y}\VEC{v}}] P'$, $\VEC{z} \subseteq \VEC{v}$ and $\FRESH{\VEC{z}\ }\ \VEC{x},\VEC{y}$.
    Now $\Gamma \vdash \NEW{\VEC{z} : \VEC{B}_z}P$ must have been concluded by \nameref{epi_simple_type_res}, and from its premise we know that $\Gamma, \VEC{z} : \VEC{B}_z \vdash P$.
    By the induction hypothesis, we then have that $\Gamma, \VEC{z} : \VEC{B}_z \vdash \LSEND{\VEC{x}}{\NEW{\VEC{y} : \VEC{B}_y}\VEC{v}}$.
    We can then conclude $\Gamma \vdash \LSEND{\VEC{x}}{\NEW{\VEC{y},\VEC{z} : \VEC{B}_y,\VEC{B}_z}\VEC{v}}$ by \nameref{epi_simple_type_label_send}.
\end{itemize}

For the remaining \ALPHA-labelled rules  (\nameref{epi_early_par1}, \nameref{epi_early_sum}, \nameref{epi_early_res} and \nameref{epi_early_rep}), the result follows directly from the induction hypothesis.
For \nameref{epi_early_in}, the statement holds vacuously, since the label is not an output- nor a \TAU-label.
\end{proof}

Finally, we can prove the Subject Reduction result stated above:
\begin{proof}[Proof of Theorem~\ref{wc:thm:epi_simple_subject_reduction}]
By induction on the inference for $P \trans[\alpha] P'$.
The base case can be inferred in two ways:
\begin{itemize}
  \item If \nameref{epi_early_in} was used,
    then the conclusion is of the form $\INPUT{\VEC{x}}{\VEC{y}}.P \trans[\LRECV{\VEC{x}}{\VEC{v}}] P\SUBSTITUTE{\VEC{v}}{\VEC{y}}$.
    Now $\Gamma \vdash \INPUT{\VEC{x}}{\VEC{y}}.P$ must have been concluded by \nameref{epi_simple_type_in}, and from its premise we know that $\Gamma; \Gamma \vdash \VEC{x} : \TCHAN{\VEC{B}}$ and $\Gamma, \VEC{y} : \VEC{B} \vdash P$.
    By assumption, $\Gamma \vdash \LRECV{\VEC{x}}{\VEC{v}}$, which must have been concluded by \nameref{epi_simple_type_label_recv}, and from the premise we know that $\Gamma \vdash \VEC{v} : \VEC{B}$. 
    By Lemma~\ref{wc:lemma:epi_simple_type_substitution}, we can therefore conclude that $\Gamma, \VEC{y} : \VEC{B} \vdash P\SUBSTITUTE{\VEC{v}}{\VEC{y}}$.
    As $\FRESH{\VEC{y}}\FN{P\SUBSTITUTE{\VEC{v}}{\VEC{y}}}$, we can then conclude $\Gamma \vdash P\SUBSTITUTE{\VEC{v}}{\VEC{y}}$ by Lemma~\ref{wc:lemma:epi_simple_type_strengthening}.
    As we know that the label here is an input label, we therefore have that $\Gamma' = \Gamma$.

  \item If \nameref{epi_early_out} was used,
    then the conclusion is of the form $\OUTPUT{\VEC{x}}{\VEC{e}}.P \trans[\LSEND{\VEC{x}}{\VEC{v}}] P$, and from the premise we have that $\VEC{e} \trans_e \VEC{v}$.
    Now $\Gamma \vdash \OUTPUT{\VEC{x}}{\VEC{e}}.P$ must have been concluded by \nameref{epi_simple_type_out}, and from its premises we know that $\Gamma \vdash P$, thus giving us the desired conclusion.
    As this is a \emph{free} output, we know that $\BN{\alpha} = \EMPTYSET$; so, no new name will be added to the type environment and $\Gamma' = \Gamma$.
\end{itemize}

For the inductive step, we reason on the last rule of
Figure~\ref{wc:fig:epi_early_semantics} used to conclude the inference.
Note that well-typedness for \TAU-labels always holds, and well-typedness for output labels is ensured by  Lemma~\ref{wc:lemma:epi_simple_type_welltyped_labels}; so, it only needs to be assumed for input labels.

\begin{itemize}
  \item Suppose \nameref{epi_early_par1} was used. 
    Then the conclusion is of the form $P_1 \PAR P_2 \trans[\alpha] P_1' \PAR P_2$, and from the premise and side condition we have that $P_1 \trans[\alpha] P_1'$ and $\FRESH{\BN{\alpha}}P_2$. 
    Now $\Gamma \vdash P_1 \PAR P_2$ must have been concluded by \nameref{epi_simple_type_par}, and from the premise we have that $\Gamma \vdash P_1$ and $\Gamma \vdash P_2$, 
    We then have by induction hypothesis that $\Gamma' \vdash P_1'$.
    Now, if $\alpha = \LSEND{\VEC{x}}{\NEW{\VEC{z} : \VEC{B}}\VEC{v}}$, then $\Gamma' = \Gamma, \VEC{z} : \VEC{B}$, and we know from the side condition of \nameref{epi_early_par1} that $\FRESH{\VEC{z}}P_2$.
    Thus, by Lemma~\ref{wc:lemma:epi_simple_type_weakening} (weakening), we also have that $\Gamma' \vdash P_2$.
    Otherwise, $\Gamma' = \Gamma$.
    In either case, we can then conclude $\Gamma' \vdash P_1' \PAR P_2$ by \nameref{epi_simple_type_par}.
    The symmetric case for \RULENAME{Par$_2$} is similar.

  \item Suppose \nameref{epi_early_sum} was used.
    Then the conclusion is of the form $\sum [\VEC{e}]\VEC{P} \trans[\ALPHA] P_i'$, and from the premise and side condition we know that $P_i \trans P_i'$ and $e_i \trans_e \TRUE$.
    Then $\Gamma \vdash \sum [\VEC{e}]\VEC{P}$ must have been concluded by \nameref{epi_simple_type_sum}.
    From the premise, we know that $\Gamma \vdash e_i: \code{bool}$ and $\Gamma \vdash P_i$ for each $i \in 1, \ldots, |\VEC{e}|$.
    The former will have been concluded by the type rules for operations and \nameref{epi_simple_type_val}.
    We then have by induction hypothesis that $\Gamma' \vdash P_i'$, where either $\Gamma' = \Gamma$ or $\Gamma' = \Gamma, \VEC{z} : \VEC{B}$ depending on the label.

  \item Suppose \nameref{epi_early_res} was used.
    Then the conclusion is of the form $\NEW{\VEC{x} : \VEC{B}}P \trans[\alpha] \NEW{\VEC{x} : \VEC{B}}P'$, and from the premise and side condition we know that $P \trans[\alpha] P'$ and $\FRESH{\VEC{x}}\alpha$.
    Now $\Gamma \vdash \NEW{\VEC{x} : \VEC{B}}P$ must have been concluded by \nameref{epi_simple_type_res}.
    From the premise, we know that $\Gamma, \VEC{x} : \VEC{B} \vdash P$.
    Therefore, by the induction hypothesis, we have that $\Gamma', \VEC{x} : \VEC{B} \vdash P'$, where $\Gamma' = \Gamma$ or $\Gamma' = \Gamma, \VEC{z} : \VEC{B}'$ as determined by the label.
    By Lemma~\ref{wc:lemma:epi_simple_type_strengthening} (strengthening), we can therefore conclude that $\Gamma' \vdash \NEW{\VEC{x} : \VEC{B}}P'$.

  \item Suppose \nameref{epi_early_rep} was used.
    Then the conclusion is of the form $\REPL{P} \trans[\alpha] P'$, and from the premise we know that $\REPL{P} \PAR P \trans[\alpha] P'$.
    Now $\Gamma \vdash \REPL{P}$ must have been concluded by \nameref{epi_simple_type_rep}, and from its premise we know that $\Gamma \vdash P$.
    Thus, by rule \nameref{epi_simple_type_par}, we have that $\Gamma \vdash \REPL{P} \PAR P$. By induction, we conclude that $\Gamma' \vdash P'$, where either $\Gamma' = \Gamma$ or $\Gamma' = \Gamma, \VEC{z} : \VEC{B}$ depending on the label.

  \item Suppose \nameref{epi_early_open} was used.
    Then the conclusion is of the form $\NEW{\VEC{z} : \VEC{B}_z}P \trans[\LSEND{\VEC{x}}{\NEW{\VEC{y},\VEC{z} : \VEC{B}_y, \VEC{B}_z}\VEC{v}}] P'$, and from the premise we know that $P \trans[\LSEND{\VEC{x}}{\NEW{\VEC{y} : \VEC{B}_y}\VEC{v}}] P'$.
    Now $\Gamma \vdash \NEW{\VEC{z} : \VEC{B}_z}P$ must have been concluded by \nameref{epi_simple_type_res}, and from its premise we know that $\Gamma, \VEC{z} : \VEC{B}_z \vdash P$.
    Then by induction hypothesis, we have that $\Gamma, \VEC{y} : \VEC{B}_y, \VEC{z} : \VEC{B}_z \vdash P'$, where $\Gamma, \VEC{y} : \VEC{B}_y, \VEC{z} : \VEC{B}_z = \Gamma'$.

  \item Suppose \nameref{epi_early_com1} was used.
    Then the conclusion is of the form $P_1 \PAR P_2 \trans[\TAU] \NEW{\VEC{z} : \VEC{B}}\PAREN{P_1' \PAR P_2'}$, and from the premises and side condition we know that $P_1 \trans[\LSEND{x}{\NEW{\VEC{z} : \VEC{B}}\VEC{v}}] P_1'$, $P_2 \trans[\LRECV{\VEC{x}}{\VEC{v}}] P_2'$ and $\FRESH{\VEC{z}}P_2$. 
    Now $\Gamma \vdash P_1 \PAR P_2$ must have been concluded by \nameref{epi_simple_type_par}, and from its premises we know that $\Gamma \vdash P_1$ and $\Gamma \vdash P_2$.
    We now have two cases to examine:
    \begin{enumerate}
      \item If the list of bound names $\VEC{z}$ is empty, then $P_1$ performs a \emph{free} output, and no new name is added to the type environment.
        Then $\Gamma' = \Gamma$, and thus by the induction hypothesis we have that $\Gamma \vdash P_1'$.

        In the case of $P_2$, we know it performs an input, so here $\Gamma' = \Gamma$, and thus by the induction hypothesis we have that $\Gamma \vdash P_2'$.
        Likewise, the restriction in the conclusion is empty, so the conclusion simplifies to $\NEW{\epsilon}\PAREN{P_1' \PAR P_2'} = P_1' \PAR P_2'$.
        Then $\Gamma \vdash P_1' \PAR P_2'$ can be concluded by \nameref{epi_simple_type_par}.

      \item Otherwise, $P_1$ performs a \emph{bound} output with label $\alpha_1 = \LSEND{x}{\NEW{\VEC{z} : \VEC{B}}\VEC{v}}$ of the names $\VEC{z}$, and thus $\Gamma' = \Gamma, \VEC{z} : \VEC{B}$.
        Hence by the induction hypothesis, $\Gamma, \VEC{z} : \VEC{B} \vdash P_1'$.

        For $P_2$, we know from the side condition of \nameref{epi_early_com1} that $\FRESH{\VEC{z}}P_2$.
        Thus, by Lemma~\ref{wc:lemma:epi_simple_type_weakening} (weakening), we can also conclude $\Gamma, \VEC{z} : \VEC{B} \vdash P_2$.
        As we know that $P_2$ performs an input with label $\alpha_2 = \LRECV{\VEC{x}}{\VEC{v}}$, we also know that the type environment does not change.
        
        We must now show that $\Gamma, \VEC{z} : \VEC{B} \vdash \alpha_2$.
        As we know that $\Gamma \vdash P_1$ and $P_1 \trans[\LSEND{x}{\NEW{\VEC{z} : \VEC{B}}\VEC{v}}] P_1'$, we have by Lemma~\ref{wc:lemma:epi_simple_type_welltyped_labels} that $\Gamma \vdash \LSEND{x}{\NEW{\VEC{z} : \VEC{B}}\VEC{v}}$, which must have been concluded by \nameref{epi_simple_type_label_send}.
        From the premise, we have that $\Gamma \vdash \VEC{x} : \TCHAN{\VEC{B}'}$ and $\Gamma, \VEC{z} : \VEC{B} \vdash \VEC{v} : \VEC{B}'$.
        By \nameref{epi_simple_type_label_recv} we can then conclude that $\Gamma, \VEC{z} : \VEC{B} \vdash \LRECV{\VEC{x}}{\VEC{v}}$.
        
        Now by the induction hypothesis, we have that $\Gamma, \VEC{z} : \VEC{B} \vdash P_2'$.
        By \nameref{epi_simple_type_par} we can then conclude $\Gamma, \VEC{z} : \VEC{B} \vdash P_1' \PAR P_2'$, and finally by \nameref{epi_simple_type_res} we can conclude $\Gamma \vdash \NEW{\VEC{z} : \VEC{B}}\PAREN{P_1' \PAR P_2'}$.
    \end{enumerate}
    The symmetric case for \RULENAME{Com$_2$} is similar.
\end{itemize}

This concludes the proof.
\end{proof}

Then, by combining Theorems~\ref{wc:thm:epi_simple_type_safety} and~\ref{wc:thm:epi_simple_subject_reduction}, we obtain that a well-typed process remains now-safe after any number of \TAU-transitions:
\begin{corollary}[Soundness]
If $\Gamma \vdash P$, then $\SAFE{P}$.
\end{corollary}

\section{The \WCLANG{} language}\label{wc:sec:wc_language}
The intuition that a single-name input $\INPUT{x}{\VEC{y}}.P$ corresponds to a (non-persistent) function definition, and that an output $\OUTPUT{x}{\VEC{e}}$ corresponds to a function call, is well-known from Milner's \PI-calculus encoding of the \LAMBDA-calculus \cite{milner1992functions}, and also from Pierce and Turner's work on the PICT language \cite{PICT, turner1996phd}.
Also, the encoding of object-oriented (OO) languages into (variants of the) \PI-calculus is by now well established \cite{KS02,Sangiorgi98,Walker95}; hence, it is not surprising that polyadic synchronisations would make the latter task even easier. 
What is less obvious is what kind of typing discipline corresponds to the type system we have presented in Section~\ref{wc:sec:epi_type_system}.
To answer this question, we shall define \WCLANG{} (\textbf{W}hile with \textbf{C}lasses), an object-oriented, imperative language similar to \code{While} \cite{nielson_nielson2007semantics_with_applications}, but extended with classes, fields and methods, and propose an encoding into \EPI.
We will prove that the type system presented here for \EPI\ exactly corresponds to the ``expectable'' type system one would devise for \WCLANG.

\begin{figure}[t]
\begin{center}
\begin{minipage}{.5\textwidth}
\begin{syntax}[h]
  DC \IS \epsilon \OR \code{class $A$ \{ $DF$ $DM$ \} } DC     \tabularnewline
  DF \IS \epsilon \OR \code{field $p$ := $v$; } DF             \tabularnewline
  DM \IS \epsilon \OR \code{method $f(\VEC{x})$ \{ $S$ \} } DM \tabularnewline
  e  \IS v \kern-1pt\OR x \OR \code{this} \OR \code{$e$.$p$} \OR \op(\VEC{e}) 
\end{syntax}
\end{minipage}
\begin{minipage}{.45\textwidth}
\begin{syntax}[h]
  S \IS \code{skip} 
    \OR \code{var $x$ := $e$ in $S$} 
  \ISOR \code{$x$ := $e$} 
    \OR \code{this.$p$ := $e$} 
  \ISOR \code{$S_1$;$S_2$}
    \OR \code{if $e$ then $S_{\TRUE}$ else $S_{\FALSE}$ } 
  \ISOR \code{while $e$ do $S$} 
    \OR \WCCALL{e}{f}{\VEC{e}}
\end{syntax}
\end{minipage}
\end{center}

\noindent where $x, y \in \VNAMES$ (variable names), $p, q \in \FNAMES$ (field names), $A, B \in \CNAMES$ (class names), $f, g \in \MNAMES$ (method names), and $v \in \mathbb{Z} \UNION \mathbb{B} \UNION \CNAMES$ (values). 
Throughout the paper, we assume, without loss of generality, that all field names are distinct from each other (within a class), and likewise for all method names.
\caption{The syntax of \WCLANG.}
\label{wc:fig:syntax_wc}
\end{figure}

\subsection{Syntax and (big-step) operational semantics}
The syntax of \WCLANG{} is given in Figure~\ref{wc:fig:syntax_wc}, where, for simplicity, we assume that all name sets are pairwise disjoint, and that field and methods names are unique within each class declaration. 

To provide a suitable operational semantics, we shall need some environments to record the bindings of classes, methods, fields, and local variables, including the `magic variable' \code{this}.
We define them as sets of partial functions as follows:

\begin{definition}[Binding model]
We define the following sets of partial functions:
\begin{align*}
  \ENV{V} \in \SETENV{V} & \DCLSYM \VNAMES \UNION \SET{\code{this}} \PARTIAL \VALUES      \\
  \ENV{F} \in \SETENV{F} & \DCLSYM \FNAMES \PARTIAL \VALUES                               \\ 
  \ENV{M} \in \SETENV{M} & \DCLSYM \MNAMES \PARTIAL \VNAMES^* \times \STM                 \\
  \ENV{S} \in \SETENV{S} & \DCLSYM \CNAMES \PARTIAL \SETENV{F}                            \\
  \ENV{T} \in \SETENV{T} & \DCLSYM \CNAMES \PARTIAL \SETENV{M}
\end{align*}
We regard each environment $\ENV{X}$, for any $X \in \SET{V, F, M, S, T}$, as a list of tuples $(d, c), \ENV{X}'$ for some $d \in \DOM{\ENV{X}}$ and $c \in \CODOM{\ENV{X}}$.
The notation $\ENV{X}\UPDATE{d}{c}$ denotes an update of $\ENV{X}$.
We write $\ENV{X}^{\EMPTYSET}$ for the empty environment.
\end{definition}

\begin{figure}[t]
\begin{center}
\begin{semantics}
  \RULE[wc-dclf$_1$]
    { \CONF{DF, \ENV{F}} \trans_{DF} \ENV{F}' }
    { \CONF{\code{field $p$ := $v$;$DF$}, \ENV{F}} \trans_{DF} (p, v), \ENV{F}' }

  \RULE[wc-dclf$_2$]
    { }
    { \CONF{\epsilon, \ENV{F}} \trans_{DF} \ENV{F} }

  \RULE[wc-dclm$_1$]
    { \CONF{DM, \ENV{M}} \trans_{DM} \ENV{M}' }
    { \CONF{\code{method $f$($\VEC{x}$) \{ $S$ \} $DM$}, \ENV{M}} \trans_{DM} (f, (\VEC{x}, S)), \ENV{M}' }

  \RULE[wc-dclm$_2$]
    { }
    { \CONF{\epsilon, \ENV{M}} \trans_{DM} \ENV{M} }

  \RULE[wc-dclc$_1$]
    {
        \CONF{DF, \ENV{F}^{\EMPTYSET}} \trans_{DF} \ENV{F} \AND
        \CONF{DM, \ENV{M}^{\EMPTYSET}} \trans_{DM} \ENV{M} \AND
        \CONF{DC, \ENV{ST}} \trans_{DC} \ENV{ST}'  
    }
    { \CONF{\code{class $A$ \{ $DF$ $DM$ \} $DC$}, \ENV{ST}} \trans_{DC} ((A, \ENV{F}), \ENV{S}'), ((A, \ENV{M}), \ENV{T}') }

  \RULE[wc-dclc$_2$]
    { }
    { \CONF{\epsilon, \ENV{ST}} \trans_{DC} \ENV{ST} }
\end{semantics}
\end{center}
\caption{Semantics of \WCLANG{} declarations.}
\label{wc:fig:wc_semantics_dcl}
\end{figure}

When two or more environments appear together, we shall use the convention of writing e.g.\@ $\ENV{ST}$ instead of $\ENV{S}, \ENV{T}$ to simplify the notation.

Our binding model then consists of two environments:
A \emph{method table} $\ENV{T}$, which maps class names to method environments, such that we for each class can retrieve the list of methods it declares; and a \emph{state} $\ENV{S}$, which maps class names to lists of fields and their values.
The method table will be constant, once all declarations are performed, but the state will change during the evaluation of a program.
The semantics of declarations concern the initial construction of the field and method environments, $\ENV{F}$ and $\ENV{M}$, and the state and method table $\ENV{S}$ and $\ENV{T}$.
We give the semantics in classic big-step style. 
Transitions are thus on the form $\CONF{DX, \ENV{X}} \trans_{DX} \ENV{X}'$ for $X \in \SET{F, M, C, ST}$.
The rules are given in Figure~\ref{wc:fig:wc_semantics_dcl}.

In practice, we shall assume that class declarations are simply in terms of the two environments $\ENV{ST}$, representing the methods and fields declared for each class. 
Thus, for every declared class $A$, we have that $\ENV{T}(A)(f) = (\VEC{x}, S)$ and $\ENV{S}(A)(p) = v$, for every method $f$ (with formal parameters $\VEC x$ and body $S$) and field $p$ (with current value $v$) declared in $A$.
Furthermore, we shall use another environment, $\ENV{V}$, to store the bindings of local variables $x$, i.e.\@ such that $\ENV{V}(x) = v$ for every local variable $x$ with current value $v$.

\begin{figure}[t]
\begin{center}
\begin{semantics}
  \RULE[wc-val][wc_expr_val]
    { }
    { \ENV{SV} \vdash v \trans_e v }

  \RULE[wc-var][wc_expr_var]({ \ENV{V}(x) = v })
    { }
    { \ENV{SV} \vdash x \trans_e v }

  \RULE[wc-op][wc_expr_op]
    { \ENV{SV} \vdash \VEC{e} \trans_e \VEC{v} \AND \op(\VEC{v}) \trans_{\op} v}
    { \ENV{SV} \vdash \op(\VEC{e}) \trans_e v }

  \RULE[wc-field][wc_expr_field]({ \ENV{S}(A)(p) = v })
    { \ENV{SV} \vdash e \trans_e A }
    { \ENV{SV} \vdash e.p \trans_e v }
\end{semantics}
\end{center}
\caption{Semantics of \WCLANG{} expressions.}
\label{wc:fig:wc_semantics_expr}
\end{figure}

Figure~\ref{wc:fig:wc_semantics_expr} gives the semantics of expressions $e$.
Expressions have no side effects, so they cannot contain method calls, but they can access both local variables and fields of any class.
Thus transitions are of the form $\ENV{SV} \vdash e \trans_e v$, i.e.\@ transitions must be concluded relative to the state and variable environments.
We do not give explicit rules for the boolean and integer operators subsumed under $\op$, but simply assume that they can be evaluated to a single value by some semantics $\op(\VEC{v}) \trans_{\op} v$.
Note that we assume that \emph{no} operations can \emph{yield} a class name $A$, so we disallow any form of pointer arithmetic.

\begin{figure}[t]
\begin{center}
\begin{semantics}
  \RULE[wc-skip][wc_stm_skip]
    {}
    { \ENV{T} \vdash \CONF{\code{skip}, \ENV{SV}} \trans \ENV{SV} }

  \RULE[wc-seq][wc_stm_seq]
    { \ENV{T} \vdash \CONF{S_1, \ENV{SV}} \trans \ENV{SV}'' \AND \ENV{T} \vdash \CONF{S_2, \ENV{SV}''} \trans \ENV{SV}' }
    { \ENV{T} \vdash \CONF{\code{$S_1$;$S_2$}, \ENV{SV}} \trans \ENV{SV}' }

  \RULE[wc-if][wc_stm_if]({b \in \SET{\TRUE, \FALSE}})
    { \ENV{SV} \vdash e \trans_e b \AND \ENV{T} \vdash \CONF{S_b, \ENV{SV}} \trans \ENV{SV}' }
    { \ENV{T} \vdash \CONF{\code{if $e$ then $S_{\TRUE}$ else $S_{\FALSE}$ }, \ENV{SV}} \trans \ENV{SV}' }

  \RULE[wc-while$_{\TRUE}$][wc_stm_whiletrue]
    {
    \begin{array}{r @{~} l}
      \ENV{SV} & \vdash e \trans_e \TRUE \tabularnewline
      \ENV{T}  & \vdash \CONF{S, \ENV{SV}} \trans \ENV{SV}'' \tabularnewline
      \ENV{T}  & \vdash \CONF{\code{while $e$ do $S$}, \ENV{SV}''} \trans \ENV{SV}'
    \end{array}
    }
    { \ENV{T} \vdash \CONF{\code{while $e$ do $S$}, \ENV{SV}} \trans \ENV{SV}' }

  \RULE[wc-while$_{\FALSE}$][wc_stm_whilefalse]
    { \ENV{SV} \vdash e \trans_e \FALSE }
    { \ENV{T} \vdash \CONF{\code{while $e$ do $S$}, \ENV{SV}} \trans \ENV{SV} }

   \RULE[wc-decv][wc_stm_decv](x \notin \DOM{\ENV{V}})
    { \ENV{SV} \vdash e \trans_e v \AND \ENV{T} \vdash \CONF{S, \ENV{S}, ((x, v), \ENV{V})} \trans \ENV{S}', ((x, v'), \ENV{V}') }
    { \ENV{T} \vdash \CONF{\code{var $x$ := $e$ in $S$}, \ENV{SV}} \trans \ENV{SV}' } 

  \RULE[wc-assv][wc_stm_assv](x \in \DOM{\ENV{V}})
    { \ENV{SV} \vdash e \trans_e v }
    { \ENV{T} \vdash \CONF{\code{$x$ := $e$}, \ENV{SV}} \trans \ENV{S}, \ENV{V}\UPDATE{x}{v} }

  \RULE[wc-assf][wc_stm_assf](p \in \DOM{\ENV{F}})
    { \ENV{SV} \vdash e \trans_e v \AND \ENV{V}(\code{this}) = A \AND \ENV{S}(A) = \ENV{F} }
    { \ENV{T} \vdash \CONF{\code{this.$p$ := $e$}, \ENV{SV}} \trans \ENV{S}\UPDATE{A}{\ENV{F}\UPDATE{p}{v}}, \ENV{V} }

  \RULE[wc-call][wc_stm_call]
    { 
    \begin{array}{lll}
      \ENV{SV} \vdash e \trans_e A \qquad
      \ENV{T}(A)(f) = (\VEC{x}, S)
      &
      \ENV{SV} \vdash \VEC{e} \trans_e \VEC{v} \qquad
      |\VEC{x}| = |\VEC{v}| = k                \tabularnewline
      \ENV{V}' = (\code{this}, A), (x_1, v_1), \ldots, (x_k, v_k)
      &
      \ENV{T} \vdash \CONF{S, \ENV{S},\ENV{V}'} \trans \ENV{SV}''
    \end{array}
    }
    { \ENV{T} \vdash \CONF{\WCCALL{e}{f}{\VEC{e}}, \ENV{SV}} \trans \ENV{S}'', \ENV{V} }
\end{semantics}
\end{center}
\caption{Semantics of \WCLANG{} statements.}
\label{wc:fig:wc_semantics_stm}
\end{figure}

A program in \WCLANG{} consists of the class definitions, recorded in the environments $\ENV{ST}$, plus a single method call.
The semantics of statements describes the actual execution steps of a program.
In Figure~\ref{wc:fig:wc_semantics_stm} we give the semantics in big-step style, where a step describes the execution of a statement in its entirety.
Statements can read from the method table, and they can modify the variable and property bindings, and hence the state.
The result of executing a statement is a new state, so transitions must here be of the form 
\begin{equation*}
  \ENV{T} \vdash \CONF{S, \ENV{SV}} \trans \ENV{SV}'
\end{equation*}

\noindent with $\ENV{SV}' = \ENV{S}', \ENV{V}'$ denoting that both the property values in $\ENV{S}$ and the values of the local variables in $\ENV{V}$ \emph{may} have been modified.

\subsection{Typed \WCLANG}
There are also textbook examples of type systems for \code{While}-like languages; for example the language \code{BUMP} given in \cite[pp.\@ 185-198]{huttel2010transitions}, which also includes procedure calls.
We can define a similar type system for \WCLANG{} as follows:

\begin{definition}[\WCLANG-types]
We use the following language of types: 
\begin{center}
\begin{minipage}{.45\textwidth}
\begin{syntax}[h]
  IC \IS \epsilon \OR \code{interface $I$ \{ $IF$ $IM$ \} $IC$} \\
  IF \IS \epsilon \OR \code{field $p$ :~$B$; $IF$}               \\
  IM \IS \epsilon \OR \code{method $f$ :~$\TPROC{\VEC{B}}$; $IM$}
\end{syntax}
\end{minipage}
\begin{minipage}{.45\textwidth}
\begin{syntax}[h]
  B \IS I \OR \TINT \OR \TBOOL                    \\
  T \IS B \kern-3pt\OR \TPROC{\VEC{B}} \OR \Delta \\
 \Gamma, \Delta \IS \NAMES \PARTIAL T
\end{syntax}
\end{minipage}
\end{center}
\noindent where $I \in \TNAMES$ is a type name and $\NAMES \DCLSYM \CNAMES \UNION \FNAMES \UNION \VNAMES \UNION \MNAMES \UNION \TNAMES$.
\end{definition}

We use a simple `interface definition language' analogous to the class definitions, to specify the signatures of fields and methods.
We shall henceforth use a typed version of the syntax, where each class is annotated with a type name 
$I$ (viz., \code{class $A$:$I$ \{ $DF$ $DM$ \}}); a type environment $\Gamma$ then is such that $\Gamma(A) = I$ and $\Gamma(I) = \Delta$, where $\Delta$ is again a type environment containing the signatures of fields and methods listed in the interface definition for $I$; thus, different classes can implement the same interface.
Finally, also local variables are now declared with a type, i.e.\@ \code{var $B$ x := $e$ in $S$}.

\begin{figure}[h]\centering
\begin{semantics}
  \RULE[t-env$_T$][wc_type_env_envt]
    { \Gamma \vdash_A \ENV{M} \AND \Gamma \vdash \ENV{T} }
    { \Gamma \vdash (A, \ENV{M}) , \ENV{T} }

  \RULE[t-env$_S$][wc_type_env_envs]
    { \Gamma \vdash_A \ENV{F} \AND \Gamma \vdash \ENV{S} }
    { \Gamma \vdash (A, \ENV{F}) , \ENV{S} }

  \RULE[t-env$_V$][wc_type_env_envv]
    { \Gamma(x) = B \AND \Gamma \vdash v : B \AND \Gamma \vdash \ENV{V} }
    { \Gamma \vdash (x, v) , \ENV{V} }
\end{semantics}
\begin{semantics}
  \RULE[t-env$^{\EMPTYSET}$][wc_type_env_empty]
    { X \in \SET{T, S, V, M, F} }
    { \Gamma \vdash \ENV{X}^{\EMPTYSET} }

  \RULE[t-env$_M$][wc_type_env_envm]
    { 
    \begin{array}{ll}
      \Gamma(A) = I
      &
      \Gamma(I)(f) = \TPROC{\VEC{B}}     \\
      \Gamma, \VEC{x} : \VEC{B} \vdash S 
      &
      \Gamma \vdash_A \ENV{M}              
    \end{array}
    }
    { \Gamma \vdash_A (f, (\VEC{x}, S)) , \ENV{M}}

  \RULE[t-env$_F$][wc_type_env_envf]
    { 
    \begin{array}{ll}
      \Gamma(A) = I 
      & \Gamma(I)(p) = B     \\
      \Gamma \vdash v : B   &
      \Gamma \vdash_A \ENV{F}
    \end{array}
    }
    { \Gamma \vdash_A (p, v) , \ENV{F}}
\end{semantics}
\caption{Type rules for \WCLANG-environment agreement.}
\label{wc:fig:wc_typerules_env}
\end{figure}

\begin{figure}[h]\centering
\begin{semantics}
  \RULE[t-var][wc_type_expr_var]
    { \Gamma(x) = B }
    { \Gamma \vdash x : B }

  \RULE[t-field][wc_type_expr_field]
    { \Gamma \vdash e : I \AND \Gamma(I)(p) = B }
    { \Gamma \vdash e.p : B }
\end{semantics}
\begin{semantics}
  \RULE[t-op][wc_type_expr_op]
    { \Gamma \vdash \VEC{e} : \VEC{B} \AND \op : \VEC{B} \to B }
    { \Gamma \vdash \op(\VEC{e}) : B }

  \RULE[t-val][wc_type_expr_val]({
  B = %
  \begin{cases}
      \TINT     & \text{if $v \in \mathbb{Z}$} \tabularnewline
      \TBOOL    & \text{if $v \in \mathbb{B}$} \tabularnewline
      \Gamma(v) & \text{if $v \in \NAMES$}
  \end{cases}
  }) 
    { }
    { \Gamma \vdash v : B }
\end{semantics}
\caption{Type rules for \WCLANG-expressions}
\label{wc:fig:wc_typerules_expr}
\end{figure}

\begin{figure}[h]\centering
\begin{semantics}
  \RULE[t-skip][wc_type_stm_skip]
    { }
    { \Gamma \vdash \code{skip} }

  \RULE[t-assv][wc_type_stm_assv]
    { 
      \Gamma \vdash x : B \AND
      \Gamma \vdash e : B 
    }
    { \Gamma \vdash \code{$x$ := $e$} }

  \RULE[t-seq][wc_type_stm_seq]
    { 
      \Gamma \vdash S_1 \AND 
      \Gamma \vdash S_2 
    }
    { \Gamma \vdash \code{$S_1$;$S_2$} }

  \RULE[t-while][wc_type_stm_while]
    { 
      \Gamma \vdash e : \TBOOL \AND 
      \Gamma \vdash S 
    }
    { \Gamma \vdash \code{while $e$ do $S$} }
\end{semantics}
\begin{semantics}
  \RULE[t-decv][wc_type_stm_decv]
    { 
      \Gamma \vdash e : B \AND 
      \Gamma, x : B \vdash S 
    }
    { \Gamma \vdash \code{var $B$ $x$ := $e$ in $S$} }

  \RULE[t-assf][wc_type_stm_assf]
    { 
      \Gamma \vdash \code{this.$p$} : B \AND 
      \Gamma \vdash e : B 
    }
    { \Gamma \vdash \code{this.$p$ := $e$} }

  \RULE[t-call][wc_type_stm_call]
    { 
      \Gamma \vdash e : I             \AND
      \Gamma(I)(f) = \TPROC{\VEC{B}}  \AND
      \Gamma \vdash \VEC{e} : \VEC{B}  
    }
    { \Gamma \vdash \WCCALL{e}{f}{\VEC{e}} }

  \RULE[t-if][wc_type_stm_if]
    { 
      \Gamma \vdash e : \TBOOL \AND 
      \Gamma \vdash S_{\TRUE}  \AND 
      \Gamma \vdash S_{\FALSE} 
    }
    { \Gamma \vdash \code{if $e$ then $S_{\TRUE}$ else $S_{\FALSE}$} }

\end{semantics}
\caption{Type rules for \WCLANG-statements.}
\label{wc:fig:wc_typerules_stm}
\end{figure}

The type rules for \WCLANG{} are given in Figure~\ref{wc:fig:wc_typerules_env} (environment agreement), Figure~\ref{wc:fig:wc_typerules_expr} (expressions), and Figure~\ref{wc:fig:wc_typerules_stm}.
Note that we continue to use the abbreviated notation for environments, and thus write e.g.\@ $\Gamma \vdash \ENV{SV}$ for $\Gamma \vdash \ENV{S} \land \Gamma \vdash \ENV{V}$.

Well-typedness of statements ensures that: types are preserved in assignments (see rules \nameref{wc_type_stm_decv} and \nameref{wc_type_stm_assf}), methods are called with the correct number and types of arguments (rule \nameref{wc_type_stm_call}), and classes contain the accessed members (rules \nameref{wc_type_expr_field} and \nameref{wc_type_stm_call}).
We shall forgo the definition of an explicit \SAFE{\cdot} predicate, since it is obvious.
We can then show the following theorem, which assures us that the type system is sound: 

\begin{theorem}[Subject reduction]\label{wc:thm:wc_subject_reduction}
Let $\Gamma \vdash \ENV{TSV}$.
If $\Gamma \vdash S$ and $\ENV{T} \vdash \CONF{S, \ENV{SV}} \trans \ENV{SV}'$, then $\Gamma \vdash \ENV{SV}'$.
\end{theorem}

This is the `expectable' result for a simple type system for a language with a big-step semantics, i.e.\@ that a well-typed program, if it terminates when given a well-typed `input' (in the form of the environments $\ENV{SV}$), will produce a well-typed output (environments $\ENV{SV}'$).

The proof of this theorem relies on some standard lemmata, that can be proved by straightforward inductions.
Firstly, weakening and strengthening express that we can add and remove unused type assumptions in $\Gamma$ when typing $\ENV{V}$:
\begin{lemma}[Weakening for $\ENV{V}$]\label{wc:lemma:wc_weakening_envv}
If $\Gamma \vdash \ENV{V}$ and $x \notin \DOM{\ENV{V}}$, then $\Gamma, x:B \vdash \ENV{V}$.
\end{lemma}

\begin{lemma}[Strengthening for $\ENV{V}$]\label{wc:lemma:wc_strengthening_envv}
If $\Gamma, x:B \vdash \ENV{V}$ and $x \notin \DOM{\ENV{V}}$, then $\Gamma \vdash \ENV{V}$.
\end{lemma}

The next lemma expresses that, if an expression $e$ is judged to have type $B$ and $e$ evaluates to some value $v$, then $v$ is indeed of type $B$.
This is shown by induction on the rules for semantics of expressions (Figure~\ref{wc:fig:wc_semantics_expr}), and the rules for evaluating operations $\op$, which we have omitted.

\begin{lemma}[Safety for expressions]\label{wc:lemma:wc_safety_expressions}
Let $\Gamma \vdash \ENV{S}$ and $\Gamma \vdash \ENV{V}$.
If $\Gamma \vdash e : B$ and $\ENV{SV} \vdash e \trans_e v$, then $\Gamma \vdash v : B$.
\end{lemma}

The next two lemmata state that we can update an entry $(x,v)$ in $\ENV{V}$ with another value $v'$, if $v$ and $v'$ are of the same type, and that we can extend $\ENV{V}$ with a new entry $(x,v)$, as long as $x$ and $v$ have the same type:
\begin{lemma}[Substitution for $\ENV{V}$]\label{wc:lemma:wc_substitution_envv}
If $\Gamma, x:B \vdash \ENV{V}$ and $x \in \DOM{\ENV{V}}$ and $\Gamma, x:B \vdash v : B$, \\then $\Gamma, x:B \vdash \ENV{V}\UPDATE{x}{v}$.
\end{lemma}

\begin{lemma}[Extension of $\ENV{V}$]\label{wc:lemma:wc_extension_envv}
If $\Gamma, x:B \vdash \ENV{V}$ and $x \notin \DOM{\ENV{V}}$ and $\Gamma, x:B \vdash v : B$, \\then $\Gamma, x:B \vdash (x,v), \ENV{V}$.
\end{lemma}

We shall need almost the same lemmata for $\ENV{S}$, except extension, since new fields cannot be declared at runtime.
As $\ENV{S}$ contains nested environments $\ENV{F}$, we shall use the notation $\FRESH{x}\ENV{S}$ to say that $x$ does not occur in $\ENV{S}$, i.e.\@ neither in its domain, or in the domain of any of its nested $\ENV{F}$ environments (nor, in principle, as a value).

\begin{lemma}[Weakening for $\ENV{S}$]\label{wc:lemma:wc_weakening_envs}
If $\Gamma \vdash \ENV{S}$ and $\FRESH{x}\ENV{S}$, then $\Gamma, x:B \vdash \ENV{S}$.
\end{lemma}

\begin{lemma}[Strengthening for $\ENV{S}$]\label{wc:lemma:wc_strengthening_envs}
If $\Gamma, x:B \vdash \ENV{S}$ and $\FRESH{x}\ENV{S}$, then $\Gamma \vdash \ENV{S}$.
\end{lemma}

\begin{lemma}[Substitution for $\ENV{S}$]\label{wc:lemma:wc_substitution_envs}
Assume $A \in \DOM{\ENV{S}}$ and $\ENV{S}(A) = \ENV{F}$.
If $\Gamma \vdash \ENV{S}$ \\ and $\Gamma \vdash A.p : B$ and $\Gamma \vdash v:B$, then $\Gamma \vdash \ENV{S}\UPDATE{A}{\ENV{F}\UPDATE{p}{v}}$.
\end{lemma}

Finally we can prove the Subject Reduction theorem:
\begin{proof}[Proof of Theorem \ref{wc:thm:wc_subject_reduction}]
By induction on the inference for $\ENV{T} \vdash \CONF{S, \ENV{SV}} \trans \ENV{SV}'$.
We have four possible base cases:
\begin{itemize}
  \item Suppose the transition was concluded by \nameref{wc_stm_skip}.
    Then we know that $S$ = \code{skip}, and the transition is of the form
    \begin{equation*}
      \ENV{T} \vdash \CONF{\code{skip}, \ENV{SV}} \trans \ENV{SV}
    \end{equation*}
    As neither $\ENV{S}$ nor $\ENV{V}$ are modified by the transition, we have that $\Gamma \vdash \ENV{SV}$ by assumption.

  \item Suppose \nameref{wc_stm_whilefalse} was used.
    Then we know that $S$ = \code{while $e$ do $S$}, and the transition is of the form
    \begin{equation*}
      \ENV{T} \vdash \CONF{\code{while $e$ do $S$}, \ENV{SV}} \trans \ENV{SV}
    \end{equation*}
    As neither $\ENV{S}$ nor $\ENV{V}$ are modified in the transition, we have that $\Gamma \vdash \ENV{SV}$ by assumption.

    \item Suppose \nameref{wc_stm_assv} was used.
    Then we know that $S$ = \code{$x$ := $e$}, and the transition is of the form 
    \begin{equation*}
      \ENV{T} \vdash \CONF{\code{$x$ := $e$}, \ENV{SV}} \trans \ENV{S}, \ENV{V}\UPDATE{x}{v}
    \end{equation*}

    From the premise of that rule we have that $\ENV{SV} \vdash e \trans_e v$. 
    Now, $\Gamma \vdash$ \code{$x$ := $e$} must have been concluded by \nameref{wc_type_stm_assv}, and from the premise of that rule we have that $\Gamma \vdash x :B$ and $\Gamma \vdash e : B$.
    We can then conclude the following:
    \begin{align*}
      \Gamma & \vdash v : B                 & \text{by Lemma~\ref{wc:lemma:wc_safety_expressions} (expressions safety),} \\
      \Gamma & \vdash \ENV{V}\UPDATE{x}{v} & \text{by Lemma~\ref{wc:lemma:wc_substitution_envv} (substitution for $\ENV{V}$),}
    \end{align*}

    \noindent and $\Gamma \vdash \ENV{S}$ by assumption, since it is not modified by the transition.
    Thus $\Gamma \vdash \ENV{S}$ and $\Gamma \vdash \ENV{V}\UPDATE{x}{v}$ as desired.

    \item Suppose \nameref{wc_stm_assf} was used.
    Then we know that $S$ = \code{this.$p$ := $e$}, and the transition is of the form
    \begin{equation*}
      \ENV{T} \vdash \CONF{\code{this.$p$ := $e$}, \ENV{SV}} \trans \ENV{S}\UPDATE{A}{\ENV{F}\UPDATE{p}{v}}, \ENV{V} 
    \end{equation*}

    From the premise of that rule we have that $\ENV{SV}\vdash e \trans_e v$,
    $\ENV{V}(\code{this}) = A$, and
    $\ENV{S}(A) = \ENV{F}$.
    Now, $\Gamma \vdash $ \code{this.$p$ := $e$} must have been concluded by \nameref{wc_type_stm_assf}, and from the premise of that rule we have that $\Gamma \vdash \code{this.$p$} : B$ and $\Gamma \vdash e : B$.
    We can then conclude that
    \begin{align*}
      \Gamma & \vdash v : B                                     & \text{by Lemma~\ref{wc:lemma:wc_safety_expressions} (expressions safety),} \\
      \Gamma & \vdash \ENV{S}\UPDATE{A}{\ENV{F}\UPDATE{p}{v}} & \text{by Lemma~\ref{wc:lemma:wc_substitution_envs} (substitution for $\ENV{S}$),}
    \end{align*}
    and $\Gamma \vdash \ENV{V}$ by assumption, since it is not modified by the transition.
    Thus $\Gamma \vdash \ENV{V}$ and $\Gamma \vdash \ENV{S}\UPDATE{A}{\ENV{F}\UPDATE{p}{v}}$ as desired.
\end{itemize}

For the inductive step, we reason on the rule of Figure~\ref{wc:fig:wc_semantics_stm} used for concluding the transition.
Note that, when stating the induction hypothesis in the cases below, we shall omit some of the premises for brevity, if they hold by assumption; especially $\Gamma \vdash \ENV{T}$ and $\ENV{T} \vdash \CONF{S, \ENV{SV}} \trans \ENV{SV}'$.

\begin{itemize}
  \item Suppose \nameref{wc_stm_seq} was used.
    Then we know that $S = S_1;S_2$, and the transition is of the form
    \begin{equation*}
      \ENV{T} \vdash \CONF{\code{$S_1$; $S_2$}, \ENV{SV}} \trans \ENV{SV}'
    \end{equation*}
    From the premise of that rule we have that $\ENV{T} \vdash \CONF{S_1, \ENV{SV}} \trans \ENV{SV}''$ and $\ENV{T} \vdash \CONF{S_2, \ENV{SV}''} \trans \ENV{SV}'$.
    Now, $\Gamma \vdash S_1;S_2$ must have been concluded by \nameref{wc_type_stm_seq}, and from the premise of that rule we have that $\Gamma \vdash S_1$ and $\Gamma \vdash S_2$.
    By the induction hypothesis we can then conclude that $\Gamma \vdash \ENV{SV}''$ and $\Gamma \vdash \ENV{SV}'$, as desired.

  \item Suppose \nameref{wc_stm_if} was used.
    Then we know that $S$ = \code{if $e$ then $S_{\TRUE}$ else $S_{\FALSE}$}, and the transition is of the form
    \begin{equation*}
      \ENV{T} \vdash \CONF{\code{if $e$ then $S_{\TRUE}$ else $S_{\FALSE}$ }, \ENV{SV}} \trans \ENV{SV}'
    \end{equation*}
    From the premise of that rule we have that $\ENV{SV} \vdash e \trans_e b$ and $\ENV{T} \vdash \CONF{S_b, \ENV{SV}} \trans \ENV{SV}'$.
    Now, $\Gamma \vdash $ \code{if $e$ then $S_{\TRUE}$ else $S_{\FALSE}$} must have been concluded by \nameref{wc_type_stm_if}, and from the premise of that rule we have that $\Gamma \vdash S_{\TRUE}$ and $\Gamma \vdash S_{\FALSE}$.
    By the induction hypothesis we can then conclude that $\Gamma \vdash \ENV{SV}'$, as desired.

  \item Suppose \nameref{wc_stm_whiletrue} was used.
    Then we know that $S$ = \code{while $e$ do $S$}, and the transition is of the form
    \begin{equation*}
      \ENV{T} \vdash \CONF{\code{while $e$ do $S$}, \ENV{SV}} \trans \ENV{SV}'
    \end{equation*}
    From the premise of that rule we have that $\ENV{T} \vdash \CONF{S, \ENV{SV}} \trans \ENV{SV}''$ and $\ENV{T} \vdash \CONF{\code{while $e$ do $S$}, \ENV{SV}''} \trans \ENV{SV}'$.
    Now, $\Gamma \vdash $ \code{while $e$ do $S$} must have been concluded by \nameref{wc_type_stm_while}, and from the premise of that rule we have that $\Gamma \vdash S$.
    As we know the transition was in fact concluded, we also know that the derivation tree for the transition is of finite height, and we can therefore use the induction hypothesis.
    Thus, by the induction hypothesis we can conclude that $\Gamma \vdash \ENV{SV}''$ and $\Gamma \vdash \ENV{SV}'$, as desired.

    \item Suppose \nameref{wc_stm_decv} was used.
    Then we know that $S$ = \code{var $x$ := $e$ in $S$}, and the transition is of the form
    \begin{equation*}
      \ENV{T} \vdash \CONF{\code{var $x$ := $e$ in $S$}, \ENV{SV}} \trans \ENV{SV}'
    \end{equation*}
    From the premise of that rule we have that $\ENV{SV} \vdash e \trans_e v$ and $\ENV{T} \vdash \CONF{S, \ENV{S}, ((x, v), \ENV{V})} \trans \ENV{S}', ((x, v'), \ENV{V}')$.
    Now, $\Gamma \vdash $ \code{var $x$ := $e$ in $S$} must have been concluded by \nameref{wc_type_stm_decv}, and from the premise of that rule we have that $\Gamma \vdash e : B$ and $\Gamma, x : B \vdash S$.
    By Lemma~\ref{wc:lemma:wc_safety_expressions} (expressions safety) we conclude that $\Gamma \vdash v : B$ and by Lemma~\ref{wc:lemma:wc_extension_envv} we conclude that $\Gamma, x:B \vdash (x,v), \ENV{V}$.
    As $\FRESH{x}\ENV{S}$, we therefore also have that $\Gamma, x:B \vdash \ENV{S}$, by Lemma~\ref{wc:lemma:wc_weakening_envs}.
    By the induction hypothesis we can then conclude that $\Gamma, x:B \vdash \ENV{S}'$ and $\Gamma, x:B \vdash (x,v'), \ENV{V}'$.
    As we know that $\FRESH{x}\ENV{S}'$ and $x \notin \DOM{\ENV{V}'}$, we can apply Lemmas~\ref{wc:lemma:wc_strengthening_envs} and~\ref{wc:lemma:wc_strengthening_envv}  to conclude that $\Gamma \vdash \ENV{S}'$ and $\Gamma \vdash \ENV{V}'$ as desired.

    \item Suppose \nameref{wc_stm_call} was used.
    Then we know that $S$ = \WCCALL{e}{f}{\VEC{e}}, and the transition is of the form
    \begin{equation*}
      \ENV{T} \vdash \CONF{\WCCALL{e}{f}{\VEC{e}}, \ENV{SV}} \trans \ENV{S}', \ENV{V}
    \end{equation*}

    From the premise of that rule we have that
    \begin{align*}
      \ENV{SV}  & \vdash e \trans_e A             \\
      \ENV{SV}  & \vdash \VEC{e} \trans_e \VEC{v} \\
      \ENV{T}   & \vdash \CONF{S, \ENV{S},\ENV{V}''} \trans \ENV{SV}'
    \end{align*}
    where $\ENV{T}(A)(f) = (\VEC{x}, S)$ and $\ENV{V}'' = (\code{this}, A), (x_1, v_1), \ldots, (x_k, v_k), \ENV{V}^{\EMPTYSET}$ for some $k = |\VEC{e}| = |\VEC{x}| = |\VEC{v}|$.
    Now, $\Gamma \vdash \WCCALL{e}{f}{\VEC{e}}$ must have been concluded by \nameref{wc_type_stm_call}, and from the premise of that rule we have that $\Gamma \vdash e : I$ and $\Gamma(I)(f) = \TPROC{\VEC{B}}$ and $\Gamma \vdash \VEC{e} : \VEC{B}$.

    By assumption $\Gamma \vdash \ENV{T}$, which was concluded by rules \nameref{wc_type_env_envt} and \nameref{wc_type_env_envm}.
    From the premise of the latter rule, we get that $\Gamma', \VEC{x} : \VEC{B} \vdash S$ where $\Gamma' = \Gamma, \code{this} : I$ and $\Gamma'(I)(f) = \TPROC{\VEC{B}}$.

    Now let $\Gamma' = \Gamma, \code{this} : I, \VEC{x} : \VEC{B}$.
    We can then conclude the following:
    \begin{align*}
      \Gamma   & \vdash A : I                & \text{by Lemma~\ref{wc:lemma:wc_safety_expressions} (expressions safety),} \\
      \Gamma   & \vdash \VEC{v} : \VEC{B}    & \text{by Lemma~\ref{wc:lemma:wc_safety_expressions} (expressions safety),} \\
      \Gamma   & \vdash \ENV{V}^{\EMPTYSET}  & \text{by \nameref{wc_type_env_empty},}                                     \\
      \Gamma'' & \vdash \ENV{V}''            & \text{by Lemma~\ref{wc:lemma:wc_extension_envv} (extension of $\ENV{V}$) and the preceding,} \\
      \Gamma'' & \vdash \ENV{S}              & \text{by Lemma~\ref{wc:lemma:wc_weakening_envs} (weakening of $\ENV{S}$),} \\
      \Gamma'' & \vdash \ENV{S}'             & \text{by the induction hypothesis.}
    \end{align*}

    Note that the $\ENV{V}'$ obtained from the execution of $S$ is discarded in the conclusion of \nameref{wc_stm_call}, so we do not need to infer its well-typedness.
    It is instead replaced with the original $\ENV{V}$, which is unmodified in the transition, and by assumption $\Gamma \vdash \ENV{V}$.
    Finally, as we know that $\code{this} \notin \DOM{\ENV{S}}$ and $\FRESH{\VEC{x}}\ENV{S}$, we can apply Lemma~\ref{wc:lemma:wc_strengthening_envs} (repeatedly) to conclude that $\Gamma \vdash \ENV{S}'$.
    Thus we have that $\Gamma \vdash \ENV{V}$ and $\Gamma \vdash \ENV{S}'$ as desired.
\end{itemize}

This concludes the proof.
\end{proof}

\section{Encoding \WCLANG{} in \EPI}\label{wc:sec:encoding}
We shall now see how to represent the \WCLANG-language in the \EPI-calculus.
In \WCLANG, we already assume that all field names are distinct from each other (within a class), and likewise for all method names.
Our strategy is to let all class names $A$, field names $p$, variable names $x$ (including \code{this}), and method names $f$, be names drawn from the set of \EPI{} names $\NAMES$, and then create composite names to scope the methods and fields, such that a field $p$ in a class $A$ will get the composite name $A \cdot p$; and likewise, a method $f$ will now get the composite name $A \cdot f$.

To encode imperative variables (and fields) and the memory store (and state) represented by the environments $\ENV{SV}$, we shall use a technique proposed in \cite{hirschkoff2020references_picalc}, where they show how a reference $\ell$ (i.e., a location/address) can be represented in the \PI-calculus simply by means of an asynchronous output:
\begin{equation*}
  \LOCNEW{\ell:B}{e}P  \DEFSYM \NEW{\ell:B}    \PAREN{ \OUTPUT{\ell}{e} \PAR P } \quad
  \LOCREAD{\ell}{y}.P  \DEFSYM \INPUT{\ell}{y}.\PAREN{ \OUTPUT{\ell}{y} \PAR P } \quad
  \LOCWRITE{\ell}{e}.P \DEFSYM \INPUT{\ell}{x}. \PAREN{ \OUTPUT{\ell}{e} \PAR P } 
\end{equation*}

\noindent where $x \not\in fn(P,e)$.
Here, 
$\LOCNEW{\ell:B}{e}P$ declares a new reference $\ell$ of type $B$ in $P$, which is initialised with the expression $e$; then,
$\LOCREAD{\ell}{y}.P$ reads the contents of $\ell$ (non-destructively) and binds it to $y$ within $P$; and
$\LOCWRITE{\ell}{e}.P$ writes the expression $e$ into $\ell$, overwriting the previous value, and continues as $P$.

In \cite{hirschkoff2020references_picalc}, $\ell$ is just a single \PI-calculus name, but as we are using this representation in \EPI, the reference will be a \emph{vector} of names.
Thus, we can encode both local variables and fields, where, in the latter case, $\ell$ will be the vector $A \cdot p$, representing the field named $p$ in the class $A$.
Note however that, unlike local variables (which are private), fields are publicly visible; therefore, their declaration must be represented as a plain output $\OUTPUT{A \cdot p}{e}$, i.e.\@ without restricting them.
This is safe, since we assume that the field names and method names are unique within each class, so the encoding of fields cannot clash with the encoding of method calls.

In the type system for \WCLANG, there is no distinction between a value of type $B$ and a variable (or field) capable of \emph{containing} such a value; however, given our representation of variables as outputs in \EPI{}, we must now explicitly distinguish between the two.
A value $v$ will have a base type $B$ as in $\WCLANG$, but variables (and fields) \emph{storing} such a value will now be given a `container type' $\TIFC{B} \mapsto (\TCHAN{B}, \EMPTYSET)$.
Thus, a variable storing e.g.\@ a boolean value in \WCLANG{} can now be given the type $\TIFC{\TBOOL} \mapsto (\TCHAN{\TBOOL}, \EMPTYSET)$, matching the encoding in \cite{hirschkoff2020references_picalc} of a channel capable of communicating a single boolean value.
We shall assume one such type is defined for each base type $B$ used in the \WCLANG-program to be translated.
Furthermore, we shall need one other type, $\TRET \mapsto (\TCHAN{}, \EMPTYSET)$, called the \emph{return type}. 
As shown below, it denotes a name that can only be used for pure synchronisation signals; we use such names in the encoding to control the execution of sequential compositions and to signal the return from method calls.

The encoding must also be parameterised with a type environment $\Gamma$ of \WCLANG-types, containing the interface definitions for the program; this is necessary because we are using a typed syntax, and the encoding will introduce some auxiliary names, which must therefore be given a type.
Note, however, that the encoding itself will not depend on $\Gamma$: had we used an untyped syntax, it could have been omitted.

Finally, recall that the semantics of \WCLANG{} is given in terms of transitions of the form $\ENV{T} \vdash \CONF{S, \ENV{SV}} \trans \ENV{SV}'$, since $\ENV{T}$ is never modified during transitions.
However, for the purpose of the encoding, we shall slightly change this format  and instead write $\CONF{S, \ENV{TSV}} \trans \ENV{TSV}'$, since all elements of the configuration, including $\ENV{T}$, must be encoded.
Likewise, we shall write transitions for expressions as $\CONF{e, \ENV{SV}} \trans_e \CONF{v, \ENV{SV}}$.
The top-level call to the translation function for initial and final configurations is then:
\begin{align*}
  \PTRANS[\CONF{S, \ENV{TSV}}]<\Gamma> & = \PTRANS[\ENV{T}]<\Gamma> \PAR \PTRANS[\ENV{S}] \PAR \NEW{r : \TRET}\PAREN{ \PTRANS[\ENV{V}]<\Gamma>\HOLE[{\PTRANS[S](r)<\Gamma>}] \PAR \INPUT{r}{}.\NIL } \\
  \PTRANS[\ENV{TSV}]<\Gamma>           & = \PTRANS[\ENV{T}]<\Gamma> \PAR \PTRANS[\ENV{S}] \PAR \PTRANS[\ENV{V}]<\Gamma>\HOLE[\NIL]
\end{align*}

\noindent where the notation $\PTRANS[\ENV{V}]<\Gamma>\HOLE[{\PTRANS[S](r)<\Gamma>}]$ (resp.\@ $\PTRANS[\ENV{V}]<\Gamma>\HOLE[\NIL]$) indicates that the variable environment $\ENV{V}$ is translated as a \emph{process context} containing a single hole, written $\HOLE$, into which the translation of the statement $\PTRANS[S](r)<\Gamma>$ (resp.\@ the stopped process $\NIL$) is inserted.
We also declare a new name $r$, i.e.\@ the \emph{return signal}, which is passed as a parameter to the encoding of $S$, and await an input on this name, which signals that the program has finished.
The two other environments are translated as processes running in parallel; their encoding is given below:

\bigskip
\noindent
\begin{math}
\begin{array}{r @{\qquad} r @{~} l}
  \PTRANS[\ENV{F}^\EMPTYSET](A)         = \NIL      &  \PTRANS[(p,v),\ENV{F}](A)                                   & = \PTRANS[\ENV{F}](A) \PAR \OUTPUT{A \cdot p}{v}            \vspace*{.2cm} \\
  \PTRANS[\ENV{M}^\EMPTYSET](A)<\Gamma> = \NIL      &  \PTRANS[{\PAREN[big]{f, (\VEC{x}, S)}, \ENV{M}}](A)<\Gamma> & = \PTRANS[\ENV{M}](A)<\Gamma> \PAR \REPL{\INPUT{A \cdot f}{r, \VEC{a}}} 
      . \LOCNEW{\VEC{x} : \TIFC{\VEC{B}}}{\VEC{a}} 
        \LOCNEW{\code{this} : \TIFC{\Gamma(A)}}{A} \PAREN{ \PTRANS[S](r)<\Gamma> } \\
                                                    &                                                              & \quad \text{where $\Gamma(A) = I_A$ and $\Gamma(I_A)(f) = \TPROC{\VEC{B}}$}\vspace*{.2cm} \\
  \PTRANS[\ENV{T}^\EMPTYSET]<\Gamma>    = \NIL      &  \PTRANS[(A, \ENV{M}),\ENV{T}]<\Gamma>                       & = \PTRANS[\ENV{T}]<\Gamma> \PAR \PTRANS[\ENV{M}](A)<\Gamma> 
\end{array}
\end{math}

\noindent 
\begin{math}
\begin{array}{r @{\qquad} r @{~} l}
  \PTRANS[\ENV{S}^\EMPTYSET]            = \NIL      &  \PTRANS[(A, \ENV{F}), \ENV{S}]                              & = \PTRANS[\ENV{S}] \PAR \PTRANS[\ENV{F}](A)                 \vspace*{.2cm} \\
  \PTRANS[\ENV{V}^{\EMPTYSET}]<\Gamma>  = \HOLE[~]  &
  \PTRANS[(x,v), \ENV{V}]<\Gamma>                   & = \LOCNEW{x :\TIFC{\Gamma(x)}}{v}\PAREN{ \PTRANS[\ENV{V}]<\Gamma> }             
\end{array}
\end{math}

\bigskip
There are a few points to note:
First, the local variable bindings in $\ENV{V}$ are translated as a process context consisting of a series of declarations and with the empty environment translated as the hole $\HOLE$;
thus, the declarations bind the translated statement $\PTRANS[S](r)<\Gamma>$, which is inserted into the hole.
Second, the encoding of the environments containing the field and method declarations (i.e., $\ENV{F}$ and $\ENV{M}$) are both parameterised with the name of the current class, $A$, which is used to compose the subject vector of the outputs.
Third, a method declaration $(f, (\VEC{x}, S))$ is translated as an input-guarded replication $\REPL{\INPUT{A \cdot f}{r, \VEC{a}}}$ to make it persistent.
The actual parameters will be received and bound to $\VEC{a}$, and these, in turn, are then bound to the formal parameters $\VEC{x}$ in a list of declarations $\LOCNEW{\VEC{x}:\VEC{B}}{\VEC{a}}$.
This extra re-binding step is necessary, since the formal parameters should act as \emph{variables} (rather than values) within the body, so they too must follow the protocol of the encoding in \cite{hirschkoff2020references_picalc}.
Likewise, we also add a binding for the special variable \code{this}, which is available within the body $S$, so it can be used to access the fields of the current class.
Finally, besides the formal parameters, we also add an extra name $r$ to be received by the input $\INPUT{A \cdot f}{r, \VEC{a}}$; this is the return signal, on which the method will signal when it finishes.

Statements are encoded as follows:
\begin{align*}
  \PTRANS[\code{skip}](r)<\Gamma>                        & = \OUTPUT{r}{}                                                                                                                          \\
  \PTRANS[\code{var $B$ $x$\,:=\,$e$ in $S$}](r)<\Gamma> & = \NEW{z:\TIFC{B}}\PAREN{ \PTRANS[e](z)<\Gamma> \PAR \INPUT{z}{y} . \LOCNEW{x:\TIFC{B}}{y}\PAREN{ \PTRANS[S](r)<\Gamma, x:B> } }        \\
  \PTRANS[\code{$x$\,:=\,$e$}](r)<\Gamma>                & = \NEW{z:\TIFC{B}}\PAREN{ \PTRANS[e](z)<\Gamma> \PAR \INPUT{z}{y} . \LOCWRITE{x}{y} . \OUTPUT{r}{} }                      \hspace*{3cm} \text{where} ~ \Gamma \vdash e : B \\
  \PTRANS[\code{this.$p$\,:=\,$e$}](r)<\Gamma>           & = \NEW{z:\TIFC{B}}\PAREN{ \PTRANS[e](z)<\Gamma> \PAR \INPUT{z}{y} . \LOCREAD{\code{this}}{Y} . \LOCWRITE{Y \cdot p}{y} . \OUTPUT{r}{} } 
  \hspace*{.65cm}\text{where} ~ \Gamma \vdash e : B
\\
  \PTRANS[\code{$S_1$;$S_2$}](r_2)<\Gamma>                         & = \NEW{r_1: \TRET} \PAREN{ \PTRANS[S_1](r_1)<\Gamma> \PAR \INPUT{r_1}{}.\PTRANS[S_2](r_2)<\Gamma> }                                                     \\
  \PTRANS[\code{if $e$ then $S_\TRUE$ else $S_\FALSE$}](r)<\Gamma> & = \NEW{z:\TIFC{\TBOOL}}\PAREN{ \PTRANS[e](z)<\Gamma> \PAR \INPUT{z}{y}. \IFTHENELSE{y}{ \PTRANS[S_\TRUE](r)<\Gamma> }{ \PTRANS[S_\FALSE](r)<\Gamma> } } \\
  \PTRANS[\code{while $e$ do $S$}](r)<\Gamma>                      & = \NEW{r':\TRET}\PAREN[Big]{ \OUTPUT{r'}{} \PAR \REPL{\INPUT{r'}{}}                                                                                     %
                                                                     . \NEW{z:\TIFC{\TBOOL}}\PAREN{ \PTRANS[e](z)<\Gamma> \PAR \INPUT{z}{y}. \IFTHENELSE{y}{ \PTRANS[S](r')<\Gamma> }{ \OUTPUT{r}{} } } }                    \\
  \PTRANS[\WCCALL{e}{f}{\VEC{e}}](r)<\Gamma>                       & = \NEW{a:\TIFC{{I_A}}, \VEC{z}:\TIFC{{\VEC{B}}}}\PAREN[Big]{ \PTRANS[e](a)<\Gamma> \PAR \PTRANS[\VEC{e}](\VEC{z})<\Gamma>  \PAR \INPUT{a}{Y}            %
                                                                     . \INPUT{z_1}{y_1} \ldots \INPUT{z_n}{y_n} . \OUTPUT{Y \cdot f}{ r, y_1, \ldots, y_n } }                                                                \\
                                                                   & \qquad \text{where} ~ \Gamma \vdash e : I_A \text{ and } \Gamma \vdash \VEC{e} : \VEC{B} \text{ and } |\VEC{e}| = |\VEC{z}| = n              \\
                                                                   & \qquad \text{and} ~ \PTRANS[\VEC{e}](\VEC{z})<\Gamma> = \PTRANS[e_1](z_1)<\Gamma> \PAR \ldots \PAR \PTRANS[e_n](z_n)<\Gamma>
\end{align*}

Most of this encoding is straightforward: \code{skip} does nothing, so it simply emits the return signal.
In a sequential composition \code{$S_1$;$S_2$}, we first declare a new return signal $r'$, which is passed to the first component; then, we wait a synchronisation on this name before executing the second component, which then will emit a signal on $r$.
The \code{if-then-else} construct maps directly to our syntactic sugar for guarded choice, with both branches receiving $r$, since only one of them is chosen.
We also use this in the encoding of \code{while}, which is just a guarded replication of an \code{if-then-else}: the true-branch will emit a return signal on a fresh name $r'$, which will trigger another replication; the false-branch will emit a return signal on the outer return name $r$.
Lastly, in method calls, $r$ is passed as the first parameter.

Note that, in each construct containing an expression $e$, we translate $e$ as a separate process, rather than mapping it directly to an expression in \EPI{}.
This, again, is a consequence of our use of the encoding from \cite{hirschkoff2020references_picalc}: since expressions in \WCLANG{} can contain variables, these must be read in accordance with the protocol.
The translation function for expressions is parametrised with a name $z$, on which the value resulting from evaluating the expression will be delivered.
It is defined as follows:
\begingroup\allowdisplaybreaks
\begin{align*}
  \PTRANS[v](z)<\Gamma>                     & = \OUTPUT{z}{v}                                                                                                                    \\
  \PTRANS[x](z)<\Gamma>                     & = \LOCREAD{x}{y} . \OUTPUT{z}{y}                                                                                                   \\ 
  \PTRANS[e.p](z)<\Gamma>                   & = \NEW{z' : \TIFC{I_A}}\PAREN{ \PTRANS[e](z')<\Gamma> \PAR \INPUT{z'}{Y} . \LOCREAD{Y \cdot p}{y} . \OUTPUT{z}{y} }                \quad \text{where} ~ \Gamma \vdash e : I_A \\
  \PTRANS[\op(e_1, \ldots, e_n)](z)<\Gamma> & = \NEW{z_1 : \TIFC{{B_1}}, \ldots, z_n : \TIFC{{B_n}}}\PAREN{ \PTRANS[e_1](z_1)<\Gamma> \PAR \ldots \PAR \PTRANS[e_n](z_n)<\Gamma> %
                                                \PAR \INPUT{z_1}{y_1} \ldots \INPUT{z_n}{y_n} . \OUTPUT{z}{\op(y_1, \ldots, y_n)} }                                              \\
                                            & \qquad \text{where} ~ \Gamma \vdash e_1 : B_1, \ldots, \Gamma \vdash e_n : B_n                                                     
\end{align*}
\endgroup

Finally we can give the translation of the types and the type environment $\Gamma$.
This is complicated by the fact that we now need to add a container type $\TIFC{B}$ for each variable and field of type $B$, and for each method of type $\TPROC{\VEC{B}}$.
We use a multi-level encoding, and we assume for the sake of simplicity that duplicate entries in the output are ignored:
\begin{center}
\begin{math}
\begin{array}{r @{~} l}
                            &                                                                           \\
  \PTRANS[x:B, \Gamma]      & = x:\TIFC{B}, \TIFC{B}:(\TCHAN{B}, \EMPTYSET), \PTRANS[\Gamma]            \\
  \PTRANS[A:I, \Gamma]      & = A:I, \PTRANS[\Gamma]                                                    \\
  \PTRANS[I:\Delta, \Gamma] & = I:(\TNIL, (\PTRANS[\Delta]<2>)), \PTRANS[\Delta]<3>, \PTRANS[\Gamma]    \\
  \PTRANS[\epsilon]         & = \epsilon                                                                \\
                            &                                                                           \\
\end{array}\hspace*{1cm}
\begin{array}{r @{~} l}
  \PTRANS[p:B, \Delta]<2>               & = p:\TIFC{B}, \PTRANS[\Delta]<2>       \\
  \PTRANS[f:\TPROC{\VEC{B}},\Delta]<2>  & = f:\TIFC{\VEC{B}}, \PTRANS[\Delta]<2> \\
  \PTRANS[\epsilon]<2>                  & = \epsilon                             \\
  \PTRANS[p:B, \Delta]<3>               & = \TIFC{B}:(\TCHAN{B}, \EMPTYSET), \PTRANS[\Delta]<3>                     \\
  \PTRANS[f:\TPROC{\VEC{B}},\Delta]<3>  & = \TIFC{\VEC{B}}:(\TCHAN{\TRET, \VEC{B}}, \EMPTYSET), \PTRANS[\Delta]<3>  \\
  \PTRANS[\epsilon]<3>                  & = \epsilon
\end{array}
\end{math}
\end{center}

Our main result now states that the two type systems exactly correspond:
well-typed \WCLANG-programs are mapped to well-typed \EPI-processes, and vice versa for ill-typed programs.
Thus we know that the \EPI-type system exactly captures the notion of well-typedness in \WCLANG.

\begin{theorem}[Type correspondence]\label{wc:thm:type_correspondence}
$\Gamma \vdash \ENV{TSV}$ and $\Gamma \vdash S$ iff $\PTRANS[\Gamma] \vdash \PTRANS[\CONF{S, \ENV{TSV}}]<\Gamma>$.
\end{theorem}

We break the proof down into a series of lemmata:
\ghostsubsection{Lemma: Values type correspondence}
\begin{lemma}\label{wc:lemma:type_correspondence_val}
$\Gamma \vdash v : B \iff \PTRANS[\Gamma] \vdash v : B$.
\end{lemma}

\begin{proof}
By case analysis of the type of $v$.

For the forward direction:
\begin{itemize}
  \item If $v \in \mathbb{Z}$ then $\Gamma \vdash v : \TINT$ by \nameref{wc_type_expr_val}.
    We can then conclude $\PTRANS[\Gamma] \vdash v : \TINT$ by \nameref{epi_simple_type_val}, since this does not depend on $\Gamma$.

  \item If $v \in \mathbb{B}$ then $\Gamma \vdash v : \TBOOL$ by \nameref{wc_type_expr_val}.
    We can then conclude $\PTRANS[\Gamma] \vdash v : \TBOOL$ by \nameref{epi_simple_type_val}, since this does not depend on $\Gamma$.

  \item If $v \in \NAMES$ then $\Gamma \vdash v : \Gamma(v)$ by \nameref{wc_type_expr_val}.
    By the syntax of \WCLANG{} (Figure~\ref{wc:fig:syntax_wc}), the only names that can appear as values are $\CNAMES$ (class names), and not method names, field names or variable names.
    Thus, $v$ must be a class name $A$, and $\Gamma(A) = I$, so we can expand the type environment as $\Gamma = A:I, \Gamma'$.

    By the encoding of $\Gamma$, we have that $\PTRANS[A :I, \Gamma'] = A:I, \PTRANS[\Gamma']$, so we can conclude that
    \begin{equation*}
      (A:I, \PTRANS[\Gamma'])(A) = I
    \end{equation*}
\end{itemize}

For the other direction:
\begin{itemize}
  \item If $v \in \mathbb{Z}$ then $\PTRANS[\Gamma] \vdash v : \TINT$ by \nameref{epi_simple_type_val} for any $\PTRANS[\Gamma]$.
    We can then conclude $\Gamma \vdash v : \TINT$ by \nameref{wc_type_expr_val}.

  \item If $v \in \mathbb{B}$ then $\PTRANS[\Gamma] \vdash v : \TBOOL$ by \nameref{epi_simple_type_val} for any $\PTRANS[\Gamma]$. 
    We can then conclude $\Gamma \vdash v : \TBOOL$ by \nameref{wc_type_expr_val}.

  \item If $v \in \NAMES$ then $\PTRANS[\Gamma] \vdash v : {\PTRANS[\Gamma]}(v)$ by \nameref{epi_simple_type_val}.
    Hence, $\PTRANS[\Gamma]$ can be expanded as $\PTRANS[\Gamma] = v:B, \PTRANS[\Gamma']$ for some $B$.
    Since we know that $\Gamma$ is used to type \WCLANG-programs, we also know by the syntax of \WCLANG{} (Figure~\ref{wc:fig:syntax_wc}) that 
    $v$ must be a class name $A$, which must have an interface type $I$, so $\PTRANS[\Gamma] = A:I, \PTRANS[\Gamma']$, and therefore $(A:I, \PTRANS[\Gamma'])(A) = I$.
    
    By the encoding of $\Gamma$, we have that $\PTRANS[A:I, \Gamma'] = A:I, \PTRANS[\Gamma']$, and thus $\Gamma = A:I, \Gamma'$.
    We can then conclude that $(A:I, \Gamma)(A) = I$, as desired.
\end{itemize}
\end{proof}

\ghostsubsection{Lemma: Expressions type correspondence}
\begin{lemma}\label{wc:lemma:type_correspondence_expr}
$\Gamma \vdash e : B \iff \PTRANS[\Gamma], z : \TIFC{B} \vdash \PTRANS[e](z)<\Gamma>$.
\end{lemma}

\begin{proof}\allowdisplaybreaks
We have two directions to prove.

For the forward direction, we must show that $\Gamma \vdash e : B \implies \PTRANS[\Gamma], z : \TIFC{B} \vdash \PTRANS[e](z)<\Gamma>$.
We proceed by induction on the structure of $e$.
\begin{itemize}
  \item Suppose the expression is $v$, so $\Gamma \vdash v : B$.
    By the encoding, $\PTRANS[v](z)<\Gamma> = \OUTPUT{z}{v}$, and we can then conclude $\PTRANS[\Gamma], z : \TIFC{B} \vdash \OUTPUT{z}{v}$ as follows:
    By Lemma~\ref{wc:lemma:type_correspondence_val}, we have that $\PTRANS[\Gamma] \vdash v : B$. Then,  by \nameref{epi_simple_type_vec2}, 
    \begin{equation*}
      (\PTRANS[\Gamma], z:\TIFC{B}); (\PTRANS[\Gamma], z:\TIFC{B}) \vdash z : \TCHAN{B} 
    \end{equation*}

    \noindent since by assumption, type definitions $\TIFC{B} \mapsto (\TCHAN{B}, \EMPTYSET)$ are added to $\PTRANS[\Gamma]$ for all the base types used.
    Finally, we can conclude 
    \begin{equation*}
      \PTRANS[\Gamma], z:\TIFC{B} \vdash \OUTPUT{z}{v}
    \end{equation*}

    \noindent by \nameref{epi_simple_type_out}, as desired.

  \item Suppose the expression is $x$, so $\Gamma \vdash x : B$.
    This must have been concluded by \nameref{wc_type_expr_var}, and from its premise, we know that $\Gamma(x) = B$.
    Hence, we can expand the type environment as $\Gamma = x:B, \Gamma'$.
    Now, by the encoding of expressions:
    \begin{equation*}
        \PTRANS[x]<\Gamma> 
      = \LOCREAD{x}{y} . \OUTPUT{z}{y} 
      = \INPUT{x}{y} . \PAREN{ \OUTPUT{x}{y} \PAR \OUTPUT{z}{y} }
    \end{equation*}

    \noindent and by the encoding for $\Gamma$:
    \begin{equation*}
      \PTRANS[x:B, \Gamma'] = x : \TIFC{B}, \TIFC{B} : (\TCHAN{B}, \EMPTYSET), \PTRANS[\Gamma']
    \end{equation*}

    \noindent so by \nameref{epi_simple_type_vec2} we can then conclude that 
    \begin{align*}
      (x : \TIFC{B}, \TIFC{B} : (\TCHAN{B}, \EMPTYSET), \PTRANS[\Gamma'], z : \TIFC{B}); (x : \TIFC{B}, \TIFC{B} : (\TCHAN{B}, \EMPTYSET), \PTRANS[\Gamma'], z : \TIFC{B}) & \vdash x : \TCHAN{B} \\
      (x : \TIFC{B}, \TIFC{B} : (\TCHAN{B}, \EMPTYSET), \PTRANS[\Gamma'], z : \TIFC{B}); (x : \TIFC{B}, \TIFC{B} : (\TCHAN{B}, \EMPTYSET), \PTRANS[\Gamma'], z : \TIFC{B}) & \vdash z : \TCHAN{B}
    \end{align*}

    Using this, we can conclude
    \begin{align*}
      x : \TIFC{B}, \TIFC{B} : (\TCHAN{B}, \EMPTYSET), \PTRANS[\Gamma'], z : \TIFC{B}, y : B & \vdash \OUTPUT{x}{y}                                             & \text{by \nameref{epi_simple_type_out}} \\
      x : \TIFC{B}, \TIFC{B} : (\TCHAN{B}, \EMPTYSET), \PTRANS[\Gamma'], z : \TIFC{B}, y : B & \vdash \OUTPUT{z}{y}                                             & \text{by \nameref{epi_simple_type_out}} \\
      x : \TIFC{B}, \TIFC{B} : (\TCHAN{B}, \EMPTYSET), \PTRANS[\Gamma'], z : \TIFC{B}, y : B & \vdash \OUTPUT{x}{y} \PAR \OUTPUT{z}{y}                          & \text{by \nameref{epi_simple_type_par}} \\
      x : \TIFC{B}, \TIFC{B} : (\TCHAN{B}, \EMPTYSET), \PTRANS[\Gamma'], z : \TIFC{B}        & \vdash \INPUT{x}{y} . \PAREN{ \OUTPUT{x}{y} \PAR \OUTPUT{z}{y} } & \text{by \nameref{epi_simple_type_in}} 
    \end{align*}

    \noindent as desired.

  \item Suppose the expression is $e.p$, so $\Gamma \vdash e.p : B$.
    This must have been concluded by \nameref{wc_type_expr_field}, and from its premise we know that $\Gamma \vdash e : I_A$ and $\Gamma(I_A)(p) = B$. 
    Thus, we know that $\Gamma$ can be written as $I_A:(p : B, \Delta), \Gamma'$ for some $\Delta$.
    By the encoding of $\Gamma$, we then have that 
    \begin{equation*}
      \PTRANS[I_A:(p : B, \Delta), \Gamma'] = I_A:(\TNIL, (p:\TIFC{B}, \PTRANS[\Delta]<2>)), \TIFC{B}:(\TCHAN{B}, \EMPTYSET), \PTRANS[\Delta]<3>, \PTRANS[\Gamma']
    \end{equation*}

    However, we shall leave $\PTRANS[\Gamma]$ unexpanded in the following for the sake of readability.
    Now, by the encoding of expressions:
    \begin{align*}
      \PTRANS[e.p](z)<\Gamma> & = \NEW{z' : \TIFC{I_A}}\PAREN{ \PTRANS[e](z')<\Gamma> \PAR \INPUT{z'}{Y} . \LOCREAD{Y \cdot p}{y} . \OUTPUT{z}{y} } \\
                              & = \NEW{z' : \TIFC{I_A}}\PAREN{ \PTRANS[e](z')<\Gamma> \PAR \INPUT{z'}{Y} . \INPUT{Y \cdot p}{y}. \PAREN{ \OUTPUT{Y \cdot p}{y} \PAR \OUTPUT{z}{y} } }
    \end{align*}

    \noindent and by the induction hypothesis $\PTRANS[\Gamma], z':\TIFC{I_A} \vdash \PTRANS[e](z')<\Gamma>$.
    Using this, we conclude as follows:
    \begin{align*}
      \PTRANS[\Gamma], z:\TIFC{B}, z':\TIFC{I_A}, y:B        & \vdash \OUTPUT{z}{y}                                                                                                         & \text{by \nameref{epi_simple_type_out}} \\
      \PTRANS[\Gamma], z:\TIFC{B}, z':\TIFC{I_A}, Y:I_A, y:B & \vdash \OUTPUT{Y \cdot p}{y}                                                                                                 & \text{by \nameref{epi_simple_type_out}} \\
      \PTRANS[\Gamma], z:\TIFC{B}, z':\TIFC{I_A}, Y:I_A, y:B & \vdash \OUTPUT{Y \cdot p}{y} \PAR \OUTPUT{z}{y}                                                                              & \text{by \nameref{epi_simple_type_par}} \\
      \PTRANS[\Gamma], z:\TIFC{B}, z':\TIFC{I_A}, Y:I_A      & \vdash \INPUT{Y \cdot p}{y} . \PAREN{ \OUTPUT{Y \cdot p}{y} \PAR \OUTPUT{z}{y} }                                             & \text{by \nameref{epi_simple_type_in}}  \\
      \PTRANS[\Gamma], z:\TIFC{B}, z':\TIFC{I_A}             & \vdash \INPUT{z'}{Y} . \INPUT{Y \cdot p}{y} . \PAREN{ \OUTPUT{Y \cdot p}{y} \PAR \OUTPUT{z}{y} }                             & \text{by \nameref{epi_simple_type_in}}  \\
      \PTRANS[\Gamma], z:\TIFC{B}, z':\TIFC{I_A}             & \vdash \PTRANS[e](z')<\Gamma> \PAR \INPUT{z'}{Y} . \INPUT{Y \cdot p}{y} . \PAREN{ \OUTPUT{Y \cdot p}{y} \PAR \OUTPUT{z}{y} } & \text{by \nameref{epi_simple_type_par}} \\
      \PTRANS[\Gamma], z:\TIFC{B}                            & \vdash \NEW{z' : \TIFC{I_A}} \PAREN{                                                                                         %
                                                                    \PTRANS[e](z')<\Gamma> \PAR \INPUT{z'}{Y} . \INPUT{Y \cdot p}{y} . \PAREN{ \OUTPUT{Y \cdot p}{y} \PAR \OUTPUT{z}{y} } } & \text{by \nameref{epi_simple_type_res}}
    \end{align*}

    \noindent as desired.

  \item Suppose the expression is $\op(e_1, \ldots, e_n)$, so $\Gamma \vdash \op(e_1, \ldots, e_n) : B$.
    This must have been concluded by \nameref{epi_simple_type_op}, and from its premise we know that 
    \begin{align*}
      \Gamma \vdash e_1 : B_1 \quad\ldots\quad
      \Gamma \vdash e_n : B_n
    \end{align*}
    and $\vdash \op : B_1, \ldots, B_n \to B$.
    Now, by the encoding of expressions:
    \begin{equation*}
        \PTRANS[\op(e_1, \ldots, e_n)](z)<\Gamma> 
      = \NEW{z_1 : \TIFC{{B_1}}, \ldots, z_n : \TIFC{{B_n}}}\PAREN{ \PTRANS[e_1](z_1)<\Gamma> \PAR \ldots \PAR \PTRANS[e_n](z_n)<\Gamma> %
          \PAR \INPUT{z_1}{y_1} \ldots \INPUT{z_n}{y_n} . \OUTPUT{z}{\op(y_1, \ldots, y_n)} }                                            \\
    \end{equation*}

    \noindent and by $n$ applications of the induction hypothesis, we have that $\PTRANS[\Gamma], z_1 : \TIFC{B_1} \vdash \PTRANS[e_1](z_1)<\Gamma>$, \ldots, $\PTRANS[\Gamma], z_n : \TIFC{B_n} \vdash \PTRANS[e_n](z_n)<\Gamma>$.
    By Lemma~\ref{wc:lemma:epi_simple_type_weakening}, we can then also conclude that 
    \begin{equation*} 
      \PTRANS[\Gamma], z:\TIFC{B}, z_1 : \TIFC{B_1}, \ldots, z_n : \TIFC{B_n} \vdash \PTRANS[e_i](z_i)<\Gamma>
    \end{equation*}

    \noindent for each $i \in 1..n$, and thus by $n-1$ applications of rule \nameref{epi_simple_type_par}, we can conclude 
    \begin{equation*} 
      \PTRANS[\Gamma], z:\TIFC{B}, z_1 : \TIFC{B_1}, \ldots, z_n : \TIFC{B_n} \vdash \PTRANS[e_1](z_1)<\Gamma> \PAR \ldots \PAR \PTRANS[e_n](z_n)<\Gamma>
    \end{equation*}

    Furthermore, we can conclude
    \begin{equation*} 
      \PTRANS[\Gamma], z:\TIFC{B}, z_1 : \TIFC{B_1}, \ldots, z_n : \TIFC{B_n}, y_1 : B_1, \ldots, y_n : B_n \vdash \OUTPUT{z}{\op(y_1, \ldots, y_n)}
    \end{equation*}

    \noindent by \nameref{epi_simple_type_out}, and then, by $n$ applications of \nameref{epi_simple_type_in}, we can conclude
    \begin{equation*} 
      \PTRANS[\Gamma], z:\TIFC{B}, z_1 : \TIFC{B_1}, \ldots, z_n : \TIFC{B_n} \vdash \INPUT{z_1}{y_1} \ldots \INPUT{z_n}{y_n} . \OUTPUT{z}{\op(y_1, \ldots, y_n)}
    \end{equation*}

    Thus, we can conclude
    \begin{equation*}
      \PTRANS[\Gamma], z:\TIFC{B}, z_1 : \TIFC{B_1}, \ldots, z_n : \TIFC{B_n} \vdash \PTRANS[e_1](z_1)<\Gamma> \PAR \ldots \PAR \PTRANS[e_n](z_n)<\Gamma> \PAR \INPUT{z_1}{y_1} \ldots \INPUT{z_n}{y_n} . \OUTPUT{z}{\op(y_1, \ldots, y_n)}
    \end{equation*}

    \noindent by another application of \nameref{epi_simple_type_par}, and finally 
    \begin{equation*}
       \PTRANS[\Gamma], z:\TIFC{B} \vdash \NEW{z_1 : \TIFC{{B_1}}, \ldots, z_n : \TIFC{{B_n}}}\PAREN{ \PTRANS[e_1](z_1)<\Gamma> \PAR \ldots \PAR \PTRANS[e_n](z_n)<\Gamma> %
                                          \PAR \INPUT{z_1}{y_1} \ldots \INPUT{z_n}{y_n} . \OUTPUT{z}{\op(y_1, \ldots, y_n)} } 
    \end{equation*}

    \noindent by \nameref{epi_simple_type_res}, as desired.
\end{itemize}

This concludes the proof for the forward direction.

For the other direction, we must show that $\PTRANS[\Gamma], z : \TIFC{B} \vdash \PTRANS[e](z)<\Gamma> \implies \Gamma \vdash e : B$.
We proceed again by induction on the structure of $e$:
\begin{itemize}
  \item Suppose the expression is $v$, so $\PTRANS[v](z)<\Gamma> = \OUTPUT{z}{v}$; by \nameref{epi_simple_type_out}, we have that $\PTRANS[\Gamma], z : \TIFC{B} \vdash \OUTPUT{z}{v}$ and, from its premise, we have that 
    $\PTRANS[\Gamma], z:\TIFC{B} \vdash x : \TCHAN{B}$ and $\PTRANS[\Gamma], z:\TIFC{B} \vdash v : B$.
    By Lemma~\ref{wc:lemma:epi_simple_type_strengthening}, we have that  $\PTRANS[\Gamma] \vdash v : B$, since we know that $\FRESH{r}v$, as the name was introduced by the translation.
    Then, by Lemma~\ref{wc:lemma:type_correspondence_val}, $\Gamma \vdash v:B$, as desired.
    
  \item Suppose the expression is $x$, so
    \begin{equation*}
      \PTRANS[x]<\Gamma> = \INPUT{x}{y} . \PAREN{ \OUTPUT{x}{y} \PAR \OUTPUT{z}{y} }
    \end{equation*}
    
    \noindent and 
    \begin{equation*}
      \PTRANS[\Gamma], z:\TIFC{B} \vdash \INPUT{x}{y} . \PAREN{ \OUTPUT{x}{y} \PAR \OUTPUT{z}{y} } 
    \end{equation*}

    \noindent which must have been concluded by \nameref{epi_simple_type_in}, \nameref{epi_simple_type_par} and \nameref{epi_simple_type_out}.
    From the premise of the latter, we get that $\PTRANS[\Gamma]; \PTRANS[\Gamma] \vdash x : \TCHAN{B}$, so it must be the case that ${\PTRANS[\Gamma]}(x) = \TIFC{B}$, and $\PTRANS[\Gamma]$ contains the type assignment $\TIFC{B} \mapsto (\TCHAN{B}, \EMPTYSET)$.
    By the translation of $\Gamma$, we then have that 
    \begin{equation*}
      \PTRANS[x:B, \Gamma'] = x : \TIFC{B}, \TIFC{B} : (\TCHAN{B}, \EMPTYSET), \PTRANS[\Gamma']
    \end{equation*}

    \noindent where $\Gamma = x:B, \Gamma'$, and $x:B, \Gamma' \vdash x:B$ can then be concluded by \nameref{wc_type_expr_var} as desired.

  \item Suppose the expression is $e.p$, so 
    \begin{equation*}
      \PTRANS[e.p](z)<\Gamma> = \NEW{z' : \TIFC{I_A}}\PAREN{ \PTRANS[e](z')<\Gamma> \PAR \INPUT{z'}{Y} . \INPUT{Y \cdot p}{y}. \PAREN{ \OUTPUT{Y \cdot p}{y} \PAR \OUTPUT{z}{y} } } 
    \end{equation*}

    \noindent and
    \begin{equation*}
      \PTRANS[\Gamma], z:\TIFC{B} \vdash \NEW{z' : \TIFC{I_A}}\PAREN{ \PTRANS[e](z')<\Gamma> \PAR \INPUT{z'}{Y} . \INPUT{Y \cdot p}{y}. \PAREN{ \OUTPUT{Y \cdot p}{y} \PAR \OUTPUT{z}{y} } }
    \end{equation*}

    \noindent which must have been concluded by \nameref{epi_simple_type_res}, \nameref{epi_simple_type_par}, \nameref{epi_simple_type_in} and \nameref{epi_simple_type_out}.
    We omit the full derivation here, but see the corresponding case for the forward direction.
    
    Our goal is to infer $\Gamma \vdash e.p : B$ using  \nameref{wc_type_expr_field}. 
    The premises of that rule are $\Gamma \vdash e : I_A$ and $\Gamma(I_A)(p) = B$, so we must show that they both are satisfied.    
    Since one of the premises of \nameref{epi_simple_type_par} is $\PTRANS[\Gamma], z':\TIFC{I_A} \vdash \PTRANS[e](z')<\Gamma>$, by the induction hypothesis we have that $\Gamma \vdash e : I_A$.

    In another branch of the derivation tree, $\PTRANS[\Gamma], Y:I_A, y:B \vdash \OUTPUT{Y \cdot p}{y}$ is concluded by \nameref{epi_simple_type_out} (some unused names are omitted for clarity); its premise requires that
    \begin{equation*}
      (\PTRANS[\Gamma], Y:I_A, y:B); (\PTRANS[\Gamma], Y:I_A, y:B) \vdash \TCHAN{B}
    \end{equation*}
    is concluded by \nameref{epi_simple_type_vec1} and \nameref{epi_simple_type_vec2}.
    Hence, $\PTRANS[\Gamma]$ must contain $I_A \mapsto(\TNIL, (p:\TIFC{B}, \PTRANS[\Delta]<2>))$ and $\TIFC{B} \mapsto (\TCHAN{B}, \EMPTYSET)$.
    Hence, by the encoding of $\Gamma$, we have that 
    \begin{equation*}
      \PTRANS[I_A:(p : B, \Delta), \Gamma'] = I_A:(\TNIL, (p:\TIFC{B}, \PTRANS[\Delta]<2>)), \TIFC{B}:(\TCHAN{B}, \EMPTYSET), \PTRANS[\Delta]<3>, \PTRANS[\Gamma']
    \end{equation*}

    \noindent and clearly $(I_A:(p : B, \Delta), \Gamma')(I_A)(p) = B$.
    Thus, as all premises of \nameref{wc_type_expr_field} are satisfied, we can conclude that $\Gamma \vdash e.p : B$.

  \item Suppose the expression is $\op(e_1, \ldots, e_n)$, so
    \begin{align*}
         & \PTRANS[\op(e_1, \ldots, e_n)](z)<\Gamma> \\ 
      = ~& \NEW{z_1 : \TIFC{{B_1}}, \ldots, z_n : \TIFC{{B_n}}}\PAREN{ \PTRANS[e_1](z_1)<\Gamma> \PAR \ldots \PAR \PTRANS[e_n](z_n)<\Gamma> %
          \PAR \INPUT{z_1}{y_1} \ldots \INPUT{z_n}{y_n} . \OUTPUT{z}{\op(y_1, \ldots, y_n)} } 
    \end{align*}

    \noindent and 
    \begin{equation*}
      \PTRANS[\Gamma], z:\TIFC{B} \vdash \NEW{z_1 : \TIFC{{B_1}}, \ldots, z_n : \TIFC{{B_n}}}\PAREN{ \PTRANS[e_1](z_1)<\Gamma> \PAR \ldots \PAR \PTRANS[e_n](z_n)<\Gamma> %
          \PAR \INPUT{z_1}{y_1} \ldots \INPUT{z_n}{y_n} . \OUTPUT{z}{\op(y_1, \ldots, y_n)} }
    \end{equation*}

    \noindent which must have been concluded by \nameref{epi_simple_type_res}.
    From its premise, we then have that 
    \begin{equation*}
      \PTRANS[\Gamma], z:\TIFC{B}, z_1 : \TIFC{{B_1}}, \ldots, z_n : \TIFC{{B_n}} \vdash  \PTRANS[e_1](z_1)<\Gamma> \PAR \ldots \PAR \PTRANS[e_n](z_n)<\Gamma> %
          \PAR \INPUT{z_1}{y_1} \ldots \INPUT{z_n}{y_n} . \OUTPUT{z}{\op(y_1, \ldots, y_n)}
    \end{equation*}

    \noindent and by $n$ applications of the induction hypothesis (and Lemma~\ref{wc:lemma:epi_simple_type_strengthening} to remove unused names), we have that $\Gamma \vdash e_i : B_i$, for all $i$.
    Furthermore,  
    \begin{equation*}
      \PTRANS[\Gamma], z:\TIFC{B}, z_1 : \TIFC{{B_1}}, \ldots, z_n : \TIFC{{B_n}} \vdash \INPUT{z_1}{y_1} \ldots \INPUT{z_n}{y_n} . \OUTPUT{z}{\op(y_1, \ldots, y_n)}
    \end{equation*}

    \noindent must have been concluded by $n$ applications of \nameref{epi_simple_type_in}, and with 
    \begin{equation*} 
      \PTRANS[\Gamma], z:\TIFC{B}, y_1 : B_1, \ldots, y_n : B_n \vdash \OUTPUT{z}{\op(y_1, \ldots, y_n)}
    \end{equation*}

    \noindent as the premise (we omit the full derivation here, but see the corresponding case for the forward direction).
    Its premise then requires that $\PTRANS[\Gamma], z:\TIFC{B}; \PTRANS[\Gamma]; z:\TIFC{B} \vdash z:\TCHAN{B}$ (which follows from the presence of $z:\TIFC{B}$ in accordance with our use of interface types for basic types) and of $\vdash \op : B_1, \ldots, B_n \to B$.
    Thus we can conclude $\Gamma \vdash \op(e_1, \ldots, e_n) : B$ by \nameref{wc_type_expr_op} as desired.
\end{itemize}

This concludes the proof for the other direction.
\end{proof}

\ghostsubsection{Lemma: Statements type correspondence}
\begin{lemma}\label{wc:lemma:type_correspondence_stm}
$\Gamma \vdash S \iff \PTRANS[\Gamma], r : \TRET \vdash \PTRANS[S](r)<\Gamma>$.
\end{lemma}

\begin{proof}\allowdisplaybreaks
We have two statements to prove.
For both directions, we proceed by induction on the structure of $S$.

For the forward direction, we must show that $\Gamma \vdash S \implies \PTRANS[\Gamma], r : \TRET \vdash \PTRANS[S](r)<\Gamma>$. 
We first recall that $\TRET \mapsto (\TCHAN{}, \EMPTYSET)$ is included in the definition of $\PTRANS[\Gamma]$.
\begin{itemize}
  \item Suppose the statement is \code{skip}, so $\Gamma \vdash \code{skip}$.
    This must have been concluded by \nameref{wc_type_stm_skip}, which is an axiom.
    By the encoding of statements, we have that 
    \begin{equation*}
      \PTRANS[\code{skip}](r)<\Gamma> = \OUTPUT{r}{}
    \end{equation*}

    \noindent and we can therefore conclude $\PTRANS[\Gamma], r : \TRET \vdash \OUTPUT{r}{}$ as desired. 

  \item Suppose the statement is \code{var $B$ $x$ := $e$ in $S$}, so $\Gamma \vdash \code{var $B$ $x$ := $e$ in $S$}$.
    This must have been concluded by \nameref{wc_type_stm_decv}, and from its premise we have that 
      $\Gamma        \vdash e : B$ and 
      $\Gamma, x : B \vdash S$.
    By the encoding of statements, we have that
    \begin{align*}
        ~& \PTRANS[\code{var $B$ $x$\,:=\,$e$ in $S$}](r)<\Gamma> \\
      = ~& \NEW{z:\TIFC{B}}\PAREN{ \PTRANS[e](z)<\Gamma> \PAR \INPUT{z}{y} . \LOCNEW{x:\TIFC{B}}{y}\PAREN{ \PTRANS[S](r)<\Gamma, x:B> } } \\
      = ~& \NEW{z:\TIFC{B}}\PAREN{ \PTRANS[e](z)<\Gamma> \PAR \INPUT{z}{y} . \NEW{x:\TIFC{B}}\PAREN{ \OUTPUT{x}{y} \PAR \PTRANS[S](r)<\Gamma, x:B> } }
    \end{align*}

    By Lemma~\ref{wc:lemma:type_correspondence_expr}, we have that $\PTRANS[\Gamma], z:\TIFC{B} \vdash \PTRANS[e](z)<\Gamma>$ and, by the induction hypothesis, that $\PTRANS[\Gamma, x:B], r:\TRET \vdash \PTRANS[S](r)<\Gamma, x:B>$.
    By the encoding of $\Gamma$, we have that $\PTRANS[\Gamma, x:B] = x:\TIFC{B}, \PTRANS[\Gamma]$ where $\TIFC{B} \mapsto \TIFC{B}:(\TCHAN{B}, \EMPTYSET)$.
    We then conclude as follows:
    \begin{align*}
      \PTRANS[\Gamma], x:\TIFC{B}, r:\TRET, y:B & \vdash \OUTPUT{x}{y}                                                                       & \text{by \nameref{epi_simple_type_out}}                            \\
      \PTRANS[\Gamma], x:\TIFC{B}, r:\TRET, y:B & \vdash \PTRANS[S](r)<\Gamma, x:B>                                                          & \text{by Lemma~\ref{wc:lemma:epi_simple_type_weakening}}              \\
      \PTRANS[\Gamma], x:\TIFC{B}, r:\TRET, y:B & \vdash \OUTPUT{x}{y} \PAR \PTRANS[S](r)<\Gamma, x:B>                                       & \text{by \nameref{epi_simple_type_par}}                            \\
      \PTRANS[\Gamma], r:\TRET, y:B             & \vdash \NEW{x:\TIFC{B}}\PAREN{ \OUTPUT{x}{y} \PAR \PTRANS[S](r)<\Gamma, x:B>}              & \text{by \nameref{epi_simple_type_res}}                            \\
      \PTRANS[\Gamma], r:\TRET, z:\TIFC{B}      & \vdash \INPUT{z}{y}.\NEW{x:\TIFC{B}}\PAREN{ \OUTPUT{x}{y} \PAR \PTRANS[S](r)<\Gamma, x:B>} & \text{by \nameref{epi_simple_type_in} and Lemma~\ref{wc:lemma:epi_simple_type_weakening}} \\
      \PTRANS[\Gamma], r:\TRET, z:\TIFC{B}      & \vdash \PTRANS[e](z)<\Gamma>                                                               & \text{by Lemma~\ref{wc:lemma:epi_simple_type_weakening}}              \\
      \PTRANS[\Gamma], r:\TRET, z:\TIFC{B}      & \vdash \PTRANS[e](z)<\Gamma> \PAR \INPUT{z}{y}.\NEW{x:\TIFC{B}}\PAREN{ \OUTPUT{x}{y} \PAR \PTRANS[S](r)<\Gamma, x:B>} & \text{by \nameref{epi_simple_type_par}} \\
      \PTRANS[\Gamma], r:\TRET                  & \vdash \NEW{z:\TIFC{B}}\PAREN{ \PTRANS[e](z)<\Gamma> \PAR \INPUT{z}{y}.\NEW{x:\TIFC{B}}\PAREN{ \OUTPUT{x}{y} \PAR \PTRANS[S](r)<\Gamma, x:B>} } & \text{by \nameref{epi_simple_type_res}}
    \end{align*}

    \noindent as desired.

  \item Suppose the statement is \code{$x$ := $e$}, so $\Gamma \vdash \code{$x$ := $e$}$.
    This must have been concluded by \nameref{wc_type_stm_assv}, and from its premise we have that 
      $\Gamma \vdash e : B$ and 
      $\Gamma \vdash x : B$ 
    By the encoding of statements, we have that
    \begin{align*}
      \PTRANS[\code{$x$\,:=\,$e$}](r)<\Gamma> & = \NEW{z:\TIFC{B}}\PAREN{ \PTRANS[e](z)<\Gamma> \PAR \INPUT{z}{y} . \LOCWRITE{x}{y} . \OUTPUT{r}{} } \\
                                              & = \NEW{z:\TIFC{B}}\PAREN{ \PTRANS[e](z)<\Gamma> \PAR \INPUT{z}{y} . \INPUT{x}{w} . \PAREN{ \OUTPUT{x}{y} \PAR \OUTPUT{r}{} } }
    \end{align*}

    \noindent for some unused name $w$.
    By Lemma~\ref{wc:lemma:type_correspondence_expr} we have that
    $\PTRANS[\Gamma], z:\TIFC{B} \vdash \PTRANS[e](z)<\Gamma> $, and so 
    \begin{equation*}
      \PTRANS[\Gamma], r:\TRET \vdash \NEW{z:\TIFC{B}}\PAREN{ \PTRANS[e](z)<\Gamma> \PAR \INPUT{z}{y} . \INPUT{x}{w} . \PAREN{ \OUTPUT{x}{y} \PAR \OUTPUT{r}{} } }
    \end{equation*}

    \noindent can be easily shown by rules \nameref{epi_simple_type_in}, \nameref{epi_simple_type_out}, \nameref{epi_simple_type_par} and concluded by \nameref{epi_simple_type_res}, \emph{if} we can show that $\PTRANS[\Gamma]; \PTRANS[\Gamma] \vdash x : \TCHAN{B}$, to be concluded by \nameref{epi_simple_type_vec2} for the premise of \nameref{epi_simple_type_in} and \nameref{epi_simple_type_out}.

    We show this as follows:
    Since we know that $\Gamma \vdash x:B$, then we also know that $\Gamma$ can be expanded as $\Gamma = (x:B), \Gamma'$.
    By the encoding of $\Gamma$, we have that
    \begin{equation*}
      \PTRANS[(x:B), \Gamma'] = x:\TIFC{B}, \TIFC{B}:(\TCHAN{B}, \EMPTYSET), \PTRANS[\Gamma']
    \end{equation*}

    \noindent and hence
    \begin{align*}
      (x:\TIFC{B}, \TIFC{B}:(\TCHAN{B}, \EMPTYSET), \PTRANS[\Gamma'])(x)        & = \TIFC{B}               \\
      (x:\TIFC{B}, \TIFC{B}:(\TCHAN{B}, \EMPTYSET), \PTRANS[\Gamma'])(\TIFC{B}) & = (\TCHAN{B}, \EMPTYSET) \\
      \fst (\TCHAN{B}, \EMPTYSET)                                               & = \TCHAN{B}
    \end{align*}

    \noindent and thus
    \begin{equation*}
      (x:\TIFC{B}, \TIFC{B}:(\TCHAN{B}, \EMPTYSET), \PTRANS[\Gamma']) ; (x:\TIFC{B}, \TIFC{B}:(\TCHAN{B}, \EMPTYSET), \PTRANS[\Gamma']) \vdash x : \TCHAN{B}
    \end{equation*}

    \noindent can indeed be concluded by \nameref{epi_simple_type_vec2}, as required.

  \item Suppose the statement is \code{this.$p$ = $e$}, so $\Gamma \vdash \code{this.$p$ = $e$}$.
    This must have been concluded by \nameref{wc_type_stm_assf}, and from its premise we have that 
    $\Gamma \vdash \code{this.$p$} : B$ and $\Gamma \vdash e : B$.
    By the encoding of statements, we have that
    \begin{align*}
      \PTRANS[\code{this.$p$ := $e$}](r)<\Gamma> & = \NEW{z:\TIFC{B}}\PAREN{ \PTRANS[e](z)<\Gamma> \PAR \INPUT{z}{y} . \LOCREAD{\code{this}}{Y} . \LOCWRITE{Y \cdot p}{y} . \OUTPUT{r}{} } \\
                                                 & = \NEW{z:\TIFC{B}}\PAREN{ \PTRANS[e](z)<\Gamma> \PAR \INPUT{z}{y} . \INPUT{\code{this}}{Y}.\PAREN{ \OUTPUT{\code{this}}{Y} \PAR \LOCWRITE{Y \cdot p}{y} . \OUTPUT{r}{} } } \\
                                                 & = \NEW{z:\TIFC{B}}\PAREN{ \PTRANS[e](z)<\Gamma> \PAR \INPUT{z}{y} . \INPUT{\code{this}}{Y}.\PAREN{ \OUTPUT{\code{this}}{Y} \PAR \INPUT{Y \cdot p}{w}.\PAREN{ \OUTPUT{Y \cdot p}{y} \PAR \OUTPUT{r}{} } } } 
    \end{align*}
    and, by Lemma~\ref{wc:lemma:type_correspondence_expr}, we have that $\PTRANS[\Gamma], z:\TIFC{B} \vdash \PTRANS[e](z)<\Gamma>$.
    As in the case for \code{$x$ := $e$} above, it is then straightforward to show that 
    \begin{equation*}
      \PTRANS[\Gamma], r:\TRET \vdash \NEW{z:\TIFC{B}}\PAREN{ \PTRANS[e](z)<\Gamma> \PAR \INPUT{z}{y} . \INPUT{\code{this}}{Y}.\PAREN{ \OUTPUT{\code{this}}{Y} \PAR \INPUT{Y \cdot p}{w}.\PAREN{ \OUTPUT{Y \cdot p}{y} \PAR \OUTPUT{r}{} } } }
    \end{equation*}

    \noindent by rules \nameref{epi_simple_type_in}, \nameref{epi_simple_type_out}, \nameref{epi_simple_type_par} and concluded by \nameref{epi_simple_type_res}, \emph{if} we can show that
    \begin{align*}
      \PTRANS[\Gamma]; \PTRANS[\Gamma]                   & \vdash \code{this} : \TCHAN{I_A} \\
      (\PTRANS[\Gamma], Y:I_A); (\PTRANS[\Gamma], Y:I_A) & \vdash Y \cdot p : \TCHAN{B}
    \end{align*}

    We show this as follows:
    $\Gamma \vdash \code{this}.p : B$ must have been concluded by \nameref{wc_type_expr_field}, and from its premise we know that
    $\Gamma \vdash \code{this} : I_A$ and
    $\Gamma(I_A)(p) = B$.
    Thus we know that $\Gamma$ can be expanded as 
    \begin{equation*}
      \Gamma = \code{this}:I_A, I_A:(p:B, \Delta), \Gamma'
    \end{equation*}

    \noindent and by the encoding for $\Gamma$, we have that
    \begin{align*}
         & \PTRANS[\code{this}:I_A, I_A:(p:B, \Delta), \Gamma'] \\
      = ~& \code{this}:\TIFC{I_A}, \TIFC{I_A}:(\TCHAN{I_A}, \EMPTYSET), I_A:(\TNIL, \PTRANS[p:B, \Delta]<2>), \PTRANS[p:B, \Delta]<3>, \PTRANS[\Gamma'] \\ 
      = ~& \code{this}:\TIFC{I_A}, \TIFC{I_A}:(\TCHAN{I_A}, \EMPTYSET), I_A:(\TNIL, (p:\TIFC{B}, \PTRANS[\Delta]<2>)), \TIFC{B}:(\TCHAN{B}, \EMPTYSET), \PTRANS[\Delta]<3>, \PTRANS[\Gamma'] 
    \end{align*}

    Hence, we can conclude:
    \begin{align*}
      (\code{this}:\TIFC{I_A}, \TIFC{I_A}:(\TCHAN{I_A}, \EMPTYSET), I_A:(\TNIL, (p:\TIFC{B}, \PTRANS[\Delta]<2>)), \TIFC{B}:(\TCHAN{B}, \EMPTYSET), \PTRANS[\Delta]<3>, \PTRANS[\Gamma'])(\code{this}) & = \TIFC{I_A}               \\
      (\code{this}:\TIFC{I_A}, \TIFC{I_A}:(\TCHAN{I_A}, \EMPTYSET), I_A:(\TNIL, (p:\TIFC{B}, \PTRANS[\Delta]<2>)), \TIFC{B}:(\TCHAN{B}, \EMPTYSET), \PTRANS[\Delta]<3>, \PTRANS[\Gamma'])(\TIFC{I_A})  & = (\TCHAN{I_A}, \EMPTYSET) \\
      \fst (\TCHAN{I_A}, \EMPTYSET)                                                                                                                                                    & = \TCHAN{I_A}
    \end{align*}

    \noindent by \nameref{epi_simple_type_vec2}.
    Thus
    \begin{equation*}
      \INPUT{\code{this}}{Y}.\PAREN{ \OUTPUT{\code{this}}{Y} \PAR \INPUT{Y \cdot p}{w}.\PAREN{ \OUTPUT{Y \cdot p}{y} \PAR \OUTPUT{r}{} } } 
    \end{equation*}

    \noindent can be typed by \nameref{epi_simple_type_in}, and with the continuation typed as
    \begin{equation*}
      \PTRANS[\Gamma], r:\TRET, Y:I_A \vdash \OUTPUT{\code{this}}{Y} \PAR \INPUT{Y \cdot p}{w}.\PAREN{ \OUTPUT{Y \cdot p}{y} \PAR \OUTPUT{r}{} }
    \end{equation*}

    Now, what remains to show is that $(\PTRANS[\Gamma], Y:I_A); \PTRANS[\Gamma], Y:I_A \vdash Y \cdot p : \TCHAN{B}$.
    Like before, we get
    \begin{align*} 
      (\code{this}:\TIFC{I_A}, \TIFC{I_A}:(\TCHAN{I_A}, \EMPTYSET), I_A:(\TNIL, (p:\TIFC{B}, \PTRANS[\Delta]<2>)), \TIFC{B}:(\TCHAN{B}, \EMPTYSET), \PTRANS[\Delta]<3>, \PTRANS[\Gamma'], Y:I_A)(Y)   & = I_A               \\
      (\code{this}:\TIFC{I_A}, \TIFC{I_A}:(\TCHAN{I_A}, \EMPTYSET), I_A:(\TNIL, (p:\TIFC{B}, \PTRANS[\Delta]<2>)), \TIFC{B}:(\TCHAN{B}, \EMPTYSET), \PTRANS[\Delta]<3>, \PTRANS[\Gamma'], Y:I_A)(I_A) & = (\TNIL, (p:\TIFC{B}, \PTRANS[\Delta]<2>)) \\
      \snd (\TNIL, (p:\TIFC{B}, \PTRANS[\Delta]<2>)) & = (p:\TIFC{B}, \PTRANS[\Delta]<2>)
    \end{align*}

    \noindent which satisfies the first premise of \nameref{epi_simple_type_vec1}.
    For the second premise, we must conclude $\PTRANS[\Gamma]; (p:\TIFC{B}, \PTRANS[\Delta]<2>) \vdash p : \TCHAN{B}$.
    We have that 
    \begin{align*}
      (p:\TIFC{B}, \PTRANS[\Delta]<2>)(p) & = \TIFC{B} \\
      (\code{this}:\TIFC{I_A}, \TIFC{I_A}:(\TCHAN{I_A}, \EMPTYSET), I_A:(\TNIL, (p:\TIFC{B}, \PTRANS[\Delta]<2>)), \TIFC{B}:(\TCHAN{B}, \EMPTYSET), \PTRANS[\Delta]<3>, \PTRANS[\Gamma'], Y:I_A)(\TIFC{B}) & = (\TCHAN{B}, \EMPTYSET) \\
      \fst (\TCHAN{B}, \EMPTYSET) & = \TCHAN{B}
    \end{align*}

    \noindent as required for \nameref{epi_simple_type_vec1}.
    Thus 
    \begin{equation*}
      \PTRANS[\Gamma], Y:I_A, y:B \vdash \OUTPUT{Y \cdot p}{y}
    \end{equation*}

    \noindent can be concluded by \nameref{epi_simple_type_out}, as required.

  \item Suppose the statement is \code{$S_1$;$S_2$}, so $\Gamma \vdash \code{$S_1$;$S_2$}$.
    This must have been concluded by \nameref{wc_type_stm_seq}, and from its premise we have that
    $\Gamma \vdash S_1$ and 
    $\Gamma \vdash S_2$.
    By the encoding of statements, we have that 
    \begin{equation*}
      \PTRANS[\code{$S_1$;$S_2$}](r_2)<\Gamma> = \NEW{r_1: \TRET} \PAREN{ \PTRANS[S_1](r_1)<\Gamma> \PAR \INPUT{r_1}{}.\PTRANS[S_2](r_2)<\Gamma> } 
    \end{equation*}

    \noindent and by the induction hypothesis we have that $\PTRANS[\Gamma], r_1:\TRET \vdash \PTRANS[S](r_1)<\Gamma>$ and $\PTRANS[\Gamma], r_2:\TRET \vdash \PTRANS[S](r_2)<\Gamma>$.
    Thus we can conclude the following:
    \begin{align*}
      \PTRANS[\Gamma], r_1 : \TRET, r_2 : \TRET & \vdash \INPUT{r_1}{}.\PTRANS[S_2](r_2)<\Gamma>                                                           & \text{by \nameref{epi_simple_type_in}} \\
      \PTRANS[\Gamma], r_1 : \TRET, r_2 : \TRET & \vdash \PTRANS[S_1](r_1)<\Gamma> \PAR \INPUT{r_1}{}.\PTRANS[S_2](r_2)<\Gamma>                            & \text{by \nameref{epi_simple_type_par}} \\
      \PTRANS[\Gamma], r_2 : \TRET              & \vdash \NEW{r_1 : \TRET}\PAREN{ \PTRANS[S_1](r_1)<\Gamma> \PAR \INPUT{r_1}{}.\PTRANS[S_2](r_2)<\Gamma> } & \text{by \nameref{epi_simple_type_res}}
    \end{align*}

    \noindent as desired.

  \item Suppose the statement is \code{if $e$ then $S_\TRUE$ else $S_\FALSE$}, so $\Gamma \vdash \code{if $e$ then $S_\TRUE$ else $S_\FALSE$}$.
    This must have been concluded by \nameref{wc_type_stm_if}, and from its premise we have that $\Gamma \vdash e : \TBOOL$, $\Gamma \vdash S_{\TRUE}$, and $\Gamma \vdash S_{\FALSE}$.
    By the encoding of statements, we have that 
    \begin{align*}
         & \PTRANS[\code{if $e$ then $S_\TRUE$ else $S_\FALSE$}](r)<\Gamma> \\ 
      = ~& \NEW{z:\TIFC{\TBOOL}}\PAREN{ \PTRANS[e](z)<\Gamma> \PAR \INPUT{z}{y}. \IFTHENELSE{y}{ \PTRANS[S_\TRUE](r)<\Gamma> }{ \PTRANS[S_\FALSE](r)<\Gamma> } } \\
      = ~& \NEW{z:\TIFC{\TBOOL}}\PAREN{ \PTRANS[e](z)<\Gamma> \PAR \INPUT{z}{y}. ([y = \TRUE]\PTRANS[S_\TRUE](r)<\Gamma> + [y = \FALSE]\PTRANS[S_\FALSE](r)<\Gamma>) }
    \end{align*}

    \noindent where we know that $\FRESH{z,y}\PTRANS[S_\TRUE](r)<\Gamma>, \PTRANS[S_\FALSE](r)<\Gamma>$.
    Furthermore, by the induction hypothesis, we have that
    $\PTRANS[\Gamma], r:\TRET \vdash \PTRANS[S_\TRUE](r)<\Gamma>$ and $\PTRANS[\Gamma], r:\TRET \vdash \PTRANS[S_\FALSE](r)<\Gamma>$;
    furthermore, by Lemma~\ref{wc:lemma:type_correspondence_expr}, we have that $\PTRANS[\Gamma], z:\TIFC{\TBOOL} \vdash \PTRANS[e](z)<\Gamma>$.
    After an application of Lemma~\ref{wc:lemma:epi_simple_type_weakening} to add the type entries $z:\TIFC{\TBOOL}, y:\TBOOL$, we can then conclude the following:
    \begin{align*}
      \PTRANS[\Gamma], r:\TRET, z:\TIFC{\TBOOL}, y:\TBOOL & \vdash [y = \TRUE]\PTRANS[S_\TRUE](r)<\Gamma> + [y = \FALSE]\PTRANS[S_\FALSE](r)<\Gamma>                                                                    & \text{by \nameref{epi_simple_type_sum}} \\
      \PTRANS[\Gamma], r:\TRET, z:\TIFC{\TBOOL}           & \vdash \INPUT{z}{y}.([y = \TRUE]\PTRANS[S_\TRUE](r)<\Gamma> + [y = \FALSE]\PTRANS[S_\FALSE](r)<\Gamma>)                                                           & \text{by \nameref{epi_simple_type_in}} \\
      \PTRANS[\Gamma], r:\TRET, z:\TIFC{\TBOOL}           & \vdash \PTRANS[e](z)<\Gamma> \PAR \INPUT{z}{y}.([y = \TRUE]\PTRANS[S_\TRUE](r)<\Gamma> + [y = \FALSE]\PTRANS[S_\FALSE](r)<\Gamma>)                                & \text{by \nameref{epi_simple_type_par}} \\
      \PTRANS[\Gamma], r:\TRET                            & \vdash \NEW{z:\TIFC{\TBOOL}}\PAREN{ \PTRANS[e](z)<\Gamma> \PAR \INPUT{z}{y}.( [y = \TRUE]\PTRANS[S_\TRUE](r)<\Gamma> + [y = \FALSE]\PTRANS[S_\FALSE](r)<\Gamma>) } & \text{by \nameref{epi_simple_type_res}}
    \end{align*}

    \noindent as desired.

  \item Suppose the statement is \code{while $e$ do $S$}, so $\Gamma \vdash \code{while $e$ do $S$}$.
    This must have been concluded by \nameref{wc_type_stm_while}, and from its premise we have that $\Gamma \vdash e : \TBOOL$ and $
      \Gamma \vdash S$.
    By the encoding of statements, we have that
    \begin{align*}
      \PTRANS[\code{while $e$ do $S$}](r)<\Gamma> & = \NEW{r':\TRET}\PAREN[Big]{ \OUTPUT{r'}{} \PAR \REPL{\INPUT{r'}{}}                                                                       %
                                                    . \NEW{z:\TIFC{\TBOOL}}\PAREN{ \PTRANS[e](z)<\Gamma> \PAR \INPUT{z}{y}. \IFTHENELSE{y}{ \PTRANS[S](r')<\Gamma> }{ \OUTPUT{r}{} } } }      \\ 
                                                  & = \NEW{r':\TRET}\PAREN[Big]{ \OUTPUT{r'}{} \PAR \REPL{\INPUT{r'}{}}                                                                       %
                                                    . \NEW{z:\TIFC{\TBOOL}}\PAREN{ \PTRANS[e](z)<\Gamma> \PAR \INPUT{z}{y}.( [y=\TRUE]\PTRANS[S](r')<\Gamma> + [y=\FALSE]\OUTPUT{r}{}) } }  
    \end{align*}

    \noindent where we know that $\FRESH{z,y}\PTRANS[S](r')<\Gamma>$, and $\FRESH{r'}\PTRANS[e](z)<\Gamma>$.
    Furthermore, by the induction hypothesis, we have that
    $\PTRANS[\Gamma], r':\TRET \vdash \PTRANS[S](r')<\Gamma>$
    and, by Lemma~\ref{wc:lemma:type_correspondence_expr}, that $\PTRANS[\Gamma], z:\TIFC{\TBOOL} \vdash \PTRANS[e](z)<\Gamma>$.
    After an application of Lemma~\ref{wc:lemma:epi_simple_type_weakening} to add the type entries $r':\TRET, z:\TIFC{\TBOOL}, y:\TBOOL$, we can then conclude the following:
    \begin{align*}
      \PTRANS[\Gamma], r:\TRET, r':\TRET, z:\TIFC{\TBOOL}, y:\TBOOL & \vdash \OUTPUT{r}{}                                                                                                                        & \text{by \nameref{epi_simple_type_out}} \\
      \PTRANS[\Gamma], r:\TRET, r':\TRET, z:\TIFC{\TBOOL}, y:\TBOOL & \vdash  [y=\TRUE]\PTRANS[S](r')<\Gamma> +[y=\FALSE]\OUTPUT{r}{}                                                                        & \text{by \nameref{epi_simple_type_sum}} \\
      \PTRANS[\Gamma], r:\TRET, r':\TRET, z:\TIFC{\TBOOL}           & \vdash \INPUT{z}{y}.( [y=\TRUE]\PTRANS[S](r')<\Gamma> + [y=\FALSE]\OUTPUT{r}{})                                                           & \text{by \nameref{epi_simple_type_in}} \\
      \PTRANS[\Gamma], r:\TRET, r':\TRET, z:\TIFC{\TBOOL}           & \vdash \PTRANS[e](z)<\Gamma> \PAR \INPUT{z}{y}.([y=\TRUE]\PTRANS[S](r')<\Gamma> + [y=\FALSE]\OUTPUT{r}{})                                & \text{by \nameref{epi_simple_type_par}} \\
      \PTRANS[\Gamma], r:\TRET, r':\TRET                            & \vdash \NEW{z:\TIFC{\TBOOL}}\PAREN{ \PTRANS[e](z)<\Gamma> \PAR \INPUT{z}{y}.( [y=\TRUE]\PTRANS[S](r')<\Gamma> + [y=\FALSE]\OUTPUT{r}{}) } & \text{by \nameref{epi_simple_type_res}}
    \end{align*}
    \begin{align*}
      \PTRANS[\Gamma], r:\TRET, r':\TRET & \vdash \INPUT{r'}{}.\NEW{z:\TIFC{\TBOOL}}\PAREN{ \PTRANS[e](z)<\Gamma> \PAR \INPUT{z}{y}.( [y=\TRUE]\PTRANS[S](r')<\Gamma> + [y=\FALSE]\OUTPUT{r}{}) }                           & \text{by \nameref{epi_simple_type_in}} \\
      \PTRANS[\Gamma], r:\TRET, r':\TRET & \vdash \REPL{\INPUT{r'}{}}.\NEW{z:\TIFC{\TBOOL}}\PAREN{ \PTRANS[e](z)<\Gamma> \PAR \INPUT{z}{y}.( [y=\TRUE]\PTRANS[S](r')<\Gamma> + [y=\FALSE]\OUTPUT{r}{}) }                    & \text{by \nameref{epi_simple_type_rep}} \\
      \PTRANS[\Gamma], r:\TRET, r':\TRET & \vdash \OUTPUT{r'}{}                                                                                                                                                              & \text{by \nameref{epi_simple_type_out}} \\
      \PTRANS[\Gamma], r:\TRET, r':\TRET & \vdash \OUTPUT{r'}{} \PAR \REPL{\INPUT{r'}{}}.\NEW{z:\TIFC{\TBOOL}}\PAREN{ \PTRANS[e](z)<\Gamma> \PAR \INPUT{z}{y}.( [y=\TRUE]\PTRANS[S](r')<\Gamma> + [y=\FALSE]\OUTPUT{r}{}) } & \text{by \nameref{epi_simple_type_par}}
    \end{align*}
    \begin{equation*}
      \PTRANS[\Gamma], r:\TRET \vdash \NEW{r':\TRET}\PAREN{ \OUTPUT{r'}{} \PAR \REPL{\INPUT{r'}{}}.\NEW{z:\TIFC{\TBOOL}}\PAREN{ \PTRANS[e](z)<\Gamma> \PAR \INPUT{z}{y}.( [y=\TRUE]\PTRANS[S](r')<\Gamma> + [y=\FALSE]\OUTPUT{r}{}) } } \ \ \text{by \nameref{epi_simple_type_res}}
    \end{equation*}

    \noindent as desired.

  \item Suppose the statement is \WCCALL{e}{f}{\VEC{e}}, so $\Gamma \vdash \WCCALL{e}{f}{\VEC{e}}$.
    This must have been concluded by \nameref{wc_type_stm_call}, and from its premise we have that 
    \begin{align*}
      \Gamma \vdash e : I_A   
      \qquad
      \Gamma \vdash e_1 : B_1
        \quad\ldots\quad
      \Gamma \vdash e_n : B_n 
      \qquad 
      \Gamma(I_A)(f) = \TPROC{B_1, \ldots, B_n} 
    \end{align*}

    By the encoding of statements, we have that 
    \begin{align*}
      \PTRANS[\WCCALL{e}{f}{\VEC{e}}](r)<\Gamma> & = \NEW{a:\TIFC{{I_A}}, \VEC{z}:\TIFC{{\VEC{B}}}}\PAREN[Big]{ \PTRANS[e](a)<\Gamma> \PAR \PTRANS[\VEC{e}](\VEC{z})<\Gamma>  %
                                                     \PAR \INPUT{a}{Y} . \INPUT{z_1}{y_1} \ldots \INPUT{z_n}{y_n} . \OUTPUT{Y \cdot f}{ r, y_1, \ldots, y_n } } 
    \end{align*}

    \noindent where $\PTRANS[\VEC{e}](\VEC{z})<\Gamma> = \PTRANS[e_1](z_1)<\Gamma> \PAR \ldots \PAR \PTRANS[e_n](z_n)<\Gamma>$ and $\VEC{z}:\TIFC{{\VEC{B}}} = z_1:\TIFC{B_1}, \ldots, z_n:\TIFC{B_n}$ for brevity.
    By Lemma~\ref{wc:lemma:type_correspondence_expr}, we have that 
    $\PTRANS[\Gamma], a : \TIFC{I_A} \vdash \PTRANS[e](a)<\Gamma>$ and that $\PTRANS[\Gamma], z_i : \TIFC{B_i} \vdash \PTRANS[e_i](z_i)<\Gamma>$, for all $i = 1,\ldots,n$.
    Hence, we can immediately conclude
    \begin{equation*}
      \PTRANS[\Gamma], r:\TRET, a:\TIFC{I_A}, z_1:\TIFC{B_1}, \ldots, z_n:\TIFC{B_n} \vdash \PTRANS[e](a)<\Gamma> \PAR \PTRANS[e_1](z_1)<\Gamma> \PAR \ldots \PAR \PTRANS[e_n](z_n)<\Gamma>
    \end{equation*}

    \noindent by using rule \nameref{epi_simple_type_par} and with Lemma~\ref{wc:lemma:epi_simple_type_weakening} used to add the assumption $r:\TRET$ (since we know that $\FRESH{r}\PTRANS[e](a)<\Gamma>, \PTRANS[\VEC{e}](\VEC{z})<\Gamma>$).
    
    What remains to be shown is that also 
    \begin{equation*}
      \PTRANS[\Gamma], r:\TRET, a:\TIFC{I_A}, z_1:\TIFC{B_1}, \ldots, z_n:\TIFC{B_n} \vdash \INPUT{a}{Y} . \INPUT{z_1}{y_1} \ldots \INPUT{z_n}{y_n} . \OUTPUT{Y \cdot f}{ r, y_1, \ldots, y_n }
    \end{equation*}

    \noindent which is shown by using rule \nameref{epi_simple_type_in}, until we are left with the final premise to show, which is:
    \begin{equation*} 
      \PTRANS[\Gamma], r:\TRET, a:\TIFC{I_A}, z_1:\TIFC{B_1}, \ldots, z_n:\TIFC{B_n}, Y:I_A, y_1:B_1, \ldots, y_n:B_n \vdash \OUTPUT{Y \cdot f}{ r, y_1, \ldots, y_n }
    \end{equation*}

    From $\Gamma(I_A)(f) = \TPROC{B_1, \ldots, B_n}$ we know that $\Gamma$ can be expanded as
    \begin{equation*}
      \Gamma = I_A : (f : \TPROC{B_1, \ldots, B_n}, \Delta), \Gamma'
    \end{equation*}

    \noindent and by the encoding of $\Gamma$ we have that 
    \begin{align*}
         & \PTRANS[I_A:(f:\TPROC{B_1, \ldots, B_n}, \Delta), \Gamma']                                                                                                              \\
      = ~& I_A:(\TNIL, (\PTRANS[f:\TPROC{B_1, \ldots, B_n}, \Delta]<2>)), \PTRANS[f:\TPROC{B_1, \ldots, B_n}, \Delta]<3>, \PTRANS[\Gamma']                                         \\
      = ~& I_A:(\TNIL, (f:\TIFC{B_1, \ldots, B_n}, \PTRANS[\Delta]<2>)), \TIFC{B_1, \ldots, B_n}:(\TCHAN{\TRET, B_1, \ldots, B_n}, \EMPTYSET), \PTRANS[\Delta]<3>, \PTRANS[\Gamma'] 
    \end{align*}

    \noindent Hence, we can conclude 
    \begin{equation*}
      \PTRANS[\Gamma], \ldots, Y:I_A; \PTRANS[\Gamma], \ldots, Y:I_A \vdash Y \cdot f : \TCHAN{\TRET, B_1, \ldots, B_n} 
    \end{equation*}

    \noindent by \nameref{epi_simple_type_vec1} (and \nameref{epi_simple_type_vec2} for its premise), and then 
    \begin{equation*} 
      \PTRANS[\Gamma], r:\TRET, a:\TIFC{I_A}, z_1:\TIFC{B_1}, \ldots, z_n:\TIFC{B_n}, Y:I_A, y_1:B_1, \ldots, y_n:B_n \vdash \OUTPUT{Y \cdot f}{ r, y_1, \ldots, y_n }
    \end{equation*}

    \noindent by \nameref{epi_simple_type_out}.
    At last, we can therefore conclude
    \begin{equation*}
      \PTRANS[\Gamma], r:\TRET \vdash \NEW{a:\TIFC{{I_A}}, \VEC{z}:\TIFC{{\VEC{B}}}}\PAREN[Big]{ \PTRANS[e](a)<\Gamma> \PAR \PTRANS[\VEC{e}](\VEC{z})<\Gamma> \PAR \INPUT{a}{Y} . \INPUT{z_1}{y_1} \ldots \INPUT{z_n}{y_n} . \OUTPUT{Y \cdot f}{ r, y_1, \ldots, y_n } } 
    \end{equation*}

    \noindent by \nameref{epi_simple_type_par} and then \nameref{epi_simple_type_res}, as desired.
\end{itemize}

This concludes the proof for the forward direction.

For the other direction, we must show that $\PTRANS[\Gamma], r : \TRET \vdash \PTRANS[S](r)<\Gamma> \implies \Gamma \vdash S$ and again the proof is by induction on $S$.
\begin{itemize}
  \item Suppose the statement is \code{skip} and $\PTRANS[\Gamma], r:\TRET \vdash \PTRANS[\code{skip}](r)<\Gamma>$.
    Then $\Gamma \vdash \code{skip}$ holds for any $\Gamma$ by \nameref{wc_type_stm_skip}, as desired.

  \item Suppose the statement is \code{var $B$ $x$ := $e$ in $S$}, and $\PTRANS[\Gamma], r:\TRET \vdash \PTRANS[\code{var $B$ $x$ := $e$ in $S$}]$.
    Then 
    \begin{align*}
        ~& \PTRANS[\code{var $B$ $x$\,:=\,$e$ in $S$}](r)<\Gamma> \\
      = ~& \NEW{z:\TIFC{B}}\PAREN{ \PTRANS[e](z)<\Gamma> \PAR \INPUT{z}{y} . \NEW{x:\TIFC{B}}\PAREN{ \OUTPUT{x}{y} \PAR \PTRANS[S](r)<\Gamma, x:B> } }
    \end{align*}

    Our goal is to infer $\Gamma \vdash \code{var $B$ $x$ := $e$ in $S$}$ using \nameref{wc_type_stm_decv}. Thus, we must show that the premises of that rule, namely $\Gamma \vdash e:B$ and $\Gamma, x:B \vdash S$, can be satisfied.
    
    We know that
    \begin{equation*}
      \PTRANS[\Gamma], r:\TRET \vdash \NEW{z:\TIFC{B}}\PAREN{ \PTRANS[e](z)<\Gamma> \PAR \INPUT{z}{y} . \NEW{x:\TIFC{B}}\PAREN{ \OUTPUT{x}{y} \PAR \PTRANS[S](r)<\Gamma, x:B> } }
    \end{equation*}

    \noindent which must have been concluded by \nameref{epi_simple_type_res}, \nameref{epi_simple_type_par}, \nameref{epi_simple_type_in}, \nameref{epi_simple_type_res} again, \nameref{epi_simple_type_par} again, and finally \nameref{epi_simple_type_out}.
    We omit the full derivation here, but see the corresponding case in the proof for the forward direction.

    From their respective premises, we get that $\PTRANS[\Gamma], z:\TIFC{B} \vdash \PTRANS[e](z)<\Gamma>$ and $\PTRANS[\Gamma], x:\TIFC{B}, r:\TRET \vdash \PTRANS[S](r)<\Gamma>$.
    By Lemma~\ref{wc:lemma:type_correspondence_expr}, the former implies that 
    $\Gamma \vdash e : B$.
    The latter implies that $\PTRANS[\Gamma]$ contains an interface definition for $\TIFC{B}$ of the form $\TIFC{B} \mapsto (\TCHAN{B}, \EMPTYSET)$; by the encoding for $\Gamma$, we have that 
    \begin{equation*}
      \PTRANS[\Gamma], x:\TIFC{B}, \TIFC{B}:(\TCHAN{B}, \EMPTYSET), r:\TRET = \PTRANS[\Gamma, x:B], r:\TRET
    \end{equation*}

    \noindent and thus $\PTRANS[\Gamma, x:B], r:\TRET \vdash \PTRANS[S](r)<\Gamma>$. By the induction hypothesis, $\Gamma, x:B \vdash S$;
    
    We can then conclude $\Gamma \vdash \code{var $B$ $x$ := $e$ in $S$}$ by \nameref{wc_type_stm_decv}, as desired.

  \item Suppose the statement is \code{$x$ := $e$}, and $\PTRANS[\Gamma], r:\TRET \vdash \PTRANS[\code{$x$ := $e$}](r)<\Gamma>$.
    Then
    \begin{equation*}
      \PTRANS[\code{$x$ := $e$}](r)<\Gamma> = \NEW{z:\TIFC{B}}\PAREN{ \PTRANS[e](z)<\Gamma> \PAR \INPUT{z}{y} . \INPUT{x}{w} . \PAREN{ \OUTPUT{x}{y} \PAR \OUTPUT{r}{} } } 
    \end{equation*}

    Our goal is to conclude by \nameref{wc_type_stm_assv}, thus we must show that its premises can be satisfied, which are $\Gamma \vdash x:B$ and $\Gamma \vdash e:B$.

    We know that 
    \begin{equation*}
      \PTRANS[\Gamma], r:\TRET \vdash \NEW{z:\TIFC{B}}\PAREN{ \PTRANS[e](z)<\Gamma> \PAR \INPUT{z}{y} . \INPUT{x}{w} . \PAREN{ \OUTPUT{x}{y} \PAR \OUTPUT{r}{} } } 
    \end{equation*}

    \noindent which must have been concluded by \nameref{epi_simple_type_res}, \nameref{epi_simple_type_par}, \nameref{epi_simple_type_in} (twice) and then \nameref{epi_simple_type_par} and \nameref{epi_simple_type_out}.
    We omit the full derivation here, but see the corresponding case in the proof for the forward direction.
    However, in brief, we infer from the entry $z:\TIFC{B}$ that $x$ must be typable as $\TCHAN{B}$.
    In particular, we get that $\PTRANS[\Gamma]$ can be expanded as $\PTRANS[\Gamma'], x:\TIFC{B}$ (since $x$ is free in the above process).
    Then $\PTRANS[\Gamma'], x:\TIFC{B}, y:B \vdash \OUTPUT{x}{y}$, which was concluded by \nameref{epi_simple_type_out}, and for its premise that
    \begin{equation*}
      (\PTRANS[\Gamma'], x:\TIFC{B}); (\PTRANS[\Gamma'], x:\TIFC{B}) \vdash x : \TCHAN{B}
    \end{equation*}

    \noindent which was concluded by \nameref{epi_simple_type_vec2}.
    This, in turn, requires that $\PTRANS[\Gamma']$ must contain $\TIFC{B} \mapsto (\TCHAN{B}, \EMPTYSET)$, and, by the encoding for $\Gamma$, we have that 
    \begin{equation*}
      \PTRANS[\Gamma'], x:\TIFC{B}, \TIFC{B}:(\TCHAN{B}, \EMPTYSET), r:\TRET = \PTRANS[\Gamma', x:B], r:\TRET
    \end{equation*}

    \noindent hence, $\Gamma = \Gamma', x:B$, and clearly $\Gamma', x:B \vdash x:B$ by \nameref{wc_type_expr_var}.
    
    Furthermore, from the premise of \nameref{epi_simple_type_res} and \nameref{epi_simple_type_par}, we have  $\PTRANS[\Gamma], z:\TIFC{B} \vdash \PTRANS[e](z)<\Gamma>$ that, by Lemma~\ref{wc:lemma:type_correspondence_expr}, implies $\Gamma \vdash e : B$.
    Thus we can conclude $\Gamma \vdash \code{$x$ := $e$}$ by \nameref{wc_type_stm_assv}, as desired.

  \item Suppose the statement is \code{this.$p$ := $e$}, and $\PTRANS[\Gamma], r:\TRET \vdash \PTRANS[\code{this.$p$ := $e$}](r)<\Gamma>$.
    Then
    \begin{equation*}
      \PTRANS[\code{this.$p$ := $e$}](r)<\Gamma>  = \NEW{z:\TIFC{B}}\PAREN{ \PTRANS[e](z)<\Gamma> \PAR \INPUT{z}{y} . \INPUT{\code{this}}{Y}.\PAREN{ \OUTPUT{\code{this}}{Y} \PAR \INPUT{Y \cdot p}{w}.\PAREN{ \OUTPUT{Y \cdot p}{y} \PAR \OUTPUT{r}{} } } } 
    \end{equation*}

    Our goal is to conclude by \nameref{wc_type_stm_assf}, thus we must show that its premises can be satisfied, which are $\Gamma \vdash \code{this}.p:B$ and $\Gamma \vdash e:B$.

    We know that
    \begin{equation*}
      \PTRANS[\Gamma], r:\TRET \vdash \NEW{z:\TIFC{B}}\PAREN{ \PTRANS[e](z)<\Gamma> \PAR \INPUT{z}{y} . \INPUT{\code{this}}{Y}.\PAREN{ \OUTPUT{\code{this}}{Y} \PAR \INPUT{Y \cdot p}{w}.\PAREN{ \OUTPUT{Y \cdot p}{y} \PAR \OUTPUT{r}{} } } } 
    \end{equation*}

    \noindent which must have been concluded by \nameref{epi_simple_type_res}, \nameref{epi_simple_type_par}, \nameref{epi_simple_type_in} and \nameref{epi_simple_type_out}.
    We omit the full derivation here, but see the corresponding case in the proof for the forward direction.
    However, in brief, we infer from the entry $z:\TIFC{B}$ that $y$ must have type $B$, and therefore $Y \cdot p$ must be typable as $\TCHAN{B}$, which is required to type the output $\OUTPUT{Y \cdot p}{y}$.
    This, in turn, requires that \code{this} has a type $\TIFC{A}$, which must have an interface definition of the form $I_A \mapsto (\TNIL, (p:\TIFC{B}, \Delta))$ for some $\Delta$.
    Thus we get that $\PTRANS[\Gamma] = 
      \code{this}:\TIFC{I_A}, \TIFC{I_A}:(\TCHAN{A}, \EMPTYSET), I_A:(\TNIL, (p:\TIFC{B}, \PTRANS[\Delta]<2>)), \TIFC{B}:(\TCHAN{B}, \EMPTYSET), \PTRANS[\Delta]<3>, \PTRANS[\Gamma']$ 
    and, by the translation for $\Gamma$, we have that
    \begin{align*}
         & \code{this}:\TIFC{I_A}, \TIFC{I_A}:(\TCHAN{I_A}, \EMPTYSET), I_A:(\TNIL, (p:\TIFC{B}, \PTRANS[\Delta]<2>)), \TIFC{B}:(\TCHAN{B}, \EMPTYSET), \PTRANS[\Delta]<3>, \PTRANS[\Gamma'] \\
      = ~& \PTRANS[\code{this}:I_A, I_A:(p:B, \Delta), \Gamma'] 
    \end{align*}

    \noindent so $\Gamma = \code{this}:I_A, I_A:(p:B, \Delta), \Gamma'$.
    Then $\code{this}:I_A, I_A:(p:B, \Delta), \Gamma' \vdash \code{this}.p : B$ can be concluded by \nameref{wc_type_expr_field}.

   Like in the previous case, from the premise of \nameref{epi_simple_type_res} and \nameref{epi_simple_type_par}, we have that $\PTRANS[\Gamma], z:\TIFC{B} \vdash \PTRANS[e](z)<\Gamma>$ and, by Lemma~\ref{wc:lemma:type_correspondence_expr}, $\Gamma \vdash e : B$.
    Thus we can conclude $\Gamma \vdash \code{this.$p$ := $e$}$ by \nameref{wc_type_stm_assf}, as desired.

  \item Suppose the statement is \code{$S_1$;$S_2$}, and $\PTRANS[\Gamma], r_2:\TRET \vdash \PTRANS[\code{$S_1$;$S_2$}](r_2)<\Gamma>$.
    Then 
    \begin{equation*}
      \PTRANS[\code{$S_1$;$S_2$}](r_2)<\Gamma> = \NEW{r_1: \TRET} \PAREN{ \PTRANS[S_1](r_1)<\Gamma> \PAR \INPUT{r_1}{}.\PTRANS[S_2](r_2)<\Gamma> } 
    \end{equation*}

    Our goal is to conclude by \nameref{wc_type_stm_seq}, thus we must show that its premises can be satisfied, which are $\Gamma \vdash S_1$ and $\Gamma \vdash S_S$.
    
    We know that
    \begin{equation*} 
      \PTRANS[\Gamma], r_2 : \TRET \vdash \NEW{r_1 : \TRET}\PAREN{ \PTRANS[S_1](r_1)<\Gamma> \PAR \INPUT{r_1}{}.\PTRANS[S_2](r_2)<\Gamma> }     
    \end{equation*}

    \noindent which must have been concluded by \nameref{epi_simple_type_res} and \nameref{epi_simple_type_par}, and with the premises:
    \begin{align*}
      \PTRANS[\Gamma], r_1 : \TRET, r_2 : \TRET \vdash \PTRANS[S_1](r_1)<\Gamma>   
      \qquad
      \PTRANS[\Gamma], r_1 : \TRET, r_2 : \TRET \vdash \INPUT{r_1}{}.\PTRANS[S_2](r_2)<\Gamma>
    \end{align*}
    By applying \nameref{epi_simple_type_in} to the latter, we also know that $\PTRANS[\Gamma], r_1 : \TRET, r_2 : \TRET \vdash \PTRANS[S_2](r_2)<\Gamma>$.
    Then, after removing the unused type assumptions by Lemma~\ref{wc:lemma:epi_simple_type_strengthening} and by two applications of the induction hypothesis, we have the desired $\Gamma \vdash S_1$ and $\Gamma \vdash S_2$.

  \item Suppose the statement is \code{if $e$ then $S_\TRUE$ else $S_\FALSE$}, and \mbox{$\PTRANS[\Gamma], r:\TRET \vdash \PTRANS[\code{if $e$ then $S_\TRUE$ else $S_\FALSE$}](r)<\Gamma>$}.
    Then
    \begin{align*}
         & \PTRANS[\code{if $e$ then $S_\TRUE$ else $S_\FALSE$}](r)<\Gamma> \\ 
      = ~& \NEW{z:\TIFC{\TBOOL}}\PAREN{ \PTRANS[e](z)<\Gamma> \PAR \INPUT{z}{y}. ([y = \TRUE]\PTRANS[S_\TRUE](r)<\Gamma> + [y = \FALSE]\PTRANS[S_\FALSE](r)<\Gamma>) }
    \end{align*}

    Our goal is to conclude by \nameref{wc_type_stm_if}, thus we must show that its premises can be satisfied, which are $\Gamma \vdash e:\TBOOL$ and $\Gamma \vdash S_\TRUE$ and $\Gamma \vdash S_\FALSE$.

    We know that 
    \begin{equation*}
      \PTRANS[\Gamma], r:\TRET \vdash \NEW{z:\TIFC{\TBOOL}}\PAREN{ \PTRANS[e](z)<\Gamma> \PAR \INPUT{z}{y}. ([y = \TRUE]\PTRANS[S_\TRUE](r)<\Gamma> + [y = \FALSE]\PTRANS[S_\FALSE](r)<\Gamma>) } 
    \end{equation*}

    \noindent must then have been concluded by \nameref{epi_simple_type_res}, \nameref{epi_simple_type_par}, \nameref{epi_simple_type_in} and \nameref{epi_simple_type_sum} (see the derivation in the corresponding forward case).
    From their respective premises, we get that $\PTRANS[\Gamma], z:\TIFC{\TBOOL} \vdash \PTRANS[e](z)<\Gamma>$ and $\PTRANS[\Gamma], r:\TRET \vdash \PTRANS[S_\TRUE](r)<\Gamma>$ and $\PTRANS[\Gamma], r:\TRET \vdash \PTRANS[S_\FALSE](r)<\Gamma>$.
    Respectively by Lemma~\ref{wc:lemma:type_correspondence_expr} and by two applications of the induction hypothesis, these imply the desired $\Gamma \vdash e : \TBOOL$, $\Gamma \vdash S_\TRUE$ and $\Gamma \vdash S_\FALSE$.
    
  \item Suppose the statement is \code{while $e$ do $S$}, and $\PTRANS[\Gamma], r:\TRET \vdash \PTRANS[\code{while $e$ do $S$}](r)<\Gamma>$.
    Then
    \begin{align*}
         & \PTRANS[\code{while $e$ do $S$}](r)<\Gamma>                       \\
      = ~& \NEW{r':\TRET}\PAREN[Big]{ \OUTPUT{r'}{} \PAR \REPL{\INPUT{r'}{}} %
         . \NEW{z:\TIFC{\TBOOL}}\PAREN{ \PTRANS[e](z)<\Gamma> \PAR \INPUT{z}{y}.( [y=\TRUE]\PTRANS[S](r')<\Gamma> + [y=\FALSE]\OUTPUT{r}{}) } }  
    \end{align*}

    Our goal is to conclude by \nameref{wc_type_stm_while}, thus we must show that its premises can be satisfied, which are $\Gamma \vdash e:\TBOOL$ and $\Gamma \vdash S$.

    We know that
    \begin{equation*}
      \PTRANS[\Gamma], r:\TRET \vdash \NEW{r':\TRET}\PAREN[Big]{ \OUTPUT{r'}{} \PAR \REPL{\INPUT{r'}{}} %
                                    . \NEW{z:\TIFC{\TBOOL}}\PAREN{ \PTRANS[e](z)<\Gamma> \PAR \INPUT{z}{y}.( [y=\TRUE]\PTRANS[S](r')<\Gamma> + [y=\FALSE]\OUTPUT{r}{}) } }
    \end{equation*}

    \noindent must have been concluded by \nameref{epi_simple_type_res}, \nameref{epi_simple_type_par}, \nameref{epi_simple_type_out}, \nameref{epi_simple_type_rep}, \nameref{epi_simple_type_in} and \nameref{epi_simple_type_sum}.
    We omit the full derivation here, but see the corresponding case in the proof for the forward direction.

    From the respective premises of these rules, we get that $\PTRANS[\Gamma], z:\TIFC{\TBOOL} \vdash \PTRANS[e](z)<\Gamma>$ and $\PTRANS[\Gamma], r':\TRET \vdash \PTRANS[S](r')<\Gamma>$, and with unused type assumptions removed by Lemma~\ref{wc:lemma:epi_simple_type_strengthening}.
    Again by Lemma~\ref{wc:lemma:type_correspondence_expr} and by the induction hypothesis, these allow us to conclude.  

  \item Suppose the statement is \WCCALL{e}{f}{\VEC{e}}, and $\PTRANS[\Gamma], r:\TRET \vdash \PTRANS[\WCCALL{e}{f}{\VEC{e}}](r)<\Gamma>$.
    Then 
    \begin{align*}
         & \PTRANS[\WCCALL{e}{f}{\VEC{e}}](r)<\Gamma> \\
      = ~& \NEW{a:\TIFC{{I_A}}, \VEC{z}:\TIFC{{\VEC{B}}}}\PAREN[Big]{ \PTRANS[e](a)<\Gamma> \PAR \PTRANS[\VEC{e}](\VEC{z})<\Gamma>  %
           \PAR \INPUT{a}{Y} . \INPUT{z_1}{y_1} \ldots \INPUT{z_n}{y_n} . \OUTPUT{Y \cdot f}{ r, y_1, \ldots, y_n } } 
    \end{align*}
    
    \noindent where $\PTRANS[\VEC{e}](\VEC{z})<\Gamma> = \PTRANS[e_1](z_1)<\Gamma> \PAR \ldots \PAR \PTRANS[e_n](z_n)<\Gamma>$ and $\VEC{z}:\TIFC{{\VEC{B}}} = z_1:\TIFC{B_1}, \ldots, z_n:\TIFC{B_n}$ for brevity.

    Our goal is to conclude by \nameref{wc_type_stm_call}, and to do that, we must satisfy its premises, which are
    \begin{equation*}
      \Gamma \vdash e:B \qquad \Gamma \vdash \VEC{e} : \VEC{B} \qquad \Gamma(I_A)(f) = \TPROC{\VEC{B}}
    \end{equation*}

    We have that 
    \begin{equation*}
      \PTRANS[\Gamma], r:\TRET \vdash \NEW{a:\TIFC{{I_A}}, \VEC{z}:\TIFC{{\VEC{B}}}}\PAREN[Big]{ \PTRANS[e](a)<\Gamma> \PAR \PTRANS[\VEC{e}](\VEC{z})<\Gamma>  %
           \PAR \INPUT{a}{Y} . \INPUT{z_1}{y_1} \ldots \INPUT{z_n}{y_n} . \OUTPUT{Y \cdot f}{ r, y_1, \ldots, y_n } }
    \end{equation*}

    \noindent must have been concluded by \nameref{epi_simple_type_res}, \nameref{epi_simple_type_par}, \nameref{epi_simple_type_in} and \nameref{epi_simple_type_out}.
    We omit the full derivation here, but see the corresponding case in the proof for the forward direction.
    From their respective premises, we get that 
    \begin{align*}
      \PTRANS[\Gamma], a:\TIFC{I_A} \vdash \PTRANS[e](a)<\Gamma>  
      \qquad
      \PTRANS[\Gamma], z_1:\TIFC{B_1}  \vdash \PTRANS[e_1](z_1)<\Gamma>
      \quad
      \ldots
      \quad
      \PTRANS[\Gamma], z_n:\TIFC{B_n} \vdash \PTRANS[e_n](z_n)<\Gamma>
    \end{align*}

    \noindent and that $\PTRANS[\Gamma], r:\TRET, Y:I_A, y_1:B_1, \ldots, y_n:B_n \vdash \OUTPUT{Y \cdot f}{ r, y_1, \ldots, y_n }$, with unused type assumptions removed by Lemma~\ref{wc:lemma:epi_simple_type_strengthening}.
    This must have been concluded by \nameref{epi_simple_type_out}, and from its premise we then get that
    \begin{equation*}
      (\PTRANS[\Gamma], r:\TRET, Y:I_A, y_1:B_1, \ldots, y_n:B_n); (\PTRANS[\Gamma], r:\TRET, Y:I_A, y_1:B_1, \ldots, y_n:B_n) \vdash Y \cdot f : \TCHAN{\TRET, B_1, \ldots, B_n}
    \end{equation*}

    \noindent which was concluded by \nameref{epi_simple_type_vec1} and \nameref{epi_simple_type_vec2}.
    Again, we omit the full derivation here, since it appears in the forward case. 
    However, it requires that $\PTRANS[\Gamma]$ must contain a type definition for $I_A$, i.e.:
    \begin{equation*}
      \PTRANS[\Gamma] = I_A:(\TNIL, (f:\TIFC{B_1, \ldots, B_n}, \PTRANS[\Delta]<2>)), \TIFC{B_1, \ldots, B_n}:(\TCHAN{\TRET, B_1, \ldots, B_n}, \EMPTYSET), \PTRANS[\Delta]<3>, \PTRANS[\Gamma']
    \end{equation*}

    \noindent and by the encoding of $\Gamma$, we have that 
    \begin{align*}
         & \PTRANS[I_A:(f:\TPROC{B_1, \ldots, B_n}, \Delta), \Gamma']                                                                                                              \\
      = ~& I_A:(\TNIL, (\PTRANS[f:\TPROC{B_1, \ldots, B_n}, \Delta]<2>)), \PTRANS[f:\TPROC{B_1, \ldots, B_n}, \Delta]<3>, \PTRANS[\Gamma']                                         \\
      = ~& I_A:(\TNIL, (f:\TIFC{B_1, \ldots, B_n}, \PTRANS[\Delta]<2>)), \TIFC{B_1, \ldots, B_n}:(\TCHAN{\TRET, B_1, \ldots, B_n}, \EMPTYSET), \PTRANS[\Delta]<3>, \PTRANS[\Gamma'] 
    \end{align*}

    Thus, $\Gamma$ can be written as $\Gamma = I_A:(f:\TPROC{B_1, \ldots, B_n}, \Delta), \Gamma'$, and clearly 
    \begin{equation*}
      (I_A:(f:\TPROC{B_1, \ldots, B_n}, \Delta), \Gamma')(I_A)(f) = \TPROC{B_1, \ldots, B_n}
    \end{equation*}

    Furthermore, by using Lemma~\ref{wc:lemma:type_correspondence_expr}, we conclude the desired
    \begin{align*}
      \Gamma \vdash e : I_A   
      \qquad
      \Gamma \vdash e_1 : B_1 
      \quad
      \ldots
      \quad
      \Gamma \vdash e_n : B_n
    \end{align*}
\end{itemize}

This concludes the proof for the other direction. 
\end{proof}

\ghostsubsection{Lemma: \ENV{M} type correspondence}
\begin{lemma}\label{wc:lemma:type_correspondence_envm}
Assume $\Gamma(A) = I_A$.
Then $\Gamma \vdash_A \ENV{M} \iff \PTRANS[\Gamma] \vdash \PTRANS[\ENV{M}](A)<\Gamma>$.
\end{lemma}

\begin{proof}
We have two directions to show.
For both, we proceed by induction on the structure of $\ENV{M}$.
For the forward direction, we must show that $\Gamma \vdash_A \ENV{M} \implies \PTRANS[\Gamma] \vdash \PTRANS[\ENV{M}](A)<\Gamma>$.
\begin{itemize}
  \item Suppose $\ENV{M} = \ENV{M}^\EMPTYSET$.
    Then $\Gamma \vdash_A \ENV{M}^\EMPTYSET$ was concluded by \nameref{wc_type_env_empty}, which is an axiom.
    The translation for environments then yields $\PTRANS[\ENV{M}^\EMPTYSET](A)<\Gamma> = \NIL$, and $\PTRANS[\Gamma] \vdash \NIL$ can then be concluded by \nameref{epi_simple_type_nil} for any type environment $\PTRANS[\Gamma]$.

  \item Suppose $\ENV{M} = (f, (\VEC{x}, S)), \ENV{M}'$.
    Then $\Gamma \vdash_A (f, (\VEC{x}, S)), \ENV{M}'$ must have been concluded by \nameref{wc_type_env_envm}, and from its premise we have that
    \begin{equation*}
      \Gamma(I_A)(f)                           = \TPROC{\VEC{B}} \qquad
      \Gamma, \code{this}:I_A, \VEC{x}:\VEC{B} \vdash S  \qquad
      \Gamma \vdash_A \ENV{M}' 
    \end{equation*}

    By the translation for environments, we have that
    \begin{align*}
         & \PTRANS[{\PAREN[big]{f, (\VEC{x}, S)}, \ENV{M}'}](A)<\Gamma> \\
      = ~& \PTRANS[\ENV{M}'](A)<\Gamma> \PAR \REPL{\INPUT{A \cdot f}{r, \VEC{a}}} . \LOCNEW{\VEC{x} : \TIFC{\VEC{B}}}{\VEC{a}} \LOCNEW{\code{this} : \TIFC{I_A}}{A} \PAREN{ \PTRANS[S](r)<\Gamma> } \\
      = ~& \PTRANS[\ENV{M}'](A)<\Gamma> \PAR \REPL{\INPUT{A \cdot f}{r, \VEC{a}}} . \NEW{x_1:\TIFC{B_1}, \ldots, x_n:\TIFC{B_n}, \code{this}:\TIFC{I_A}}\PAREN{ \OUTPUT{x_1}{a_1} \PAR \ldots \PAR \OUTPUT{x_n}{a_n} \PAR \OUTPUT{\code{this}}{A} \PAR \PTRANS[S](r)<\Gamma> } 
    \end{align*}

    \noindent where $\Gamma(A) = I_A$ and $\Gamma(I_A)(f) = \TPROC{\VEC{B}}$.
    By the induction hypothesis, we have that $\PTRANS[\Gamma] \vdash \PTRANS[\ENV{M}](A)<\Gamma>$.
    Thus, to allow us to conclude by \nameref{epi_simple_type_par}, what remains to be shown is that
    \begin{equation*}
      \PTRANS[\Gamma] \vdash \REPL{\INPUT{A \cdot f}{r, \VEC{a}}} . \NEW{x_1:\TIFC{B_1}, \ldots, x_n:\TIFC{B_n}, \code{this}:\TIFC{I_A}}\PAREN{ \OUTPUT{x_1}{a_1} \PAR \ldots \PAR \OUTPUT{x_n}{a_n} \PAR \OUTPUT{\code{this}}{A} \PAR \PTRANS[S](r)<\Gamma> }
    \end{equation*}

    We note that 
    \begin{equation*} 
      \PTRANS[\Gamma], r:\TRET, a_1:B_1, \ldots, a_n:B_n, x_1:\TIFC{B_1}, \ldots, x_n:\TIFC{B_n}, \code{this}:\TIFC{I_A} \vdash \OUTPUT{x_1}{a_1} \PAR \ldots \PAR \OUTPUT{x_n}{a_n} \PAR \OUTPUT{\code{this}}{A}
    \end{equation*}

    \noindent easily holds by using \nameref{epi_simple_type_out} and \nameref{epi_simple_type_par}.
    What remains to be shown is that
    \begin{equation*} 
      \PTRANS[\Gamma], x_1:\TIFC{B_1}, \ldots, x_n:\TIFC{B_n}, r:\TRET, \code{this}:\TIFC{I_A} \vdash \PTRANS[S](r)<\Gamma>
    \end{equation*}

    \noindent where the entries $a_1:B_1, \ldots, a_n:B_n$ have been removed by Lemma~\ref{wc:lemma:epi_simple_type_strengthening}, since none of them occur in $\PTRANS[S](r)<\Gamma>$.
    As we know that $\Gamma, \code{this}:I_A, \VEC{x}:\VEC{B} \vdash S$, we have by Lemma~\ref{wc:lemma:type_correspondence_stm} that 
    \begin{equation*}
      \Gamma, \VEC{x}:\VEC{B}, \code{this}:I_A \vdash S \implies \PTRANS[\Gamma, \VEC{x}:\VEC{B}, \code{this}:I_A], r:\TRET \vdash \PTRANS[S](r)<\Gamma>
    \end{equation*}

    \noindent so we must show that 
    \begin{equation*}
      \PTRANS[\Gamma, \VEC{x}:\VEC{B}, \code{this}:I_A] = \PTRANS[\Gamma], x_1:\TIFC{B_1}, \ldots, x_n:\TIFC{B_n}, r:\TRET
    \end{equation*}

    This is seen to be the case from the translation of $\Gamma$ (local variables), which gives $\PTRANS[x:B, \Gamma] = x:\TIFC{B}, \TIFC{B}:(\TCHAN{B}, \EMPTYSET), \PTRANS[\Gamma]$.
    In the case above, the interface definition $\TIFC{B}:(\TCHAN{B}, \EMPTYSET)$ is omitted, since it is already assumed to be in $\PTRANS[\Gamma]$, since we assume interface types are defined for each base type $B$ that is used.
    Thus, by using \nameref{epi_simple_type_res}, we can conclude 
    \begin{equation*}
      \PTRANS[\Gamma], r:\TRET, a_1:B_1, \ldots, a_n:B_n \vdash \NEW{x_1:\TIFC{B_1}, \ldots, x_n:\TIFC{B_n}, \code{this}:\TIFC{I_A}}\PAREN{ \OUTPUT{x_1}{a_1} \PAR \ldots \PAR \OUTPUT{x_n}{a_n} \PAR \OUTPUT{\code{this}}{A} \PAR \PTRANS[S](r)<\Gamma> }
    \end{equation*}

    Let 
    \begin{equation*}
      P \DEFSYM \NEW{x_1:\TIFC{B_1}, \ldots, x_n:\TIFC{B_n}, \code{this}:\TIFC{I_A}}\PAREN{ \OUTPUT{x_1}{a_1} \PAR \ldots \PAR \OUTPUT{x_n}{a_n} \PAR \OUTPUT{\code{this}}{A} \PAR \PTRANS[S](r)<\Gamma> }
    \end{equation*}

    \noindent for readability.
    What remains to be shown is that 
    \begin{equation*}
      \PTRANS[\Gamma] \vdash \REPL{\INPUT{A \cdot f}{r, \VEC{a}}} . P
    \end{equation*}

    \noindent which holds by \nameref{epi_simple_type_rep} and \nameref{epi_simple_type_in}, \emph{if}, by using \nameref{epi_simple_type_vec1}, we can conclude that
    \begin{equation*}
      \PTRANS[\Gamma]; \PTRANS[\Gamma] \vdash A \cdot p : \TCHAN{\TRET, \VEC{B}}
    \end{equation*}

    To show this, recall that we know that $\Gamma(A) = I_A$ and $\Gamma(I_A)(f) = \TPROC{\VEC{B}}$.
    Thus, we also know that the type environment can be expanded as $\Gamma = A:I_A, I_A:(f:\TPROC{\VEC{B}}, \Delta), \Gamma'$.
    \noindent and by the encoding of $\Gamma$ we have that 
    \begin{align*}
         & \PTRANS[A:I_A, I_A:(f:\TPROC{B_1, \ldots, B_n}, \Delta), \Gamma']                                                                                                              \\
      = ~& A:I_A, I_A:(\TNIL, (\PTRANS[f:\TPROC{B_1, \ldots, B_n}, \Delta]<2>)), \PTRANS[f:\TPROC{B_1, \ldots, B_n}, \Delta]<3>, \PTRANS[\Gamma']                                         \\
      = ~& A:I_A, I_A:(\TNIL, (f:\TIFC{B_1, \ldots, B_n}, \PTRANS[\Delta]<2>)), \TIFC{B_1, \ldots, B_n}:(\TCHAN{\TRET, B_1, \ldots, B_n}, \EMPTYSET), \PTRANS[\Delta]<3>, \PTRANS[\Gamma'] 
    \end{align*}

    \noindent which lets us conclude that 
    \begin{align*}
      {\PTRANS[A:I_A, I_A:(f:\TPROC{B_1, \ldots, B_n}, \Delta), \Gamma']}(A)                         & = I_A                                                      \\
      {\PTRANS[A:I_A, I_A:(f:\TPROC{B_1, \ldots, B_n}, \Delta), \Gamma']}(I_A)                       & = (\TNIL, (f:\TIFC{B_1, \ldots, B_n}, \PTRANS[\Delta]<2>)) \\
      \snd (\TNIL, (f:\TIFC{B_1, \ldots, B_n}, \PTRANS[\Delta]<2>))                                  & = (f:\TIFC{B_1, \ldots, B_n}, \PTRANS[\Delta]<2>)          \\
      (f:\TIFC{B_1, \ldots, B_n}, \PTRANS[\Delta]<2>)(f)                                             & = \TIFC{B_1, \ldots, B_n}                                  \\
      {\PTRANS[A:I_A, I_A:(f:\TPROC{B_1, \ldots, B_n}, \Delta), \Gamma']}(\TIFC{B_1, \ldots, B_n})   & = (\TCHAN{\TRET, B_1, \ldots, B_n}, \EMPTYSET)             \\
      \fst (\TCHAN{\TRET, B_1, \ldots, B_n}, \EMPTYSET)                                              & = \TCHAN{\TRET, B_1, \ldots, B_n} 
    \end{align*}

    \noindent as required.
\end{itemize}

This concludes the proof for the forward direction.

For the other direction, we must show that $\PTRANS[\Gamma] \vdash \PTRANS[\ENV{M}](A)<\Gamma> \implies \Gamma \vdash_A \ENV{M}$.
\begin{itemize}
  \item Suppose $\ENV{M} = \ENV{M}^\EMPTYSET$.
    The translation for environments  yields $\PTRANS[\ENV{M}^\EMPTYSET](A)<\Gamma> = \NIL$, and $\PTRANS[\Gamma] \vdash \NIL$ was then concluded by \nameref{epi_simple_type_nil}, which holds for any type environment $\PTRANS[\Gamma]$.
    Then $\Gamma \vdash_A \ENV{M}^\EMPTYSET$ can be concluded by \nameref{wc_type_env_empty}.

  \item Suppose $\ENV{M} = (f, (\VEC{x}, S)), \ENV{M}'$.
    By the translation for environments, we have that
    \begin{align*}
         & \PTRANS[{\PAREN[big]{f, (\VEC{x}, S)}, \ENV{M}'}](A)<\Gamma> \\
      = ~& \PTRANS[\ENV{M}'](A)<\Gamma> \PAR \REPL{\INPUT{A \cdot f}{r, \VEC{a}}} . \NEW{x_1:\TIFC{B_1}, \ldots, x_n:\TIFC{B_n}, \code{this}:\TIFC{I_A}}\PAREN{ \OUTPUT{x_1}{a_1} \PAR \ldots \PAR \OUTPUT{x_n}{a_n} \PAR \OUTPUT{\code{this}}{A} \PAR \PTRANS[S](r)<\Gamma> } 
    \end{align*}

    \noindent where $\Gamma(A) = I_A$ and $\Gamma(I_A)(f) = \TPROC{\VEC{B}}$.
    Our goal is to conclude by \nameref{wc_type_env_envm}, which requires us to show that the following premises are satisfied:
    \begin{equation*}
      \Gamma(A)                                = I_A             \qquad
      \Gamma(I_A)(f)                           = \TPROC{\VEC{B}} \qquad
      \Gamma, \code{this}:I_A, \VEC{x}:\VEC{B} \vdash S          \qquad
      \Gamma \vdash_A \ENV{M}' 
    \end{equation*}

    We know that 
    \begin{equation*}
      \PTRANS[\Gamma] \vdash \PTRANS[\ENV{M}'](A)<\Gamma> \PAR \REPL{\INPUT{A \cdot f}{r, \VEC{a}}} . \NEW{x_1:\TIFC{B_1}, \ldots, x_n:\TIFC{B_n}, \code{this}:\TIFC{I_A}}\PAREN{ \OUTPUT{x_1}{a_1} \PAR \ldots \PAR \OUTPUT{x_n}{a_n} \PAR \OUTPUT{\code{this}}{A} \PAR \PTRANS[S](r)<\Gamma> }
    \end{equation*}

    \noindent which must have been concluded by \nameref{epi_simple_type_par}, and from its premise we then get that 
    \begin{align*}
      \PTRANS[\Gamma] & \vdash \PTRANS[\ENV{M}'](A)<\Gamma> \\
      \PTRANS[\Gamma] & \vdash \REPL{\INPUT{A \cdot f}{r, \VEC{a}}} . \NEW{x_1:\TIFC{B_1}, \ldots, x_n:\TIFC{B_n}, \code{this}:\TIFC{I_A}}\PAREN{ \OUTPUT{x_1}{a_1} \PAR \ldots \PAR \OUTPUT{x_n}{a_n} \PAR \OUTPUT{\code{this}}{A} \PAR \PTRANS[S](r)<\Gamma> }
    \end{align*}

    By the induction hypothesis applied to the former, we obtain $\Gamma \vdash_A \ENV{M}'$.
    The other premise must have been concluded by a combination of \nameref{epi_simple_type_rep}, \nameref{epi_simple_type_in}, \nameref{epi_simple_type_res} and \nameref{epi_simple_type_par}.
    We omit the lengthy derivation here, but see the corresponding case for the forward direction.
    However, it yields a premise
    \begin{equation*} 
      \PTRANS[\Gamma], x_1:\TIFC{B_1}, \ldots, x_n:\TIFC{B_n}, r:\TRET, \code{this}:\TIFC{I_A} \vdash \PTRANS[S](r)<\Gamma>
    \end{equation*}

    \noindent for one of the applications of \nameref{epi_simple_type_par}, and by Lemma~\ref{wc:lemma:type_correspondence_stm} we have that 
    \begin{equation*} 
      \PTRANS[\Gamma, \VEC{x}:\VEC{B}, \code{this}:I_A], r:\TRET \vdash \PTRANS[S](r)<\Gamma> \implies \Gamma, \VEC{x}:\VEC{B}, \code{this}:I_A \vdash S
    \end{equation*}

    Finally, typing the input $\INPUT{A \cdot f}{r, \VEC{a}}.P$ with the continuation 
    \begin{equation*}
      P = \NEW{x_1:\TIFC{B_1}, \ldots, x_n:\TIFC{B_n}, \code{this}:\TIFC{I_A}}\PAREN{ \OUTPUT{x_1}{a_1} \PAR \ldots \PAR \OUTPUT{x_n}{a_n} \PAR \OUTPUT{\code{this}}{A} \PAR \PTRANS[S](r)<\Gamma> }
    \end{equation*}

    \noindent is likewise shown to require the subject $A \cdot p$ to be typable as $\TCHAN{\TRET, B_1, \ldots, B_n}$, which is concluded by \nameref{epi_simple_type_vec1} and \nameref{epi_simple_type_vec2}.
    It must therefore be the case that $\PTRANS[\Gamma] = 
      A:I_A, I_A:(\TNIL, (f:\TIFC{B_1, \ldots, B_n}, \PTRANS[\Delta]<2>)), \TIFC{B_1, \ldots, B_n}:(\TCHAN{\TRET, B_1, \ldots, B_n}, \EMPTYSET), \PTRANS[\Delta]<3>, \PTRANS[\Gamma']$ and, by the encoding of $\Gamma$:
    \begin{align*}
         & \PTRANS[A:I_A, I_A:(f:\TPROC{B_1, \ldots, B_n}, \Delta), \Gamma']                                                                                                              \\
      = ~& A:I_A, I_A:(\TNIL, (f:\TIFC{B_1, \ldots, B_n}, \PTRANS[\Delta]<2>)), \TIFC{B_1, \ldots, B_n}:(\TCHAN{\TRET, B_1, \ldots, B_n}, \EMPTYSET), \PTRANS[\Delta]<3>, \PTRANS[\Gamma'] 
    \end{align*}

    Thus, $\Gamma = A:I_A, I_A:(f:\TPROC{B_1, \ldots, B_n}, \Delta), \Gamma'$, and clearly, 
    \begin{align*}
      (A:I_A, I_A:(f:\TPROC{B_1, \ldots, B_n}, \Delta), \Gamma')(A)      & = I_A \\
      (A:I_A, I_A:(f:\TPROC{B_1, \ldots, B_n}, \Delta), \Gamma')(I_A)(f) & = \TPROC{B_1, \ldots, B_n} 
    \end{align*}
     
     Thus, all the premises of \nameref{wc_type_env_envm} are satisfied, and we can therefore conclude that $\Gamma \vdash (f, (\VEC{x}, S)), \ENV{M}'$, as desired.
\end{itemize}

This concludes the proof for the other direction.
\end{proof}

\ghostsubsection{Lemma: \ENV{T} type correspondence}
\begin{lemma}\label{wc:lemma:type_correspondence_envt}
$\Gamma \vdash \ENV{T} \iff \PTRANS[\Gamma] \vdash \PTRANS[\ENV{T}]<\Gamma>$.
\end{lemma}

\begin{proof}
We have two directions to show.
For both, we proceed by induction on the structure of $\ENV{T}$.
For the forward direction, we must show that $\Gamma \vdash \ENV{T} \implies \PTRANS[\Gamma] \vdash \PTRANS[\ENV{T}]<\Gamma>$.
\begin{itemize}
  \item Suppose $\ENV{T} = \ENV{T}^\EMPTYSET$.
    Then $\Gamma \vdash \ENV{T}^\EMPTYSET$ was concluded by \nameref{wc_type_env_empty}, which is an axiom.
    The translation for environments then yields $\PTRANS[\ENV{T}^\EMPTYSET]<\Gamma> = \NIL$, and $\PTRANS[\Gamma] \vdash \NIL$ can then be concluded by \nameref{epi_simple_type_nil} for any type environment $\PTRANS[\Gamma]$.

  \item Suppose $\ENV{T} = (A, \ENV{M}), \ENV{T}'$.
    Then $\Gamma \vdash (A, \ENV{M}), \ENV{T}'$ must have been concluded by \nameref{wc_type_env_envt}, and from its premise we have that $\Gamma \vdash_A \ENV{M}$ and $\Gamma \vdash \ENV{T}'$, where $\Gamma(A) = I_A$ for some interface type $I_A$.
    By the encoding for environments, we have that 
    \begin{equation*}
      \PTRANS[(A, \ENV{M}),\ENV{T}']<\Gamma> = \PTRANS[\ENV{T}']<\Gamma> \PAR \PTRANS[\ENV{M}](A)<\Gamma> 
    \end{equation*}

    \noindent By the induction hypothesis, we have that $\PTRANS[\Gamma] \vdash \PTRANS[\ENV{T}']$ and, by Lemma~\ref{wc:lemma:type_correspondence_envm}, that $\PTRANS[\Gamma] \vdash \PTRANS[\ENV{M}](A)<\Gamma>$. Hence 
    \begin{equation*}
      \PTRANS[\Gamma] \vdash \PTRANS[\ENV{T}']<\Gamma> \PAR \PTRANS[\ENV{M}](A)<\Gamma>
    \end{equation*}

    \noindent can be concluded by \nameref{epi_simple_type_par} as desired. 
\end{itemize}

This concludes the proof for the forward direction.

For the other direction, we must show that $\PTRANS[\Gamma] \vdash \PTRANS[\ENV{T}]<\Gamma> \implies \Gamma \vdash \ENV{T}$.
\begin{itemize}
  \item Suppose $\ENV{T} = \ENV{T}^\EMPTYSET$.
    The translation for environments then yields $\PTRANS[\ENV{T}^\EMPTYSET]<\Gamma> = \NIL$, and $\PTRANS[\Gamma] \vdash \NIL$ must therefore have been concluded by \nameref{epi_simple_type_nil}, which is an axiom and holds for any type environment $\PTRANS[\Gamma]$.
    Then $\Gamma \vdash \ENV{T}^\EMPTYSET$ can be concluded by \nameref{wc_type_env_empty}.

  \item Suppose $\ENV{T} = (A, \ENV{M}), \ENV{T}'$.
    By the encoding for environments, we have that 
    \begin{equation*}
      \PTRANS[(A, \ENV{M}),\ENV{T}']<\Gamma> = \PTRANS[\ENV{T}']<\Gamma> \PAR \PTRANS[\ENV{M}](A)<\Gamma> 
    \end{equation*}

    \noindent and $\PTRANS[\Gamma] \vdash \PTRANS[\ENV{T}']<\Gamma> \PAR \PTRANS[\ENV{M}](A)<\Gamma>$ must therefore have been concluded by \nameref{epi_simple_type_par}.
    From its premise, we then know that $\PTRANS[\Gamma] \vdash \PTRANS[\ENV{T}']<\Gamma>$ and $\PTRANS[\Gamma] \vdash \PTRANS[\ENV{M}](A)<\Gamma>$.
    By the induction hypothesis on the former and Lemma~\ref{wc:lemma:type_correspondence_envm} to the latter, we obtain $\Gamma \vdash \ENV{T}'$ and $\Gamma \vdash_A \ENV{M}$.
    Then we can conclude $\Gamma \vdash (A, \ENV{M}), \ENV{T}'$ by \nameref{wc_type_env_envt}, as desired.
\end{itemize}

This concludes the proof for the other direction.
\end{proof}

\ghostsubsection{Lemma: \ENV{F} type correspondence}
\begin{lemma}\label{wc:lemma:type_correspondence_envf}
Assume $\Gamma(A) = I_A$.
Then $\Gamma \vdash_A \ENV{F} \iff \PTRANS[\Gamma] \vdash \PTRANS[\ENV{F}](A)$.
\end{lemma}

\begin{proof}
We have two directions to show.
For both, we proceed by induction on the structure of $\ENV{F}$.
For the forward direction, we must show that $\Gamma \vdash_A \ENV{F} \implies \PTRANS[\Gamma] \vdash \PTRANS[\ENV{F}](A)$.
\begin{itemize}
  \item Suppose $\ENV{F} = \ENV{F}^\EMPTYSET$.
    Then $\Gamma \vdash \ENV{F}^\EMPTYSET$ was concluded by \nameref{wc_type_env_empty}, which is an axiom.
    The translation for environments then yields $\PTRANS[\ENV{F}^\EMPTYSET] = \NIL$, and $\PTRANS[\Gamma] \vdash \NIL$ can then be concluded by \nameref{epi_simple_type_nil} for any type environment $\PTRANS[\Gamma]$.

  \item Suppose $\ENV{F} = (p,v), \ENV{F}'$.
    Then $\Gamma \vdash_A (p,v), \ENV{F}'$ must have been concluded by \nameref{wc_type_env_envf}, and from its premise we have that
    \begin{align*}
      \Gamma(A) = I_A \qquad
      \Gamma(I_A)(p) = B \qquad
      \Gamma \vdash v : B \qquad
      \Gamma \vdash_A \ENV{F}' 
    \end{align*}

    By the encoding of environments, we have that
    \begin{equation*}
      \PTRANS[(p,v), \ENV{F}'](A) = \PTRANS[\ENV{F}'](A) \PAR \OUTPUT{A \cdot p}{v}
    \end{equation*}

    By Lemma~\ref{wc:lemma:type_correspondence_val}, we have that $\PTRANS[\Gamma] \vdash v : B$ and, by the induction hypothesis, that $\PTRANS[\Gamma] \vdash \PTRANS[\ENV{F}']$.
    What remains to be shown is therefore that $\PTRANS[\Gamma]; \PTRANS[\Gamma] \vdash A \cdot p : \TCHAN{B}$.
    As we know that $\Gamma(A) = I_A$ and $\Gamma(I_A)(p) = B$, we can write $\Gamma$ as $A:I_A, I_A:(p : B, \Delta), \Gamma'$ for some $\Delta$.
    By the encoding of $\Gamma$, we then have that 
    \begin{equation*}
      \PTRANS[A:I_A, I_A:(p : B, \Delta), \Gamma'] = A:I_A, I_A:(\TNIL, (p:\TIFC{B}, \PTRANS[\Delta]<2>)), \TIFC{B}:(\TCHAN{B}, \EMPTYSET), \PTRANS[\Delta]<3>, \PTRANS[\Gamma']
    \end{equation*}

    \noindent and we can therefore conclude
    \begin{align*}
      (A:I_A, I_A:(\TNIL, (p:\TIFC{B}, \PTRANS[\Delta]<2>)), \TIFC{B}:(\TCHAN{B}, \EMPTYSET), \PTRANS[\Delta]<3>, \PTRANS[\Gamma'])(A)        & = I_A                                       \\
      (A:I_A, I_A:(\TNIL, (p:\TIFC{B}, \PTRANS[\Delta]<2>)), \TIFC{B}:(\TCHAN{B}, \EMPTYSET), \PTRANS[\Delta]<3>, \PTRANS[\Gamma'])(I_A)      & = (\TNIL, (p:\TIFC{B}, \PTRANS[\Delta]<2>)) \\
      \snd (\TNIL, (p:\TIFC{B}, \PTRANS[\Delta]<2>))                                                                                          & = (p:\TIFC{B}, \PTRANS[\Delta]<2>)          \\
      (p:\TIFC{B}, \PTRANS[\Delta]<2>)(p)                                                                                                     & = \TIFC{B}                                  \\
      (A:I_A, I_A:(\TNIL, (p:\TIFC{B}, \PTRANS[\Delta]<2>)), \TIFC{B}:(\TCHAN{B}, \EMPTYSET), \PTRANS[\Delta]<3>, \PTRANS[\Gamma'])(\TIFC{B}) & = (\TCHAN{B}, \EMPTYSET)                    \\
      \fst (\TCHAN{B}, \EMPTYSET)                                                                                                             & = \TCHAN{B}
    \end{align*}

    \noindent by \nameref{epi_simple_type_vec1} and \nameref{epi_simple_type_vec2}.
    Then, by \nameref{epi_simple_type_out} and  \nameref{epi_simple_type_par}, we can conclude 
    \begin{equation*}
      \PTRANS[\Gamma] \vdash \PTRANS[\ENV{F}'](A) \PAR \OUTPUT{A \cdot p}{v}
    \end{equation*}

    \noindent as desired.
\end{itemize}

This concludes the proof for the forward direction.

For the other direction, we must show that $\PTRANS[\Gamma] \vdash \PTRANS[\ENV{F}](A) \implies \Gamma \vdash_A \ENV{F}$.
\begin{itemize}
  \item Suppose $\ENV{F} = \ENV{F}^\EMPTYSET$.
    The translation for environments yields $\PTRANS[\ENV{F}^\EMPTYSET] = \NIL$, and $\PTRANS[\Gamma] \vdash \NIL$ was concluded by \nameref{epi_simple_type_nil} for any type environment $\PTRANS[\Gamma]$.
    Then $\Gamma \vdash \ENV{F}^\EMPTYSET$ can be concluded by \nameref{wc_type_env_empty}.

  \item Suppose $\ENV{F} = (p,v), \ENV{F}'$.
    By the encoding of environments, we have that
    \begin{equation*}
      \PTRANS[(p,v), \ENV{F}'](A) = \PTRANS[\ENV{F}'](A) \PAR \OUTPUT{A \cdot p}{v}
    \end{equation*}

    \noindent so $\PTRANS[\Gamma] \vdash \PTRANS[\ENV{F}'](A) \PAR \OUTPUT{A \cdot p}{v}$ must have been concluded by \nameref{epi_simple_type_par}, and from its premise we have that $\PTRANS[\Gamma] \vdash \PTRANS[\ENV{F}'](A)$ and $\PTRANS[\Gamma] \vdash \OUTPUT{A \cdot p}{v}$. by the induction hypothesis on the former, we have that $\Gamma \vdash_A \ENV{F}'$.
    The latter must have been concluded by \nameref{epi_early_out}, and from its premise we then get that $\PTRANS[\Gamma]; \PTRANS[\Gamma] \vdash A \cdot p : \TCHAN{B}$ and $\PTRANS[\Gamma] \vdash v : B$. by Lemma~\ref{wc:lemma:type_correspondence_val} to the latter, we conclude $\Gamma \vdash v : B$.
    
    What remains to be shown is that $\Gamma(A) = I_A$ and $\Gamma(I_A)(p) = B$.
    As we know that $\PTRANS[\Gamma]; \PTRANS[\Gamma] \vdash A \cdot p : \TCHAN{B}$ was concluded by \nameref{epi_simple_type_vec1} and \nameref{epi_simple_type_vec2}, we also know that $\PTRANS[\Gamma]$ must contain the entries
    \begin{equation*}
      \PTRANS[\Gamma] = A:I_A, I_A:(\TNIL, (p:\TIFC{B}, \PTRANS[\Delta]<2>)), \TIFC{B}:(\TCHAN{B}, \EMPTYSET), \PTRANS[\Delta]<3>, \PTRANS[\Gamma']
    \end{equation*}

    \noindent (see the derivation in the corresponding case in the forward direction).
    Hence, by the translation of $\Gamma$ we then have that 
    \begin{equation*}
      \PTRANS[A:I_A, I_A:(p : B, \Delta), \Gamma'] = A:I_A, I_A:(\TNIL, (p:\TIFC{B}, \PTRANS[\Delta]<2>)), \TIFC{B}:(\TCHAN{B}, \EMPTYSET), \PTRANS[\Delta]<3>, \PTRANS[\Gamma']
    \end{equation*}

    \noindent and clearly $(A:I_A, I_A:(p : B, \Delta), \Gamma')(A) = I_A$ and $(A:I_A, I_A:(p : B, \Delta), \Gamma')(I_A)(p) = B$.
    This allows us to conclude $\Gamma \vdash_A (p,v), \ENV{F}'$ by \nameref{wc_type_env_envf} as desired.
\end{itemize}

This concludes the proof for the other direction.
\end{proof}

\ghostsubsection{Lemma: \ENV{P} type correspondence}
\begin{lemma}\label{wc:lemma:type_correspondence_envs}
$\Gamma \vdash \ENV{S} \iff \PTRANS[\Gamma] \vdash \PTRANS[\ENV{S}]$.
\end{lemma}

\begin{proof}
We have two directions to show.
For both, we proceed by induction on the structure of $\ENV{S}$.
For the forward direction, we must show that $\Gamma \vdash \ENV{S} \implies \PTRANS[\Gamma] \vdash \PTRANS[\ENV{S}]$.
\begin{itemize}
   \item Suppose $\ENV{S} = \ENV{S}^\EMPTYSET$.
    Then $\Gamma \vdash \ENV{S}^\EMPTYSET$ was concluded by \nameref{wc_type_env_empty}, which is an axiom.
    The translation for environments then yields $\PTRANS[\ENV{S}^\EMPTYSET] = \NIL$, and $\PTRANS[\Gamma] \vdash \NIL$ can then be concluded by \nameref{epi_simple_type_nil} for any type environment $\PTRANS[\Gamma]$.

  \item Suppose $\ENV{S} = (A, \ENV{F}, \ENV{S}')$.
    Then $\Gamma \vdash (A, \ENV{F}), \ENV{S}'$ must have been concluded by \nameref{wc_type_env_envs}, and from its premise we have that $\Gamma \vdash_A \ENV{F}$ and $\Gamma \vdash \ENV{S}'$, where $\Gamma(A) = I_A$ for some interface type $I_A$.
    By the encoding of environments, we have that 
    \begin{equation*}
      \PTRANS[(A, \ENV{F}),\ENV{S}']  = \PTRANS[\ENV{S}'] \PAR \PTRANS[\ENV{F}](A)
    \end{equation*}

    By the induction hypothesis, we have that $\PTRANS[\Gamma] \vdash \PTRANS[\ENV{S}']$ and, by Lemma~\ref{wc:lemma:type_correspondence_envf}, that $\PTRANS[\Gamma] \vdash \PTRANS[\ENV{F}](A)$. 
    Hence 
    \begin{equation*}
      \PTRANS[\Gamma] \vdash \PTRANS[\ENV{S}']<\Gamma> \PAR \PTRANS[\ENV{F}](A)
    \end{equation*}

    \noindent can be concluded by \nameref{epi_simple_type_par} as desired.
\end{itemize}

This concludes the proof for the forward direction.

For the other direction, we must show that $\PTRANS[\Gamma] \vdash \PTRANS[\ENV{S}] \implies \Gamma \vdash \ENV{S}$.
\begin{itemize}
   \item Suppose $\ENV{S} = \ENV{S}^\EMPTYSET$.
    Then $\PTRANS[\ENV{S}^\EMPTYSET] = \NIL$, and $\PTRANS[\Gamma] \vdash \NIL$ was concluded by \nameref{epi_simple_type_nil}, which is an axiom and holds for any \PTRANS[\Gamma].
    Then $\Gamma \vdash \ENV{S}^\EMPTYSET$ can be concluded by \nameref{wc_type_env_empty}.

  \item Suppose $\ENV{S} = (A, \ENV{F}, \ENV{S}')$.
    Then $\PTRANS[(A, \ENV{F}),\ENV{S}']  = \PTRANS[\ENV{S}'] \PAR \PTRANS[\ENV{F}](A)$, and 
    \begin{equation*}
      \PTRANS[\Gamma] \vdash \PTRANS[(A, \ENV{F}),\ENV{S}']  = \PTRANS[\ENV{S}'] \PAR \PTRANS[\ENV{F}](A) 
    \end{equation*}

    \noindent was concluded by \nameref{epi_simple_type_par}, and from its premise we know that $\PTRANS[\Gamma] \vdash \PTRANS[\ENV{S}']$ and $\PTRANS[\Gamma] \vdash \PTRANS[\ENV{F}](A)$.
    By the induction hypothesis on the former and Lemma~\ref{wc:lemma:type_correspondence_envf} on the latter, 
    we have that $\Gamma \vdash \ENV{S}'$ and $\Gamma \vdash_A \ENV{F}$.
    Then we can conclude $\Gamma \vdash (A, \ENV{F}), \ENV{S}'$ by \nameref{wc_type_env_envs}, as desired.
\end{itemize}

This concludes the proof for the other direction.
\end{proof}

\ghostsubsection{Lemma: \ENV{V} type correspondence}
\begin{lemma}\label{wc:lemma:type_correspondence_envv}
Assume $\PTRANS[\Gamma] \vdash P$ for some process $P$.
Then $\Gamma \vdash \ENV{V} \iff \PTRANS[\Gamma] \vdash \PTRANS[\ENV{V}]<\Gamma>\HOLE[P]$.
\end{lemma}

\begin{proof}
We have two directions to show.
For both, we proceed by induction on the structure of $\ENV{V}$.
For the forward direction, we must show that $\Gamma \vdash \ENV{V} \implies \PTRANS[\Gamma] \vdash \PTRANS[\ENV{V}]<\Gamma>\HOLE[P]$.
\begin{itemize}
  \item Suppose $\ENV{V} = \ENV{V}^\EMPTYSET$.
    Then $\Gamma \vdash \ENV{V}^\EMPTYSET$ was concluded by \nameref{wc_type_env_empty}, which is an axiom.
    The translation for environments then yields $\PTRANS[\ENV{V}^\EMPTYSET]<\Gamma> = \HOLE$, hence $\PTRANS[\ENV{V}^\EMPTYSET]<\Gamma>\HOLE[P] = P$, and $\PTRANS[\Gamma] \vdash P$ holds by assumption.
  
  \item Suppose $\ENV{V} = (x,v), \ENV{V}'$.
    Then $\Gamma \vdash (x,v), \ENV{V}'$ must have been concluded by \nameref{wc_type_env_envv}, and from its premise we have that $\Gamma(x) = B$ and $\Gamma \vdash v : B$ and $\Gamma \vdash \ENV{V}'$.

    By the encoding of environments, we have that 
    \begin{align*}
         & \PTRANS[(x,v), \ENV{V}']<\Gamma>                                       \\
      = ~& \LOCNEW{x :\TIFC{\Gamma(x)}}{v}\PAREN{ \PTRANS[\ENV{V}']<\Gamma> }     \\
      = ~& \NEW{x:\TIFC{B}}\PAREN{ \OUTPUT{x}{v} \PAR \PTRANS[\ENV{V}']<\Gamma> } 
    \end{align*}

    \noindent and since $\Gamma(x) = B$, we can expand $\Gamma$ as $\Gamma = x:B, \Gamma'$.
    The translation for $\Gamma$ then yields
    \begin{equation*}
      \PTRANS[x:B, \Gamma'] = x:\TIFC{B}, \TIFC{B}:(\TCHAN{B}, \EMPTYSET), \PTRANS[\Gamma']
    \end{equation*}

    \noindent so we know that the interface definition $\TIFC{B}:(\TCHAN{B}, \EMPTYSET)$ is present in $\PTRANS[\Gamma]$.
    Thus we can conclude
    \begin{align*}
      \PTRANS[\Gamma]             & \vdash \PTRANS[\ENV{V}']<\Gamma>\HOLE[P]                                       & \text{by the induction hypothesis} \\
      \PTRANS[\Gamma], x:\TIFC{B} & \vdash \OUTPUT{x}{v} \PAR \PTRANS[\ENV{V}']<\Gamma>\HOLE[P]                    & \text{by \nameref{epi_simple_type_out} and \nameref{epi_simple_type_par}} \\
      \PTRANS[\Gamma]             & \NEW{x:\TIFC{B}}\PAREN{ \OUTPUT{x}{v} \PAR \PTRANS[\ENV{V}']<\Gamma>\HOLE[P] } & \text{by \nameref{epi_simple_type_res}}
    \end{align*}

    \noindent as desired.

    A technical detail:
    The entry $x:\TIFC{B}$ still appears in $\PTRANS[\Gamma]$, even though the name is bound in $\PTRANS[\ENV{V}]<\Gamma>$, and can thus be removed afterwards.
    This happens, because the actual type assignment is transferred to the encoding of $\ENV{V}$ via the parameter $\Gamma$.
    Hence, even though there might exist a different $\Gamma'$, such that $\PTRANS[\Gamma']$ does not contain the entry $x:\TIFC{V}$, and $\PTRANS[\Gamma'] \vdash \PTRANS[\ENV{V}]<\Gamma>\HOLE[P]$ still holds, then $\PTRANS[\Gamma'] \vdash \PTRANS[\ENV{V}]<\Gamma'>\HOLE[P]$ would \emph{not} hold.
\end{itemize}

This concludes the proof for the forward direction.

For the other direction, we must show that $\PTRANS[\Gamma] \vdash \PTRANS[\ENV{V}]<\Gamma>\HOLE[P] \implies \Gamma \vdash \ENV{V}$.
\begin{itemize}
  \item Suppose $\ENV{V} = \ENV{V}^\EMPTYSET$, so $\PTRANS[\ENV{V}^\EMPTYSET]<\Gamma> = \HOLE$, and $\PTRANS[\Gamma] \vdash \PTRANS[\ENV{V}^\EMPTYSET]<\Gamma>\HOLE[P]$ therefore holds by assumption.
    Then $\Gamma \vdash \ENV{V}^\EMPTYSET$ can be concluded by \nameref{wc_type_env_empty} for any $\Gamma$.

  \item Suppose $\ENV{V} = (x_1,v_1), \ldots, (x_n, v_n), \ENV{V}^\EMPTYSET$, so $\PTRANS[\Gamma] \vdash \PTRANS[(x_1,v_1), \ldots, (x_n, v_n), \ENV{V}^\EMPTYSET]<\Gamma>\HOLE[P]$.
    By the encoding of $\ENV{V}$, we have that it is translated as a process context consisting of a list of scopes and bound outputs
    \begin{equation*}
      \NEW{x_1 : \TIFC{\Gamma(x_1)}, \ldots, x_n : \TIFC{\Gamma(x_n)}}\PAREN{ \OUTPUT{x_1}{v_1} \PAR \ldots \PAR \OUTPUT{x_n}{v_n} \PAR \HOLE }
    \end{equation*}

    \noindent So, $\PTRANS[\Gamma] \vdash \PTRANS[\ENV{V}]<\Gamma>\HOLE[P]$ is concluded by \nameref{epi_simple_type_res} and \nameref{epi_simple_type_par}, and from the premise of the latter, we get that 
    \begin{equation*}
      \PTRANS[\Gamma], x_1 : \TIFC{\Gamma(x_1)}, \ldots, x_n : \TIFC{\Gamma(x_n)} \vdash \OUTPUT{x_1}{v_1} \PAR \ldots \PAR \OUTPUT{x_n}{v_n} \PAR P
    \end{equation*}

    However, by Lemma~\ref{wc:lemma:type_correspondence_expr} applied to $\PTRANS[\Gamma], x:\TIFC{\Gamma(x)} \vdash x$, we also know that $\Gamma, x:\Gamma(x) \vdash x:\Gamma(x)$, hence $\Gamma$ must already contain a type entry $B$ for $x$, which is used in the encoding of $\ENV{V}$ given above.
    Then by the encoding of $\Gamma$
    \begin{equation*}
      \PTRANS[x_1:B_1, \ldots, x_n:B_n, \Gamma'] = x_1:\TIFC{B_1}, \TIFC{B_1}:(\TCHAN{B_1}, \EMPTYSET), \ldots, x_n:\TIFC{B_n}, \TIFC{B_n}:(\TCHAN{B_n}, \EMPTYSET), \PTRANS[\Gamma']
    \end{equation*}

    \noindent so the type entries $x_1:\TIFC{B_1}, \ldots, x_n:\TIFC{B_n}$ are already in $\PTRANS[\Gamma]$.

    Also by the premise of \nameref{epi_simple_type_out}, we have that $\PTRANS[\Gamma] \vdash v_i : B_i$ for each of the object values, where the corresponding subject has type $\TIFC{B_i}$, since the definition of each such interface is $\TIFC{B_i} \mapsto (\TCHAN{B_i}, \EMPTYSET)$.

    Then, by the translation of $\Gamma$, we have that entries $x_i:B_i$ are translated as $x_i:\TIFC{B_i}$, so we know that $\Gamma$ can be written as $\Gamma = x_1:B_1, \ldots, x_n:B_n, \Gamma'$.
    Then by repeated applications of \nameref{wc_type_env_envv}, we can therefore conclude that 
    \begin{equation*}
      x_1:B_1, \ldots, x_n:B_n, \Gamma' \vdash (x_1,v_1), \ldots, (x_n,v_n), \ENV{V}^\EMPTYSET
    \end{equation*}

    \noindent as desired.
\end{itemize}

This concludes the proof for the other direction.
\end{proof}

\ghostsubsection{Proof of the Type Correspondence Theorem}
Finally, we can prove the type correspondence theorem:
\begin{proof}[Proof of Theorem~\ref{wc:thm:type_correspondence}]
We have two directions to prove.
For the forward direction, we must show that $\Gamma \vdash \ENV{T}$ and $\Gamma \vdash \ENV{S}$ and $\Gamma \vdash \ENV{V}$ and $\Gamma \vdash S$ together imply that $\PTRANS[\Gamma] \vdash \PTRANS[\CONF{S, \ENV{TSV}}]<\Gamma>$.
The translation for initial configurations yields:
\begin{equation*}
  \PTRANS[\CONF{S, \ENV{TSV}}]<\Gamma> = \PTRANS[\ENV{T}]<\Gamma> \PAR \PTRANS[\ENV{S}] \PAR \NEW{r : \TRET}\PAREN{ \PTRANS[\ENV{V}]<\Gamma>\HOLE[{\PTRANS[S](r)<\Gamma>}] \PAR \INPUT{r}{}.\NIL }
\end{equation*}

Assume $\PTRANS[\Gamma] \vdash P$ for some process $P$.
We then have the following:
\begin{align*}
  \Gamma \vdash S       & \implies \PTRANS[\Gamma], r:\TRET \vdash \PTRANS[S](r)<\Gamma>   & \text{by Lemma~\ref{wc:lemma:type_correspondence_stm}}  \\
  \Gamma \vdash \ENV{T} & \implies \PTRANS[\Gamma] \vdash \PTRANS[\ENV{T}]<\Gamma>         & \text{by Lemma~\ref{wc:lemma:type_correspondence_envt}} \\
  \Gamma \vdash \ENV{S} & \implies \PTRANS[\Gamma] \vdash \PTRANS[\ENV{S}]                 & \text{by Lemma~\ref{wc:lemma:type_correspondence_envs}} \\
  \Gamma \vdash \ENV{V} & \implies \PTRANS[\Gamma] \vdash \PTRANS[\ENV{V}]<\Gamma>\HOLE[P] & \text{by Lemma~\ref{wc:lemma:type_correspondence_envv}}
\end{align*}

We can then conclude the following:
\begin{align*}
  \PTRANS[\Gamma] & \vdash \NEW{r:\TRET}\PAREN{ \PTRANS[S](r)<\Gamma> \PAR \INPUT{r}{}.\NIL }  & \text{by \nameref{epi_simple_type_in}, \nameref{epi_simple_type_par}, \nameref{epi_simple_type_res}} \\
  \PTRANS[\Gamma] & \vdash \PTRANS[\ENV{V}]<\Gamma>\HOLE[{ \NEW{r:\TRET}\PAREN{ \PTRANS[S](r)<\Gamma> \PAR \INPUT{r}{}.\NIL } }] \\
  \PTRANS[\Gamma] & \vdash \NEW{r:\TRET}\PAREN{ \PTRANS[\ENV{V}]<\Gamma>\HOLE[{ \PTRANS[S](r)<\Gamma> }] \PAR \INPUT{r}{}.\NIL } 
\end{align*}

\noindent where the last step follows, because we know that $\FRESH{r}\PTRANS[\ENV{V}]<\Gamma>$, so the scope can be extruded.
Then we can conclude that 
\begin{equation*}
  \PTRANS[\Gamma] \vdash \PTRANS[\ENV{T}]<\Gamma> \PAR \PTRANS[\ENV{S}] \PAR \NEW{r : \TRET}\PAREN{ \PTRANS[\ENV{V}]<\Gamma>\HOLE[{\PTRANS[S](r)<\Gamma>}] \PAR \INPUT{r}{}.\NIL }
\end{equation*}

\noindent by \nameref{epi_simple_type_res} and \nameref{epi_simple_type_par}, as desired.

For the other direction, we must show that $\PTRANS[\Gamma] \vdash \PTRANS[\CONF{S, \ENV{TSV}}]<\Gamma>$ implies that $\Gamma \vdash \ENV{T}$ and $\Gamma \vdash \ENV{S}$ and $\Gamma \vdash \ENV{V}$ and $\Gamma \vdash S$.
First, we unfold the translation of the initial configuration, which yields:
\begin{equation*}
  \PTRANS[\CONF{S, \ENV{TSV}}]<\Gamma> = \PTRANS[\ENV{T}]<\Gamma> \PAR \PTRANS[\ENV{S}] \PAR \NEW{r : \TRET}\PAREN{ \PTRANS[\ENV{V}]<\Gamma>\HOLE[{\PTRANS[S](r)<\Gamma>}] \PAR \INPUT{r}{}.\NIL }
\end{equation*}

\noindent and thus we know that 
\begin{equation*}
  \PTRANS[\Gamma] \vdash \PTRANS[\ENV{T}]<\Gamma> \PAR \PTRANS[\ENV{S}] \PAR \NEW{r : \TRET}\PAREN{ \PTRANS[\ENV{V}]<\Gamma>\HOLE[{\PTRANS[S](r)<\Gamma>}] \PAR \INPUT{r}{}.\NIL }
\end{equation*}

\noindent which must have been concluded by the rules \nameref{epi_simple_type_in}, \nameref{epi_simple_type_par} and \nameref{epi_simple_type_res}.
From their premises, we get  
\begin{align*}
  \PTRANS[\Gamma] \vdash \PTRANS[\ENV{T}]<\Gamma>  
  \qquad
  \PTRANS[\Gamma] \vdash \PTRANS[\ENV{S}] 
  \qquad
  \PTRANS[\Gamma] \vdash \NEW{r:\TRET}\PAREN{ \PTRANS[\ENV{V}]<\Gamma>\HOLE[{ \PTRANS[S](r)<\Gamma> }] \PAR \INPUT{r}{}.\NIL }
\end{align*}
As we know that the name $r$ does not occur in $\PTRANS[\ENV{V}]<\Gamma>$, since the name is introduced by the translation, we can also conclude that $\PTRANS[\Gamma] \vdash \PTRANS[\ENV{V}]<\Gamma>\HOLE[{ \NEW{r:\TRET}\PAREN{ \PTRANS[S](r)<\Gamma>  \PAR \INPUT{r}{}.\NIL } }]$ by intruding the scope.

By the encoding of $\ENV{V}$, we have that it is translated as a process context consisting of a list of scopes and bound outputs
\begin{equation*}
  \NEW{x_1 : \TIFC{\Gamma(x_1)}, \ldots, x_n : \TIFC{\Gamma(x_n)}}\PAREN{ \OUTPUT{x_1}{v_1} \PAR \ldots \PAR \OUTPUT{x_n}{v_n} \PAR \HOLE }
\end{equation*}

\noindent so $\PTRANS[\Gamma] \vdash \PTRANS[\ENV{V}]<\Gamma>\HOLE[{ \NEW{r:\TRET}\PAREN{ \PTRANS[S](r)<\Gamma>  \PAR \INPUT{r}{}.\NIL } }]$ is concluded by \nameref{epi_simple_type_res} and \nameref{epi_simple_type_par}, and from the premise of the latter, we get that 
\begin{equation*}
  \PTRANS[\Gamma], x_1 : \TIFC{\Gamma(x_1)}, \ldots, x_n : \TIFC{\Gamma(x_n)} \vdash \NEW{r:\TRET}\PAREN{ \PTRANS[S](r)<\Gamma>  \PAR \INPUT{r}{}.\NIL } 
\end{equation*}

However, by Lemma~\ref{wc:lemma:type_correspondence_expr}applied to $\PTRANS[\Gamma], x:\TIFC{\Gamma(x)} \vdash x$, we also know that $\Gamma, x:\Gamma(x) \vdash x:\Gamma(x)$, hence $\Gamma$ must already contain a type entry $B$ for $x$, which is used in the encoding of $\ENV{V}$ given above.
Then by the encoding of $\Gamma$
\begin{equation*}
  \PTRANS[x:B, \Gamma'] = x:\TIFC{B}, \TIFC{B}:(\TCHAN{B}, \EMPTYSET), \PTRANS[\Gamma']
\end{equation*}

\noindent so the type entries $x:\TIFC{B}$ are already in $\PTRANS[\Gamma]$.
Thus, we can simplify the statement above as 
\begin{equation*}
  \PTRANS[\Gamma] \vdash \NEW{r:\TRET}\PAREN{ \PTRANS[S](r)<\Gamma>  \PAR \INPUT{r}{}.\NIL } 
\end{equation*}

\noindent which must be inferred by \nameref{epi_simple_type_res}, and from its premise, we have that 
\begin{equation*}
  \PTRANS[\Gamma], r:\TRET \vdash \PTRANS[S](r)<\Gamma>  \PAR \INPUT{r}{}.\NIL 
\end{equation*}

Now, let $P \DEFSYM \NEW{r:\TRET}\PAREN{ \PTRANS[S](r)<\Gamma>  \PAR \INPUT{r}{}.\NIL }$ and note that we know that $\PTRANS[\Gamma] \vdash P$.
Finally, we have that 
\begin{align*}
  \PTRANS[\Gamma] \vdash \PTRANS[\ENV{T}]<\Gamma>         & \implies \Gamma \vdash \ENV{T} & \text{by Lemma~\ref{wc:lemma:type_correspondence_envt}} \\
  \PTRANS[\Gamma] \vdash \PTRANS[\ENV{S}]                 & \implies \Gamma \vdash \ENV{S} & \text{by Lemma~\ref{wc:lemma:type_correspondence_envs}} \\
  \PTRANS[\Gamma], r:\TRET \vdash \PTRANS[S](r)<\Gamma>   & \implies \Gamma \vdash S       & \text{by Lemma~\ref{wc:lemma:type_correspondence_stm}}  \\
  \PTRANS[\Gamma] \vdash \PTRANS[\ENV{V}]<\Gamma>\HOLE[P] & \implies \Gamma \vdash \ENV{V} & \text{by Lemma~\ref{wc:lemma:type_correspondence_envv}}
\end{align*}

\noindent which allows us to conclude that $\Gamma \vdash \ENV{T}$, $\Gamma \vdash \ENV{S}$, $\Gamma \vdash \ENV{V}$ and $\Gamma \vdash S$, as desired.
\end{proof}

\subsection{On the quality of the encoding}
Theorem~\ref{wc:thm:type_correspondence} above only speaks of the correspondence between the two type systems, not of the correctness of the encoding itself.
However, once we define an encoding, it is natural to also inquire about its quality, since there is by now a well-established literature on the topic (see~\cite{G12,G18,GORLA,gorla_nestmann2014_full_abstraction,Par08,PG15}, just to mention the most methodological papers).
We conjecture that the encoding presented in this paper satisfies most of the properties a `good encoding' should have, in particular operational correspondence and divergence sensitiveness. 
However, proving these properties for an encoding between languages  that rely on different styles of semantics (small-step for \EPI{} vs.\@ big-step for \WCLANG) is a non-trivial task.
For example, the definition of when a statement diverges under a big-step operational semantics is in general not straightforward (see, e.g., \cite{Char13,LG09,PM14} on this topic). 
We realise that for \WCLANG{}, divergence coincides with not having any (big-step) evolution; however, even with this simplifying feature, proving divergence reflection turned out to be technically challenging. 
Since the quality of our encoding is an orthogonal issue to the main aim of this paper (that is the study of typing composite names), we prefer to  leave this investigation for future work. 
One possibility is to keep the big-step semantics for \WCLANG{} and follow the path put forward in~\cite{DBZD20}; alternatively, we can abandon the (more straightforward) big-step semantics in favour of a small-step one, which would simplify the proofs for the quality of the encoding at the price of complicating those for subject reduction in \WCLANG.

\section{Conclusion and future works}\label{wc:sec:conclusion}
In the present paper we have explored the question of how to create a type system with subtyping for a language with composite channels, when such channels can be composed at runtime.
The core of our type system is equivalent to the nominal type system given by Carbone in \cite[Chapt. 6.5]{carbone2005phd}, but the structure of our composite types is different, which yields a more intuitive correspondence with the concept of interfaces known from OO languages, as is made clear through our encoding of a small, class-based language.
In that sense, our work is in line with the remark by Milner in \cite[p.\@ 120]{milner1992functions}, where he notes the connexion between the \PI-calculus and object-oriented programming: both are characterised by transmitting \emph{access} to agents (resp.\@ objects), rather than the agents themselves.
The polyadic synchronisation of \EPI{} makes this connexion even stronger, as evidenced by the simplicity of our encoding, since it yields a natural way to represent fields and methods, scoped under a common class name, as synchronisation vectors with a common prefix.
Our `tree-shaped' types then follow naturally from this similarity.

As mentioned in Section~\ref{wc:sec:related_works}, our work relates to that of Hüttel \cite{huttel2011typed, huttel2013resourcespsi, huttel2016sessionpsi}, since these papers did not provide a general solution to the question of how to type composite subjects.
It is therefore also worth emphasising here that our type system can be used to type the encoding of D\PI, even in the cases where the location name or the located subject name might be bound.
This can be done by creating types for the locations $l$ of the form
\begin{equation*}
  I_l \mapsto (\TNIL, (I_{1} \mapsto (\TCHAN{B_1}, \EMPTYSET), \ldots, I_{n} \mapsto (\TCHAN{B_n}, \EMPTYSET)))
\end{equation*}

\noindent and with the types $I_{i}$ assigned to the subjects $x_i$ occurring at location $l$.
This closely corresponds to our encoding of the \WCLANG-types, with location names $l$ corresponding to class names $A$.
The \WCLANG-language could in principle also have been encoded directly in D\PI, since the encoding only makes use of the $\PI^2$ fragment of \EPI.
However, \WCLANG{} could easily be extended to allow nested class declarations, as is found in some object-oriented languages, i.e.\@ with declarations of the form 
\begin{equation*}
  DC \DCLSYM \epsilon \ORSYM \code{class $A$ \{ $DC$ $DF$ $DM$ \}} 
\end{equation*}

\noindent which would then necessitate subject vectors of arbitrary length to encode an arbitrary nesting depth.
In the present paper, we have forgone this possibility for the sake of simplicity.
Nevertheless, our correspondence results contribute to further the understanding of how a typed, object-oriented language may be represented in the setting of typed process calculi, whilst still preserving some of the structure afforded by interface types in such languages.

A natural further step in this direction would be to consider subtyping, as was done for the \PI-calculus by Pierce and Sangiorgi \cite{PICAPABLE}, who distinguish between the input and output capabilities of a channel.
Indeed, such a subtyping relation is also assumed (to be provided as a parameter) in the type systems by Hüttel for the \PSI-calculus, but none of the instance examples provide any hints as to how it might be defined.
The type system for \EPI{} in \cite{carbone2005phd} also does not have a subtyping relation, so the question of how it might be defined is entirely open.

Input/output capabilities can of course easily be introduced for the channel capability component $C$ of our types, so the interesting question in this regard is rather what a subtyping relation might look like for the \emph{composition capability} component $\Delta$.
Again, the correspondence with object-oriented languages may provide some intuitions.
For example, it would seem natural to require that the composition capabilities of the subtype $\Delta_1$ should be \emph{at least} the same as the supertype $\Delta_2$; or that each type name $I$ in $\Delta_1$ again should be a subtype of an entry in $\Delta_2$.
This would correspond to the usual understanding from object-oriented languages, where an object $O_1$ is a subtype of another object $O_2$, if $O_1$ contains at least the same public fields and methods as $O_2$, and the type of each such field, resp.\@ method, in $O_1$ again is a subtype of the type of the corresponding field, resp.\@ method, in $O_2$.

\bibliographystyle{eptcs}
\bibliography{literature}

\begin{thebibliography}{10}
\providecommand{\bibitemdeclare}[2]{}
\providecommand{\surnamestart}{}
\providecommand{\surnameend}{}
\providecommand{\urlprefix}{Available at }
\providecommand{\url}[1]{\texttt{#1}}
\providecommand{\href}[2]{\texttt{#2}}
\providecommand{\urlalt}[2]{\href{#1}{#2}}
\providecommand{\doi}[1]{doi:\urlalt{https://doi.org/#1}{#1}}
\providecommand{\eprint}[1]{arXiv:\urlalt{https://arxiv.org/abs/#1}{#1}}
\providecommand{\bibinfo}[2]{#2}

\bibitemdeclare{inproceedings}{bengtson2009psi}
\bibitem{bengtson2009psi}
\bibinfo{author}{Jesper \surnamestart Bengtson\surnameend},
  \bibinfo{author}{Magnus \surnamestart Johansson\surnameend},
  \bibinfo{author}{Joachim \surnamestart Parrow\surnameend} \&
  \bibinfo{author}{Bj{\"o}rn \surnamestart Victor\surnameend}
  (\bibinfo{year}{2009}): \emph{\bibinfo{title}{Psi-calculi: Mobile processes,
  nominal data, and logic}}.
\newblock In: {\slshape \bibinfo{booktitle}{2009 24th Annual IEEE Symposium on
  Logic In Computer Science}}, \bibinfo{organization}{IEEE}, pp.
  \bibinfo{pages}{39--48}, \doi{10.1016/S1571-0661(05)80361-5}.

\bibitemdeclare{article}{bengtson2011psi}
\bibitem{bengtson2011psi}
\bibinfo{author}{Jesper \surnamestart Bengtson\surnameend},
  \bibinfo{author}{Magnus \surnamestart Johansson\surnameend},
  \bibinfo{author}{Joachim \surnamestart Parrow\surnameend} \&
  \bibinfo{author}{Björn \surnamestart Victor\surnameend}
  (\bibinfo{year}{2011}): \emph{\bibinfo{title}{{Psi-calculi: a framework for
  mobile processes with nominal data and logic}}}.
\newblock {\slshape \bibinfo{journal}{{Logical Methods in Computer Science}}}
  \bibinfo{volume}{{Volume 7, Issue 1}}, \doi{10.2168/LMCS-7(1:11)2011}.
\newblock \urlprefix\url{https://lmcs.episciences.org/696}.

\bibitemdeclare{phdthesis}{carbone2005phd}
\bibitem{carbone2005phd}
\bibinfo{author}{Marco \surnamestart Carbone\surnameend}
  (\bibinfo{year}{2005}): \emph{\bibinfo{title}{Trust and Mobility}}.
\newblock Ph.D. thesis, \bibinfo{school}{University of Aarhus}.
\newblock \bibinfo{note}{BRICS Dissertation Series Number DS-05-3}.

\bibitemdeclare{article}{CARBONEMAFFEIS}
\bibitem{CARBONEMAFFEIS}
\bibinfo{author}{Marco \surnamestart Carbone\surnameend} \&
  \bibinfo{author}{Sergio \surnamestart Maffeis\surnameend}
  (\bibinfo{year}{2003}): \emph{\bibinfo{title}{On the Expressive Power of
  Polyadic Synchronisation in Pi-Calculus}}.
\newblock {\slshape \bibinfo{journal}{Nordic Journal of Computing}}
  \bibinfo{volume}{10}(\bibinfo{number}{2}), pp. \bibinfo{pages}{70--98},
  \doi{10.1016/S1571-0661(05)80361-5}.

\bibitemdeclare{inproceedings}{Char13}
\bibitem{Char13}
\bibinfo{author}{Arthur \surnamestart Chargu{\'{e}}raud\surnameend}
  (\bibinfo{year}{2013}): \emph{\bibinfo{title}{Pretty-Big-Step Semantics}}.
\newblock In: {\slshape \bibinfo{booktitle}{Proc. of {ESOP}}}, {\slshape
  \bibinfo{series}{LNCS}} \bibinfo{volume}{7792},
  \bibinfo{publisher}{Springer}, pp. \bibinfo{pages}{41--60},
  \doi{10.1007/978-3-642-37036-6\_3}.

\bibitemdeclare{inproceedings}{DBZD20}
\bibitem{DBZD20}
\bibinfo{author}{Francesco \surnamestart Dagnino\surnameend},
  \bibinfo{author}{Viviana \surnamestart Bono\surnameend},
  \bibinfo{author}{Elena \surnamestart Zucca\surnameend} \&
  \bibinfo{author}{Mariangiola \surnamestart Dezani{-}Ciancaglini\surnameend}
  (\bibinfo{year}{2020}): \emph{\bibinfo{title}{Soundness Conditions for
  Big-Step Semantics}}.
\newblock In: {\slshape \bibinfo{booktitle}{Proc. of {ESOP}}}, {\slshape
  \bibinfo{series}{LNCS}} \bibinfo{volume}{12075},
  \bibinfo{publisher}{Springer}, pp. \bibinfo{pages}{169--196},
  \doi{10.1007/978-3-030-44914-8\_7}.

\bibitemdeclare{article}{gabbay2002nominal}
\bibitem{gabbay2002nominal}
\bibinfo{author}{Murdoch \surnamestart Gabbay\surnameend} \&
  \bibinfo{author}{Andrew \surnamestart Pitts\surnameend}
  (\bibinfo{year}{2002}): \emph{\bibinfo{title}{A New Approach to Abstract
  Syntax with Variable Binding}}.
\newblock {\slshape \bibinfo{journal}{Formal Asp. Comput.}}
  \bibinfo{volume}{13}, pp. \bibinfo{pages}{341--363},
  \doi{10.1007/s001650200016}.

\bibitemdeclare{inproceedings}{G12}
\bibitem{G12}
\bibinfo{author}{Rob \surnamestart van Glabbeek\surnameend}
  (\bibinfo{year}{2012}): \emph{\bibinfo{title}{Musings on Encodings and
  Expressiveness}}.
\newblock In: {\slshape \bibinfo{booktitle}{Proc. of {EXPRESS/SOS}}}, {\slshape
  \bibinfo{series}{{EPTCS}}}~\bibinfo{volume}{89}, pp. \bibinfo{pages}{81--98},
  \doi{10.4204/EPTCS.89.7}.

\bibitemdeclare{inproceedings}{G18}
\bibitem{G18}
\bibinfo{author}{Rob \surnamestart van Glabbeek\surnameend}
  (\bibinfo{year}{2018}): \emph{\bibinfo{title}{A Theory of Encodings and
  Expressiveness (Extended Abstract)}}.
\newblock In: {\slshape \bibinfo{booktitle}{Proc. of {FoSSaCS}}}, {\slshape
  \bibinfo{series}{LNCS}} \bibinfo{volume}{10803},
  \bibinfo{publisher}{Springer}, pp. \bibinfo{pages}{183--202},
  \doi{10.1007/978-3-319-89366-2\_10}.

\bibitemdeclare{article}{GORLA}
\bibitem{GORLA}
\bibinfo{author}{Daniele \surnamestart Gorla\surnameend}
  (\bibinfo{year}{2010}): \emph{\bibinfo{title}{Towards a unified approach to
  encodability and separation results for process calculi}}.
\newblock {\slshape \bibinfo{journal}{Information and Computation}}
  \bibinfo{volume}{208}(\bibinfo{number}{9}), pp. \bibinfo{pages}{1031--1053},
  \doi{10.1016/j.ic.2010.05.002}.

\bibitemdeclare{article}{gorla_nestmann2014_full_abstraction}
\bibitem{gorla_nestmann2014_full_abstraction}
\bibinfo{author}{Daniele \surnamestart Gorla\surnameend} \&
  \bibinfo{author}{Uwe \surnamestart Nestmann\surnameend}
  (\bibinfo{year}{2016}): \emph{\bibinfo{title}{Full abstraction for
  expressiveness: history, myths and facts}}.
\newblock {\slshape \bibinfo{journal}{Mathematical Structures in Computer
  Science}} \bibinfo{volume}{26}, pp. \bibinfo{pages}{639 -- 654},
  \doi{10.1017/S0960129514000279}.

\bibitemdeclare{article}{DPICALC}
\bibitem{DPICALC}
\bibinfo{author}{Matthew \surnamestart Hennessy\surnameend} \&
  \bibinfo{author}{James \surnamestart Riely\surnameend}
  (\bibinfo{year}{2002}): \emph{\bibinfo{title}{Resource Access Control in
  Systems of Mobile Agents}}.
\newblock {\slshape \bibinfo{journal}{Information and Computation}}
  \bibinfo{volume}{173}(\bibinfo{number}{1}), pp. \bibinfo{pages}{82--120},
  \doi{10.1006/inco.2001.3089}.

\bibitemdeclare{inproceedings}{hirschkoff2020references_picalc}
\bibitem{hirschkoff2020references_picalc}
\bibinfo{author}{Daniel \surnamestart Hirschkoff\surnameend},
  \bibinfo{author}{Enguerrand \surnamestart Prebet\surnameend} \&
  \bibinfo{author}{Davide \surnamestart Sangiorgi\surnameend}
  (\bibinfo{year}{2020}): \emph{\bibinfo{title}{On the Representation of
  References in the Pi-Calculus}}.
\newblock In \bibinfo{editor}{Igor \surnamestart Konnov\surnameend} \&
  \bibinfo{editor}{Laura \surnamestart Kov{\'{a}}cs\surnameend}, editors:
  {\slshape \bibinfo{booktitle}{31st International Conference on Concurrency
  Theory, {CONCUR} 2020, September 1-4, 2020, Vienna, Austria (Virtual
  Conference)}}, {\slshape \bibinfo{series}{LIPIcs}} \bibinfo{volume}{171},
  \bibinfo{publisher}{Schloss Dagstuhl - Leibniz-Zentrum f{\"{u}}r Informatik},
  pp. \bibinfo{pages}{34:1--34:20}, \doi{10.4230/LIPICS.CONCUR.2020.34}.

\bibitemdeclare{book}{huttel2010transitions}
\bibitem{huttel2010transitions}
\bibinfo{author}{Hans \surnamestart H{\"{u}}ttel\surnameend}
  (\bibinfo{year}{2010}): \emph{\bibinfo{title}{Transitions and Trees - An
  Introduction to Structural Operational Semantics}}.
\newblock \bibinfo{publisher}{Cambridge University Press},
  \doi{10.1017/CBO9780511840449}.

\bibitemdeclare{inproceedings}{huttel2011typed}
\bibitem{huttel2011typed}
\bibinfo{author}{Hans \surnamestart Hüttel\surnameend} (\bibinfo{year}{2011}):
  \emph{\bibinfo{title}{Typed $\psi$-calculi}}.
\newblock In: {\slshape \bibinfo{booktitle}{International Conference on
  Concurrency Theory}}, \bibinfo{organization}{Springer}, pp.
  \bibinfo{pages}{265--279}, \doi{10.1007/978-3-642-23217-6\_18}.

\bibitemdeclare{inproceedings}{huttel2013resourcespsi}
\bibitem{huttel2013resourcespsi}
\bibinfo{author}{Hans \surnamestart Hüttel\surnameend} (\bibinfo{year}{2014}):
  \emph{\bibinfo{title}{Types for Resources in $\psi$-calculi}}.
\newblock In \bibinfo{editor}{Martín \surnamestart Abadi\surnameend} \&
  \bibinfo{editor}{Alberto \surnamestart Lluch~Lafuente\surnameend}, editors:
  {\slshape \bibinfo{booktitle}{Trustworthy Global Computing}},
  \bibinfo{publisher}{Springer International Publishing},
  \bibinfo{address}{Cham}, pp. \bibinfo{pages}{83--102},
  \doi{10.1007/978-3-319-05119-2\_6}.

\bibitemdeclare{inproceedings}{huttel2016sessionpsi}
\bibitem{huttel2016sessionpsi}
\bibinfo{author}{Hans \surnamestart Hüttel\surnameend} (\bibinfo{year}{2016}):
  \emph{\bibinfo{title}{Binary Session Types for Psi-Calculi}}.
\newblock In \bibinfo{editor}{Atsushi \surnamestart Igarashi\surnameend},
  editor: {\slshape \bibinfo{booktitle}{Programming Languages and Systems}},
  \bibinfo{publisher}{Springer International Publishing},
  \bibinfo{address}{Cham}, pp. \bibinfo{pages}{96--115},
  \doi{10.1007/978-3-319-47958-3\_6}.

\bibitemdeclare{article}{huttel/2024/iandc/hopsitypes}
\bibitem{huttel/2024/iandc/hopsitypes}
\bibinfo{author}{Hans \surnamestart Hüttel\surnameend}, \bibinfo{author}{Stian
  \surnamestart Lybech\surnameend}, \bibinfo{author}{Alex~R. \surnamestart
  Bendixen\surnameend} \& \bibinfo{author}{Bjarke~B. \surnamestart
  Bojesen\surnameend} (\bibinfo{year}{2024}): \emph{\bibinfo{title}{A Generic
  Type System for Higher-Order $\Psi$-calculi}}.
\newblock {\slshape \bibinfo{journal}{Information and Computation}}, p.
  \bibinfo{pages}{105190}, \doi{https://doi.org/10.1016/j.ic.2024.105190}.
\newblock
  \urlprefix\url{https://www.sciencedirect.com/science/article/pii/S0890540124000555}.

\bibitemdeclare{article}{KS02}
\bibitem{KS02}
\bibinfo{author}{Josva \surnamestart Kleist\surnameend} \&
  \bibinfo{author}{Davide \surnamestart Sangiorgi\surnameend}
  (\bibinfo{year}{2002}): \emph{\bibinfo{title}{Imperative objects as mobile
  processes}}.
\newblock {\slshape \bibinfo{journal}{Sci. Comput. Program.}}
  \bibinfo{volume}{44}(\bibinfo{number}{3}), pp. \bibinfo{pages}{293--342},
  \doi{10.1016/S0167-6423(02)00034-5}.

\bibitemdeclare{article}{LG09}
\bibitem{LG09}
\bibinfo{author}{Xavier \surnamestart Leroy\surnameend} \&
  \bibinfo{author}{Herv{\'{e}} \surnamestart Grall\surnameend}
  (\bibinfo{year}{2009}): \emph{\bibinfo{title}{Coinductive big-step
  operational semantics}}.
\newblock {\slshape \bibinfo{journal}{Inf. Comput.}}
  \bibinfo{volume}{207}(\bibinfo{number}{2}), pp. \bibinfo{pages}{284--304},
  \doi{10.1016/J.IC.2007.12.004}.

\bibitemdeclare{article}{lybech/2024/iandc/rhocalc}
\bibitem{lybech/2024/iandc/rhocalc}
\bibinfo{author}{Stian \surnamestart Lybech\surnameend} (\bibinfo{year}{2024}):
  \emph{\bibinfo{title}{The reflective higher-order calculus: Encodability,
  typability and separation}}.
\newblock {\slshape \bibinfo{journal}{Information and Computation}}
  \bibinfo{volume}{297}, p. \bibinfo{pages}{105138},
  \doi{10.1016/j.ic.2024.105138}.

\bibitemdeclare{article}{milner1992functions}
\bibitem{milner1992functions}
\bibinfo{author}{Robin \surnamestart Milner\surnameend} (\bibinfo{year}{1992}):
  \emph{\bibinfo{title}{Functions as processes}}.
\newblock {\slshape \bibinfo{journal}{Mathematical structures in computer
  science}} \bibinfo{volume}{2}(\bibinfo{number}{2}), pp.
  \bibinfo{pages}{119--141}.

\bibitemdeclare{incollection}{PICALC}
\bibitem{PICALC}
\bibinfo{author}{Robin \surnamestart Milner\surnameend} (\bibinfo{year}{1993}):
  \emph{\bibinfo{title}{The Polyadic $\pi$-Calculus: a Tutorial}}.
\newblock In: {\slshape \bibinfo{booktitle}{Logic and Algebra of
  Specification}}, \bibinfo{publisher}{Springer Berlin Heidelberg}, pp.
  \bibinfo{pages}{203--246}, \doi{10.1007/978-3-642-58041-3\_6}.

\bibitemdeclare{article}{milner_walker_parrow1992picalc}
\bibitem{milner_walker_parrow1992picalc}
\bibinfo{author}{Robin \surnamestart Milner\surnameend},
  \bibinfo{author}{Joachim \surnamestart Parrow\surnameend} \&
  \bibinfo{author}{David \surnamestart Walker\surnameend}
  (\bibinfo{year}{1992}): \emph{\bibinfo{title}{A calculus of mobile processes,
  I}}.
\newblock {\slshape \bibinfo{journal}{Information and Computation}}
  \bibinfo{volume}{100}(\bibinfo{number}{1}), pp. \bibinfo{pages}{1--40},
  \doi{10.1016/0890-5401(92)90008-4}.

\bibitemdeclare{inproceedings}{NR99}
\bibitem{NR99}
\bibinfo{author}{Uwe \surnamestart Nestmann\surnameend} \&
  \bibinfo{author}{Ant{\'{o}}nio \surnamestart Ravara\surnameend}
  (\bibinfo{year}{1999}): \emph{\bibinfo{title}{Semantics of Objects as
  Processes {(SOAP)}}}.
\newblock In: {\slshape \bibinfo{booktitle}{ECOOP'99 Workshops}}, {\slshape
  \bibinfo{series}{LNCS}} \bibinfo{volume}{1743},
  \bibinfo{publisher}{Springer}, pp. \bibinfo{pages}{314--325}.

\bibitemdeclare{book}{nielson_nielson2007semantics_with_applications}
\bibitem{nielson_nielson2007semantics_with_applications}
\bibinfo{author}{Hanne~Riis \surnamestart Nielson\surnameend} \&
  \bibinfo{author}{Flemming \surnamestart Nielson\surnameend}
  (\bibinfo{year}{2007}): \emph{\bibinfo{title}{Semantics with Applications: An
  Appetizer}}.
\newblock \bibinfo{publisher}{Springer-Verlag London},
  \doi{10.1007/978-1-84628-692-6}.

\bibitemdeclare{incollection}{parrow2001introduction}
\bibitem{parrow2001introduction}
\bibinfo{author}{Joachim \surnamestart Parrow\surnameend}
  (\bibinfo{year}{2001}): \emph{\bibinfo{title}{An introduction to the
  $\pi$-calculus}}.
\newblock In: {\slshape \bibinfo{booktitle}{Handbook of Process Algebra}},
  \bibinfo{publisher}{Elsevier}, pp. \bibinfo{pages}{479--543},
  \doi{10.1016/B978-044482830-9/50026-6}.

\bibitemdeclare{inproceedings}{Par08}
\bibitem{Par08}
\bibinfo{author}{Joachim \surnamestart Parrow\surnameend}
  (\bibinfo{year}{2006}): \emph{\bibinfo{title}{Expressiveness of Process
  Algebras}}.
\newblock In: {\slshape \bibinfo{booktitle}{Emerging Trends in Concurrency
  Theory}}, {\slshape \bibinfo{series}{ENTCS}} \bibinfo{volume}{209},
  \bibinfo{publisher}{Elsevier}, pp. \bibinfo{pages}{173--186},
  \doi{10.1016/J.ENTCS.2008.04.011}.

\bibitemdeclare{article}{parrow2014higher}
\bibitem{parrow2014higher}
\bibinfo{author}{Joachim \surnamestart Parrow\surnameend},
  \bibinfo{author}{Johannes \surnamestart Borgstr{\"o}m\surnameend},
  \bibinfo{author}{Palle \surnamestart Raabjerg\surnameend} \&
  \bibinfo{author}{Johannes \surnamestart {\AA}man~Pohjola\surnameend}
  (\bibinfo{year}{2014}): \emph{\bibinfo{title}{Higher-order psi-calculi}}.
\newblock {\slshape \bibinfo{journal}{Mathematical Structures in Computer
  Science}} \bibinfo{volume}{24}(\bibinfo{number}{2}),
  \doi{10.1017/S0960129513000170}.

\bibitemdeclare{inproceedings}{PG15}
\bibitem{PG15}
\bibinfo{author}{Kirstin \surnamestart Peters\surnameend} \&
  \bibinfo{author}{Rob~J. \surnamestart van Glabbeek\surnameend}
  (\bibinfo{year}{2015}): \emph{\bibinfo{title}{Analysing and Comparing
  Encodability Criteria}}.
\newblock In: {\slshape \bibinfo{booktitle}{Proc. of {EXPRESS/SOS}}}, {\slshape
  \bibinfo{series}{{EPTCS}}} \bibinfo{volume}{190}, pp.
  \bibinfo{pages}{46--60}, \doi{10.4204/EPTCS.190.4}.

\bibitemdeclare{inproceedings}{PICAPABLE}
\bibitem{PICAPABLE}
\bibinfo{author}{Benjamin \surnamestart Pierce\surnameend} \&
  \bibinfo{author}{Davide \surnamestart Sangiorgi\surnameend}
  (\bibinfo{year}{1993}): \emph{\bibinfo{title}{Typing and subtyping for mobile
  processes}}.
\newblock In: {\slshape \bibinfo{booktitle}{[1993] Proceedings Eighth Annual
  IEEE Symposium on Logic in Computer Science}}, \bibinfo{organization}{IEEE},
  pp. \bibinfo{pages}{376--385}, \doi{10.1109/LICS.1993.287570}.

\bibitemdeclare{inproceedings}{PICT}
\bibitem{PICT}
\bibinfo{author}{Benjamin~C. \surnamestart Pierce\surnameend} \&
  \bibinfo{author}{David~N. \surnamestart Turner\surnameend}
  (\bibinfo{year}{2000}): \emph{\bibinfo{title}{Pict: a programming language
  based on the Pi-Calculus}}.
\newblock In \bibinfo{editor}{Gordon~D. \surnamestart Plotkin\surnameend},
  \bibinfo{editor}{Colin \surnamestart Stirling\surnameend} \&
  \bibinfo{editor}{Mads \surnamestart Tofte\surnameend}, editors: {\slshape
  \bibinfo{booktitle}{Proof, Language, and Interaction, Essays in Honour of
  Robin Milner}}, \bibinfo{publisher}{The {MIT} Press}, pp.
  \bibinfo{pages}{455--494}.

\bibitemdeclare{inproceedings}{PM14}
\bibitem{PM14}
\bibinfo{author}{Casper~Bach \surnamestart Poulsen\surnameend} \&
  \bibinfo{author}{Peter~D. \surnamestart Mosses\surnameend}
  (\bibinfo{year}{2014}): \emph{\bibinfo{title}{Deriving Pretty-Big-Step
  Semantics from Small-Step Semantics}}.
\newblock In: {\slshape \bibinfo{booktitle}{Proc. of {ESOP}}}, {\slshape
  \bibinfo{series}{LNCS}} \bibinfo{volume}{8410},
  \bibinfo{publisher}{Springer}, pp. \bibinfo{pages}{270--289},
  \doi{10.1007/978-3-642-54833-8\_15}.

\bibitemdeclare{techreport}{RV96}
\bibitem{RV96}
\bibinfo{author}{Ant\'onio \surnamestart Ravara\surnameend} \&
  \bibinfo{author}{Vasco~Thudichum \surnamestart Vasconcelos\surnameend}
  (\bibinfo{year}{1996}): \emph{\bibinfo{title}{An Operational Semantics and a
  Type System for {GNOME} Based on a Typed Calculus of Objects}}.
\newblock \bibinfo{type}{Technical Report}, \bibinfo{institution}{17-96, Dep.
  of Mathematics, Tech. Univ. of Lisbon}.

\bibitemdeclare{article}{Sangiorgi98}
\bibitem{Sangiorgi98}
\bibinfo{author}{Davide \surnamestart Sangiorgi\surnameend}
  (\bibinfo{year}{1998}): \emph{\bibinfo{title}{An Interpretation of Typed
  Objects into Typed pi-Calculus}}.
\newblock {\slshape \bibinfo{journal}{Inf. Comput.}}
  \bibinfo{volume}{143}(\bibinfo{number}{1}), pp. \bibinfo{pages}{34--73},
  \doi{10.1006/INCO.1998.2711}.

\bibitemdeclare{book}{sangiorgi2003pi}
\bibitem{sangiorgi2003pi}
\bibinfo{author}{Davide \surnamestart Sangiorgi\surnameend} \&
  \bibinfo{author}{David \surnamestart Walker\surnameend}
  (\bibinfo{year}{2003}): \emph{\bibinfo{title}{The pi-calculus: a Theory of
  Mobile Processes}}.
\newblock \bibinfo{publisher}{Cambridge university press}.

\bibitemdeclare{phdthesis}{turner1996phd}
\bibitem{turner1996phd}
\bibinfo{author}{David~N. \surnamestart Turner\surnameend}
  (\bibinfo{year}{1996}): \emph{\bibinfo{title}{The Polymorphic Pi-calculus:
  Theory and Implementation}}.
\newblock Ph.D. thesis, \bibinfo{school}{University of Edinburgh, {UK}}.
\newblock \urlprefix\url{https://hdl.handle.net/1842/395}.

\bibitemdeclare{inproceedings}{Vasco95}
\bibitem{Vasco95}
\bibinfo{author}{Vasco~Thudichum \surnamestart Vasconcelos\surnameend}
  (\bibinfo{year}{1995}): \emph{\bibinfo{title}{An Operational Semantics and a
  Type System for {ABCL/1} Based on a Calculus of Objects}}.
\newblock In: {\slshape \bibinfo{booktitle}{Object-Oriented Computing III,
  Lecture Notes}}, \bibinfo{publisher}{Kindai Kagaku Sha}.

\bibitemdeclare{article}{Walker95}
\bibitem{Walker95}
\bibinfo{author}{David \surnamestart Walker\surnameend} (\bibinfo{year}{1995}):
  \emph{\bibinfo{title}{Objects in the $\pi$-Calculus}}.
\newblock {\slshape \bibinfo{journal}{Information and Computation}}
  \bibinfo{volume}{116}(\bibinfo{number}{2}), pp. \bibinfo{pages}{253--271},
  \doi{10.1006/inco.1995.1018}.

\end{thebibliography}

\end{document}